%% file: main.tex
\documentclass[sigconf]{acmart}
\AtBeginDocument{%
  \providecommand\BibTeX{{%
    \normalfont B\kern-0.5em{\scshape i\kern-0.25em b}\kern-0.8em\TeX}}}

\setcopyright{acmcopyright}
\copyrightyear{2018}
\acmYear{2018}
\acmDOI{XXXXXXX.XXXXXXX}

\acmConference[Conference acronym 'XX]{Make sure to enter the correct
  conference title from your rights confirmation emai}{June 03--05,
  2018}{Woodstock, NY}
\acmPrice{15.00}
\acmISBN{978-1-4503-XXXX-X/18/06}
\usepackage{xspace}
\newcommand{\sysName}{{\sf Veer}\xspace}
\newcommand{\sysOpt}{{\sf Veer}\textsuperscript{+}\xspace}
\definecolor{blush}{rgb}{0.87, 0.36, 0.51}
\definecolor{myorange}{RGB}{255, 181, 112}
\newcommand{\sadeem}[1]{{\color{black} #1}}

\newcommand{\boldstart}{\vspace{0.2em}\noindent\textbf}
\newtheorem*{problem}{Problem Statement} 
\newtheorem*{assume}{Assumption}
\newtheorem{example}{Example} 

\usepackage{subcaption}
\usepackage{wrapfig}
\usepackage{tabularx}
\usepackage{ragged2e}
\usepackage{adjustbox}
\usepackage{longtable}
\usepackage{multicol}
\usepackage{multirow}
\usepackage[linesnumbered,ruled]{algorithm2e}
\SetKw{NOT}{not}
\SetKw{OR}{or}
\SetKw{goto}{go to}
\SetKw{Continue}{continue}

\definecolor{myblue}{RGB}{236,243,253}
\definecolor{mydarkblue}{RGB}{191,206,228}
\SetCommentSty{mycommfont}
\usepackage{color,soul}
\usepackage{tikz}
\usetikzlibrary{shapes}

\SetKw{BREAK}{break}
\SetKw{AND}{and}
\SetKw{CONTINUE}{continue}

\DeclareRobustCommand\mytikzdot{\tikz \fill[myorange] (1ex,1ex) circle (0.8ex);}

\DeclareRobustCommand\mytikzbox{\tikz \filldraw[fill=myblue, draw=mydarkblue,thick,dashed] (1ex,1ex) rectangle (3.4ex,2.4ex);}

\begin{document}
\title{\sysName: Verifying Equivalence of  Dataflow Versions in Iterative Data Analytics (Extended Version)}

\author{Sadeem Alsudais, Avinash Kumar, and Chen Li}
\affiliation{%
  \institution{Department of Computer Science, UC Irvine, CA 92697, USA}
}
\email{{salsudai, avinask1, chenli}@ics.uci.edu}

\renewcommand{\shortauthors}{Alsudais and Kumar et al.}
\renewcommand{\shorttitle}{\sysName: Verifying Equivalence of Workflow Versions}

\input{abstract.tex}


\maketitle

\input{sec1-introduction}
\input{sec2-problem-formulation}
\input{sec3-overview}

\input{sec4-single-change}
\input{sec5-multiple-changes.tex}
\input{sec6-completeness}
\input{sec7-optimizations.tex}

\input{sec8-extensions.tex}
\input{sec9-experiments}
\input{sec1-related-works}

\input{sec10-conclusion}

\begin{acks}
This work is supported by a graduate fellowship from King Saud
University and was supported by NSF award III 2107150.
\end{acks}


\bibliographystyle{ACM-Reference-Format}
\bibliography{localrefs,references}

\input{appendix}

\end{document}

%% file: abstract.tex
\begin{abstract}
Data analytics using GUI-based workflows is an iterative process in which an analyst makes many iterations of changes to refine the dataflow, generating a different version at each iteration. In many cases, the result of executing a dataflow version is equivalent to a result of a prior executed version. Identifying such equivalence between the execution results of different dataflow versions is important for optimizing the performance of a dataflow by reusing results from a previous run.  The size of the dataflows and the complexity of their operators (e.g., UDF and ML models) often make existing equivalence verifiers (EVs) unable to solve the problem.  In this paper, we present ``\sysName,'' which leverages the fact that two dataflow versions can be very similar except for a few changes. The solution divides the dataflow version pair into small parts, called {\em windows}, and verifies the equivalence within each window by using an existing EV as a black box.  We develop solutions to efficiently generate windows and verify the equivalence within each window.  Our thorough experiments on real dataflows show that  \sysName is able to not only verify the equivalence of dataflows that cannot be supported by existing EVs but also do the verification efficiently.
\end{abstract}

%% file: sec1-introduction.tex
\section{Introduction} 
\label{sec:intro}

Big data processing platforms, especially GUI-based systems, enable users to quickly construct complex analytical dataflows~\cite{alteryx,journals/pvldb/KumarWNL20,databricks-tracking}.
These dataflows are refined in iterations, generating a new version at each iteration, before a final dataflow is constructed, due to the nature of exploratory and iterative data analytics~\cite{conf/sigmod/DerakhshanMKRM22,conf/sigmod/XuKAR22}.
For example, Figure~\ref{fig:running-example} shows a dataflow for finding the relevant Tweets by the top $k$ non-commercial influencers based on their tweeting rate on a specific topic.
After the analyst constructs the initial dataflow version (a) and executes it, she refines the dataflow to achieve the desired results. This yields the following edit operations highlighted in the figure, 1) deleting the filter `o' operator, 2) adding the filter `g' operator, and 3) adding the filter `h' operator.

\begin{figure}[hbt]
  \centering 
  \begin{subfigure}[t]{\linewidth}
         \centering
         \includegraphics[width=\linewidth]{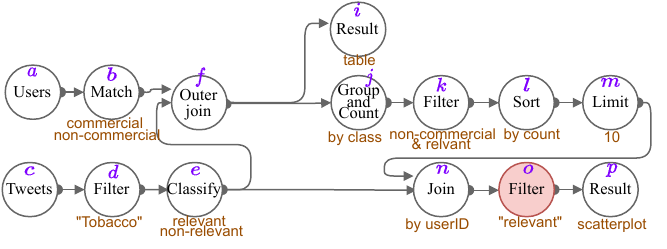}
                \vspace{-0.2in}
         \caption{Version 1: Initial dataflow with sinks $s_i$ of all users' tweets, and $s_p$ of top $k$ non-commercial influencers' relevant tweets. The highlighted operator indicates that it is deleted in a subsequent version.}
         \label{fig:running-example-v1}
     \end{subfigure}
     \begin{subfigure}[t]{\linewidth}
         \centering
\includegraphics[width=\linewidth]{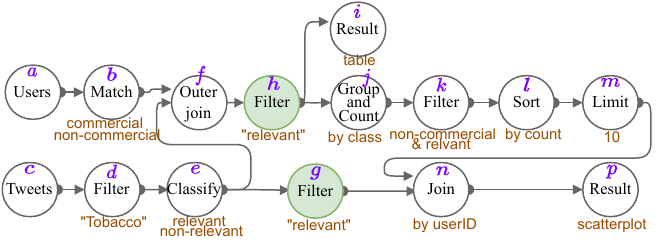}
         \caption{Version 2: Refined version to optimize the dataflow performance and filter on relevant tweets of all users. The highlighted operators are newly added in the new version.}
         \label{fig:running-example-v2}
     \end{subfigure}
       \vspace{-0.1in}
  \caption{Example dataflow and its evolution in two versions.}
  \label{fig:running-example}
  \vspace{-0.15in}
\end{figure}


There has been a growing interest recently in keeping track of these dataflow versions and their execution results~\cite{klaus_greff-proc-scipy-2017,conf/sigmod/VartakSLVHMZ16,databricks-tracking,conf/vldb/Alsudais22,conf/sc/WoodmanHWM11}.
In many applications, these dataflows have a significant amount of overlap and equivalence~\cite{journals/corr/abs-2004-00481, journals/pvldb/JindalKRP18,journals/pvldb/ZhouANHX19,conf/vldb/Alsudais22}. 
For example, $45\%$ of the daily jobs in Microsoft's analytics clusters have overlaps~\cite{journals/pvldb/JindalKRP18}. $27\%$ of $9,486$ dataflows to detect fraud transactions from Ant Financial have overlaps, $6\%$ of which is equivalent~\cite{journals/corr/abs-2004-00481}.
The final results of executing a dataflow are called the dataflow's ``sinks''.
In the running example, the edits applied on version (a) that led to a new version (b) had no effect on the result of the sink labeled `p'.  Identifying such equivalence between the execution results of different dataflow versions is important.  The following are two example use cases.

{\em Use case 1: Optimizing dataflow execution.}
dataflows can take a long time to run due to the size of the data and their computational complexity, especially when they have advanced machine learning operations~\cite{journals/corr/abs-2004-00481,journals/pvldb/KumarWNL20,misc/flink}.
Optimizing the performance of a dataflow execution has been studied extensively in the literature~\cite{conf/icde/PerezJ14,conf/sigmod/DursunBCK17}. 
One optimizing technique is by leveraging the iterative nature of data analytics to reuse previously materialized results~\cite{conf/sigmod/DerakhshanMKRM22,journals/pvldb/JindalKRP18}.  

{\em Use case 2: Reducing storage space.} The execution of a dataflow may produce a large number of results and storing the output of all generated jobs is impractical~\cite{journals/pvldb/ElghandourA12}. Due to the nature of the overlap and equivalence of consecutive versions, one line of work~\cite{journals/pvldb/AhmadKKN12,conf/sigmod/DerakhshanMKRM22} periodically performs a view de-duplication to remove duplicate stored results. Identifying the equivalence between the dataflow versions can be used to avoid storing duplicate results and helps in avoiding periodic clean-up of duplicate results.




These use cases show the need for effective and efficient solutions to decide the equivalence of two dataflow versions. 
We observe the following two unique traits of these GUI-based iterative dataflows.  ({\em T1}) these dataflows can be large and complex, with operators that are semantically rich~\cite{alteryx,conf/sigmod/XuKAR22,conf/sigmod/DerakhshanMKRM22}. 
For example, $8$ of randomly selected Alteryx's dataflows~\cite{alteryx} had an average of $48$ operators, with one of the dataflows containing $102$ operators, and comprised of mostly non-relational operators~\footnote{Details of these dataflows are in Appendix~\ref{sec:appendix-c}}.  Real dataflows in Texera~\cite{texera} had an average size of $23$ operators, and most of them had visualization and UDF operators.
Some operators are user-defined functions (UDF) that implement highly customized logic including machine learning techniques for analyzing data of different modalities such as text, images, audios, and videos~\cite{conf/sigmod/XuKAR22}.  For instance, the dataflows in the running example contain two non-relational operators, namely a {\sf Dictionary Matcher} and a {\sf Classifier}. 
 ({\em T2}) Those adjacent versions of the same dataflow tend to be similar, especially during the phase where the developer is refining the dataflow to do fine tuning~\cite{conf/sigmod/DursunBCK17,conf/sigmod/XuKAR22}.  For example, $50$\% of the dataflows belonging to the benchmarks that simulated real iterative tasks on video~\cite{conf/sigmod/XuKAR22} and TPC-H~\cite{conf/sigmod/DursunBCK17} data had overlap. The refinements between the successive versions comprised of only a few changes over a particular part of the dataflow. Thus, we want to study the following:
 
 %

\vspace{2mm}
\noindent\fbox{
    \parbox{0.9\linewidth}{
        \textbf{Problem Statement}: Given two similar versions of a complex dataflow, verify if they produce the same results.
    }
}
\vspace{2mm}

\boldstart{Limitations of existing solutions.}
dataflows include relational operators and UDFs~\cite{journals/pvldb/KumarWNL20}. Thus, we can view the problem of checking the equivalence of two dataflow versions as the problem of checking the equivalence of two SQL queries. The latter is undecidable in general~\cite{10.5555/551350} (based on the reduction from First-order logic).  There have been many Equivalence Verifiers (EVs) proposed to verify the equivalence of two SQL queries~\cite{journals/pvldb/ChuMRCS18,journals/pvldb/ZhouANHX19,journals/corr/abs-2004-00481}.
These EVs have {\em restrictions} on the type of operators they can support, and mainly  focus on relational operators such as SPJ, aggregation, and union. 
They cannot support many semantically rich operators common in dataflows, such as dictionary matching and classifier operators in the running example, and other operators such as unnest and sentiment analyzer. 
To investigate their limitations, we analyzed the SQL queries and dataflows from $6$ workloads, and generated an equivalent version by adding an empty filter operator. Then, we used EVs from the literature~\cite{journals/pvldb/ChuMRCS18,journals/pvldb/ZhouANHX19,journals/corr/abs-2004-00481,conf/sigmod/WangZYDHDT0022} to test the equivalence of these two versions. Table~\ref{tbl:supported} shows the average percentage of pairs for each workload that can be verified by these EVs, which is very low.

\begin{table*}[htbp]
\caption{Limitations of existing EVs to verify equivalence of dataflow versions from real workloads. \textmd{Details of the dataflows are in Appendix~\ref{sec:appendix-a}.} \label{tbl:supported}}
      \vspace{-0.15in}
\begin{adjustbox}{width=\linewidth,center}
\begin{tabular}{|l|r|r|r|r|r|r|} \hline
\textbf{Workload/EV} & \textbf{\# of pairs} & \textbf{UDP~\cite{journals/pvldb/ChuMRCS18}} &  \textbf{Equitas~\cite{journals/pvldb/ZhouANHX19}} & \textbf{Spes~\cite{journals/corr/abs-2004-00481}} & \textbf{WeTune~\cite{conf/sigmod/WangZYDHDT0022}} & \begin{tabular}[c]{@{}l@{}}\textbf{AVG. \% of pairs} \\ \textbf{supported by} \\ \textbf{existing EVs} \end{tabular} \\
\hline \hline
Calcite benchmark~\cite{calcitebenchmark:website} & $232$ & $39$ &  $91$ & $120$ & $73$ &  $34.81\%$  \\ \hline

Knime workflows hub~\cite{knimeworkflows:website} & $37$ & $1$ &  $1$ & $1$ & $1$ & $2.70\%$ \\ \hline
Orange workflows~\cite{orangeworkflows:website} & $32$ & $0$ &  $0$ & $0$ & $0$ & $0.00\%$ \\ \hline
IMDB sample workload~\cite{imdbload:website} & $5$  & $0$ &  $0$ & $0$ & $0$ & $0.00\%$  \\ \hline
TPC-DS benchmark~\cite{misc/tpcds} & $99$ & $2$ &  $2$ & $2$ & $2$ & $2.02\%$ \\ \hline
Texera dataflows~\cite{texera} & $105$ & $0$ &  $0$ & $0$ & $0$ & $0.00\%$ \\ \hline
\end{tabular}
\end{adjustbox}
\end{table*}

\boldstart{Our Approach.}
To solve the problem of verifying the equivalence of two dataflow versions, we leverage the fact that the two dataflow versions are almost identical except for a few local changes ({\em T2}). 
In this paper, we present \sysName\footnote{It stands for ``Versioned Execution Equivalence Verifier.''}, a verifier to test the equivalence of two dataflow versions.
It addresses the aforementioned problem by utilizing existing EVs as a black box. In \S\ref{sec:overview}, we give an overview of the solution, which divides the dataflow version pair into small parts, called ``windows'', so that each window satisfies the EV's restrictions in order to push testing the equivalence of a window to the EV. Our approach is simple yet highly effective in solving a challenging problem, making it easily applicable to a wide range of applications.

\boldstart{Why not develop a new EV?}
A natural question arises: why do we choose to use existing EVs instead of developing a new one? Since the problem itself is undecidable, any developed solution will inherently have limitations and incompleteness. Our goal is to create a general-purpose solution that maximizes completeness by harnessing the capabilities of these existing EVs. This approach allows us to effectively incorporate any new EVs that may emerge in the future, ensuring the adaptability and flexibility of our solution.

\boldstart{Challenges and Contributions.}
During the exploration of the proposed idea, we encountered several challenges in developing \sysName:
1) How can we enhance the completeness of the solution while maintaining efficiency and effectively handling the incompleteness of the EVs?
2) How do we efficiently handle dataflow versions with a single edit and perform the verification?
3) How can we effectively handle dataflow versions with multiple edits, and can the windows overlap?
We thoroughly investigate these challenges and present the following \textbf{contributions}.

\begin{enumerate}
\item We formulate the problem of verifying the equivalence of two complex dataflow versions in iterative data analytics. To the best of our knowledge, \sysName is the first work that studies this problem by incorporating the knowledge of user edit operations into the solution (\S\ref{sec:formal}).

\item We give an overview of the solution and formally define the ``window'' concept that is used in the equivalence verification algorithm (\S~\ref{sec:overview}).

\item We first consider the case where there is a single edit. We analyze how the containment between two windows is related to their equivalence results, and use this analysis to derive the concept of ``maximal covering window'' for an efficient verification. We give insights on how to use EVs and the subtle cases of the EVs' restrictions and completeness (\S\ref{sec:single-edit}).

\item We study the general case where the two versions have multiple edits. We analyze the challenges of using overlapping windows, and  propose a solution based on the ``decomposition'' concept. We discuss the correctness and the completeness of our algorithm (\S\ref{sec:multi-changes}).

\item We provide a number of optimizations in \sysOpt to improve the performance of the baseline algorithm (\S~\ref{sec:optimize}).

\item We report the results of a thorough experimental evaluation of the proposed solutions. The experiments show that the proposed solution is not only able to verify dataflows that cannot be verified by existing EVs, but also able to do the verification efficiently (\S~\ref{sec:experiment}).

\end{enumerate}

%% file: sec2-problem-formulation.tex
\section{Problem Formulation}\label{sec:formal}


In this section, we use an example dataflow to describe the setting. We also formally define the problem of verifying equivalence of two dataflow versions. 
Table~\ref{tbl:notation} shows a summary of the notations used in this section.

\begin{table}[htbp]
\caption{Notations used for a single dataflow. \label{tbl:notation}}
\begin{adjustbox}{width=\linewidth,center}
\begin{tabular}{|l|l|}
\hline
\textbf{Notation}              & \textbf{Description}                 \\ \hline \hline
$W$, DAG & A data processing workflow                  \\
$\mathbb D_w = \{D_1, \ldots, D_l\}$ & A set of data sources in the dataflow       \\
$\mathbb{S}_w = \{s_1, \ldots, s_n\}$ & A set of sinks in the dataflow             \\
$\mathcal{M}$ & An edit mapping between two versions\\
$\delta_j$ & A set of edit operations to \\  & transform DAG $v_j$ to $v_{j+1}$ \\
$\oplus$ & Applying aggregated edit operations \\ & on a dataflow version                          \\
$\mathcal{V}_w = [v_1, \ldots, v_m]$   & A list of dataflow versions \\
\hline 
\end{tabular}
\end{adjustbox}
\end{table}

\boldstart{Data processing workflow.} 
We consider a data processing workflow $W$ as a directed acyclic graph (DAG), where each vertex is an operator and each link represents the direction of data flow. Each operator contains a computation function, we call it a {\em property} such as a predicate condition, e.g., Price $< 20$. Each operator has outgoing links, and its produced data is sent on each outgoing link. An operator without any incoming links is called a {\sf Source}.
An operator without any outgoing links is called a {\sf Sink}, and it produces the final results as a table to be consumed by the user. A dataflow may have multiple data source operators denoted as $\mathbb D_W = \{D_1, \ldots, D_l\}$ and multiple sink operators denoted as $\mathbb S_W = \{s_1, \ldots, s_n\}$.

For example, consider a dataflow in Figure~\ref{fig:running-example-v1}. It has two source operators {\sf ``Tweets''} and {\sf ``Users''} and  two sink operators $s_i$ and $s_p$ to show a tabular result and a scatterplot visualization, respectively. The {\sf OuterJoin} operator has two outgoing links to push its processed tuples to the downstream {\sf Aggregate} and {\sf Sink} operators. The {\sf Filter} operator's properties include the boolean selection predicate.

\subsection{Dataflow Version Control}
A dataflow $W$ undergoes many edits from the time it was first constructed as part of the iterative process of data analytics~\cite{journals/debu/0001D18, conf/sc/WoodmanHWM11}. A dataflow $W$ has a list of versions $V_W = [v_1, \ldots, v_m]$ along a timeline in which the dataflow changes. Each $v_j$ is an immutable version of dataflow $W$ in one time point following version $v_{j-1}$, and contains a number of edit operations to transform $v_{j-1}$ to $v_j$.

\begin{definition}[Dataflow edit operation]
We consider the following edit operations on a dataflow:
\begin{itemize}
    \item An addition of a new operator.
    \item A deletion of an existing operator.
    \item A modification of the properties of an operator while the operator's type remains the same, e.g., changing the predicate condition of a {\sf Select} operator.
    \item An addition of a new link.
    \item A removal of an existing link. \footnote{We assume links do not have properties. Our solution can be generalized to the case where links have properties.}
\end{itemize}
\end{definition}\label{def:edit}

A combination of these edit operations is a {\em transformation}, denoted as $\delta_j$. 
The operation of applying the transformation $\delta_j$ to a dataflow version $v_j$ is denoted as $\oplus$, which produces a new version $v_{j+1}$. Formally, 

\begin{equation} 
\label{equ:transform} 
v_{j+1} \quad = \quad  v_j \quad \oplus \quad \delta_j.
\end{equation}

In the running example, the analyst makes edits to revise the dataflow version $v_1$ in Figure~\ref{fig:running-example-v1}. In particular, she (1) deletes the {\sf Filter\textsubscript{o}} operator; (2) adds a new {\sf Filter\textsubscript{h}} operator; (3) and adds a new {\sf Filter\textsubscript{g}} operator.
These operations along with the necessary link changes to add those operators correspond to a transformation, $\delta_1$ and applying it on $v_1$ will result in a new version $v_2$, illustrated in Figure~\ref{fig:running-example-v2}.

\boldstart{Dataflow edit mapping.} Given a pair of versions ($P, Q$) and an edit mapping $\mathcal M$, there is a corresponding transformation from $P$ to $Q$, which aligns every operator in $P$ to at most one operator in $Q$. Each operator in $Q$ is mapped onto by at most one operator in $P$. A link between two operators in $P$ maps to a link between the corresponding operators in $Q$. Those operators and links in $P$ that are not mapped to any operators and links in $Q$ are assumed to be deleted. Similarly, those operators and links in $Q$ that are not mapped onto by any operators and links in $P$ are assumed to be inserted.

Figure~\ref{fig:mapping} shows an example edit mapping between the two versions $v_1$ and $v_2$ in the running example.  As {\sf Filter\textsubscript{y}} from $v_1$ is deleted, the opertor is not mapped to any operator in $v_2$. 

\begin{figure}[hbt]
           \begin{subfigure}[t]{0.49\linewidth}
\includegraphics[width=\linewidth]{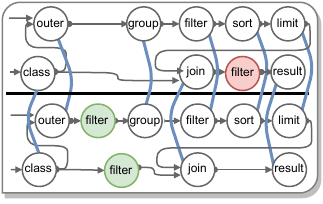}
         \caption{Mapping of operators.}
         \label{fig:ops}
         \end{subfigure}
           \begin{subfigure}[t]{0.49\linewidth}
             \includegraphics[width=\linewidth]{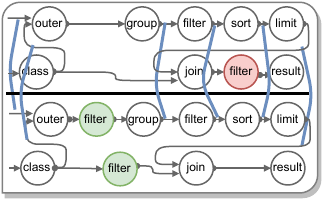}
         \caption{Mapping of links.}
         \label{fig:links}
         \end{subfigure}
         \caption{Example of an edit mapping between version $v_1$ and $v_2$. \textmd{Portions of the dataflows are omitted for clarity.}}
         \label{fig:mapping}
     \end{figure}

Figure~\ref{fig:completeness-mapping} illustrates an example to show the relation between edit mapping and edit operations. Suppose two versions $P$ and $Q$ as follows: $$P = \{Project(all) \rightarrow Filter(age > 24) \rightarrow Aggregate(\textrm{count by age}) \}.$$
$$Q = \{Aggregate(\textrm{count by age}) \rightarrow Filter(age > 24) \rightarrow Project(all) \}.$$ There can be many different mappings that correspond to different set of edit operations. Consider a mapping $\mathcal M_1$ where Project from $P$ is mapped to Aggregate in $Q$, Filter in $P$ is mapped to Filter in $Q$, and Aggregate in $P$ is mapped to Project in $Q$. This mapping yields the set of edit operations of (swapping project and aggregate in both versions). A different mapping $\mathcal M_2$, which maps Aggregation in $P$ to that in $Q$ yields the edit operations of deleting both Project and Filter in $P$ and inserting them after Aggregate in version $Q$.

\begin{figure}[hbtp]
\includegraphics[width=\linewidth]{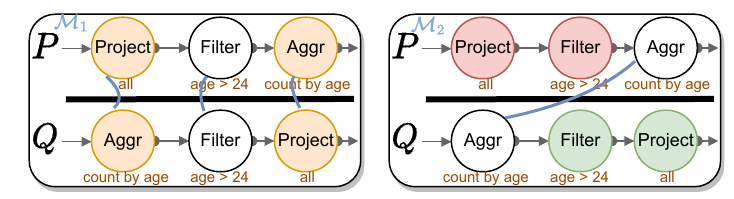}
         \caption{An example of two edit mappings ($\mathcal M_1$ on the left and $\mathcal M_2$ on the right) leading to two different sets of edit operations. \textmd{Update edit operation is highlighted in orange, removal in red, and addition in green.}}
         \label{fig:completeness-mapping}
\end{figure}

\subsection{Dataflow's Execution and Results}
A user submits an execution request to run a dataflow version. The execution produces {\em result} of each sink in the version. 
We make the following assumption: 
\begin{assume}
Multiple executions of a dataflow (or a portion of the dataflow) will always produce the same results~\footnote{This assumption is valid in many real-world applications as we detail in the experiment Section~\ref{sec:experiment}}.   
\end{assume} \label{def:assume}

\boldstart{Result equivalence of dataflow versions.}
The execution request for the version $v_j$ may produce a sink result equivalent to the corresponding sink of a previous executed version $v_{j-k}$, where $k < j$.  For example, in Figure~\ref{fig:running-example-v2}, executing the dataflow version $v_2$ produces a result of the scatterplot sink $s_2$ equivalent to the result of the corresponding scatterplot of $v_1$. In particular, $v_2$'s edit is pushing down the {\sf Filter} operator and the scatterplot result remains the same. 
Notice however that the result of $s_i$ in $v_2$ is not equivalent to the result of $s_i$ in $v_1$ because of the addition of the new {\sf Filter\textsubscript{h}} operator. Now, we formally define ``sink equivalence.''

\begin{definition}[Sink Equivalence and Version-Pair Equivalence] 
Consider two dataflow versions $P$ and $Q$  with a set of edits $\delta = \{c_1 \ldots c_n\}$ and the corresponding mapping $\mathcal M$ from $P$ to $Q$. Each version can have multiple sinks.
For each sink $s$ of $P$, consider the corresponding sink $\mathcal M(s)$ of $Q$.
We say $s$ is {\em equivalent to} $\mathcal M(s)$, denoted as ``$s \equiv \mathcal M(s)$,'' if for every instance of data sources of $P$ and $Q$, the two sinks produce the same result under the application's specified table semantics. In other words, for every tuple $t$ in $s$ \textit{exists} in $\mathcal M(s)$ under Set table semantics, every tuple $t$ in $s$ has an \textit{identical} tuple in $\mathcal M(s)$ under Bag table semantics, or every tuple $t$ in $s$ has an \textit{identical} tuple and the same \textit{order} in $\mathcal M(s)$ under Ordered Bag table semantics.
\sadeem{We say $s$ is {\em inequivalent to} $\mathcal M(s)$, denoted as ``$s \not \equiv \mathcal M(s)$,'' if there exists an instance of data sources of $P$ and $Q$ where the two sinks produce different results.}
The two versions are called \textit{equivalent}, denoted as ``$P \equiv Q$'', if each pair of their sinks under the mapping is equivalent.
\sadeem{The two versions are called \textit{inequivalent}, denoted as ``$P \not \equiv Q$, if any pair of their sinks under the mapping is inequivalent.}
\end{definition} 
\label{def:equiv} 

In this paper, we study the problem where the two versions have a single sink. We generalize the solution to the case of multiple sinks in a followup work~\cite{conf/hilda/AlsudaisK023}.

\boldstart{Expressive power of dataflows and SQL queries.}
Dataflows may involve complex operations, such as user-defined functions (UDFs). We consider dataflow DAGs that can be viewed as a class of SQL queries that do not contain recursion. \sadeem{In this paper, we focus on dataflows using relational operators including union, selection, projection, join, and aggregation (USPJA) possibly with arithmetic expressions, as well as UDFs with deterministic functions under the application's specific table semantics (sets, bags, or ordered bags).} Thus, the problem of testing the equivalence of two dataflow versions can be treated as testing the equivalence of two SQL queries without recursion.
In the remaining of the paper we use ``query'' and a ``dataflow version'' interchangeably.

\subsection{Equivalence Verifiers (EVs)}
An equivalence verifier (or ``EV'' for short) takes as an input a pair of SQL queries $\mathbf{Q_1}$ and $\mathbf{Q_2}$. \sadeem{An EV returns {\tt True} when $\mathbf{Q_1} \equiv \mathbf{Q_2}$, {\tt False} when $\mathbf{Q_1} \not \equiv \mathbf{Q_2}$, or {\tt Unknown} when the EV cannot determine the equivalence of the pair under a specific table semantics~\cite{conf/sigmod/WangZYDHDT0022, journals/pvldb/ChuMRCS18, journals/pvldb/ZhouANHX19, journals/corr/abs-2004-00481, conf/tacas/MouraB08, conf/cav/GrossmanCIRS17}.}
For instance, {\sf UDP}~\cite{journals/pvldb/ChuMRCS18} and {\sf Equitas}~\cite{journals/pvldb/ZhouANHX19} are two EVs. The former uses U-expressions to model a query while the latter uses a symbolic representation. Both EVs internally convert the expressions to a first-order-logic (FOL) formula and then push the formula to a solver such as an SMT solver~\cite{conf/tacas/MouraB08} to decide its satisfiability. 
An EV requires two queries to meet certain requirements (called ``restrictions'') in order to test their equivalence. We will discuss these restrictions in detail in Section~\ref{subsec:ev}.


\begin{problem} 
Given an EV and two dataflow versions $P$ and $Q$ with their mapping $\mathcal M$, verify if the two versions are equivalent.
\end{problem}

%% file: sec3-overview.tex
\section{\sysName: Verifying equivalence of a version pair}
\label{sec:overview}

In this section, we give an overview of \sysName for checking equivalence of a pair of dataflow versions (Section~\ref{subsec:solution-overview}).  We formally define the concepts of ``window'' and ``covering window'' (Section~\ref{subsec:windows}).

\subsection{\sysName: Overview} 
\label{subsec:solution-overview}

To verify the equivalence of a pair of sinks in two dataflow versions, \sysName leverages the fact that the two versions are mostly identical except for a few places with edit operations. It uses existing EVs as a black box. Given an EV, our approach is to break the version pair into multiple ``windows,'' each of which includes local changes and satisfies the EV's restrictions to verify if the pair of portions of the dataflow versions in the window is equivalent, as illustrated in Figure~\ref{fig:overview}. 
We consider different semantics of equivalence between two tuple collections, including sets, bags, and lists, depending on the application of the dataflow and the given EV.
\sysName is agnostic to the underlying EVs, making it usable for any EV of choice. 

\begin{figure}[hbt]
  \centering 
         \includegraphics[width=\linewidth]{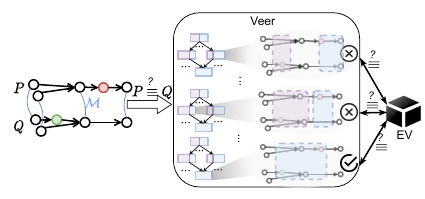}
  \caption{Overview of \sysName. Given an EV and two versions with their mapping, \sysName breaks (decomposes) the version pair into small windows, each of which satisfies the EV's restrictions. It finds different possible decompositions until it finds one with each of the windows verified as equivalent by the EV.}
  \label{fig:overview}
\end{figure}

Next we define concepts used in this approach.

\subsection{Windows and Covering Windows} 
\label{subsec:windows}

\begin{definition}[Window]
Consider two dataflow versions $P$ and $Q$  with a set of edits $\delta = \{c_1 \ldots c_n\}$ from $P$ to $Q$ and a corresponding mapping $\mathcal M$ from $P$ to $Q$. A {\em window}, denoted as $\omega$, is a pair of sub-DAGs $\omega(P)$ and $\omega(Q)$, where $\omega(P)$ (respectively $\omega(Q)$) is a connected induced sub-DAG of $P$ (respectively $Q$). Each pair of operators/links under the  mapping $\mathcal M$ must be either both in $\omega$ or both outside $\omega$. 
\end{definition} \label{def:window}

The operators in the sub-DAGs $\omega(P)$ and $\omega(Q)$ without outgoing links are called their {\em sinks}.  Recall that we assume each dataflow has a single sink. However, the sub-DAG $\omega(P)$ and $\omega(Q)$ may have more than one sink. This can happen, for example, when the window contains a {\sf Replicate} operator. 
Figure~\ref{fig:window-r} shows a window $\omega$, where each sub-DAG includes the {\sf Classifier} operator and two downstream operators {\sf Left-Outerjoin} and {\sf Join}, which are two sinks of the sub-DAG. We omit portions of the dataflows in all the figures throughput the paper for clarity.

\begin{figure}[htbp]
        \vspace{-0.15in}
\includegraphics[width=\linewidth]{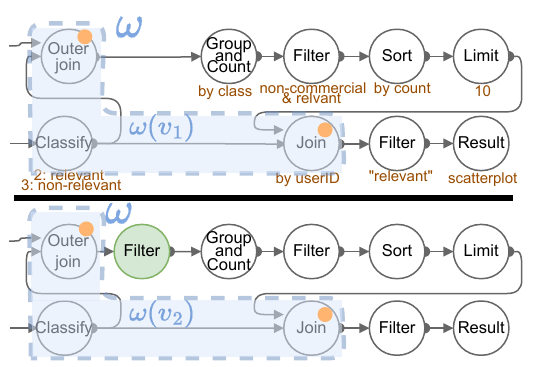}
         \caption{An example window $\omega$ (shown as ``\mytikzbox'') and each sub-DAG of $\omega(v_1)$ and $\omega(v_2)$ contains two sinks (shown as ``\mytikzdot''). \textmd{$v_1$ is above the horizontal line and $v_2$ is below the line.}}
         \label{fig:window-r}
         \end{figure}

\begin{definition}[Neighboring Window]
Consider two dataflow versions $P$ and $Q$  with a set of edits $\delta = \{c_1 \ldots c_n\}$ from $P$ to $Q$ and a corresponding mapping $\mathcal M$ from $P$ to $Q$. We say two windows, $w_1$ and $w_2$, are \textit{neighbors} if there exists a sink (or source) operator of the sub-DAGs in one of the windows that is a direct upstream (respectively downstream) operator in the original DAG to a source (respectively sink) operator in of the sub-DAGs in the other window.
\end{definition}

Figure~\ref{fig:neighbor} shows two neighboring windows, $\omega_1$ and $\omega_2$. The sink operator of $\omega_1$ ({\sf Filter}) is a direct upstream operator to the source operator of $\omega_2$ ({\sf Sort}) in the original DAG.

\begin{figure}[htbp]
        \vspace{-0.1in}
\includegraphics[width=\linewidth]{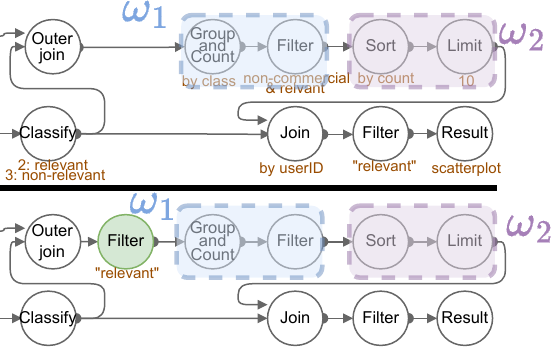}
         \caption{An example showing two neighboring windows.}
         \label{fig:neighbor}
         \end{figure}
         
\begin{definition}[Covering window]
Consider two dataflow versions $P$ and $Q$  with a set of edits $\delta = \{c_1 \ldots c_n\}$ from $P$ to $Q$ and a corresponding mapping $\mathcal M$ from $P$ to $Q$.  A {\em covering window}, denoted as $\omega_C$, is a window to cover a set of changes $C \subseteq \delta$. That is, the sub-DAG in $P$ (respectively sub-DAG in $Q$) in the window includes the operators/links of the edit operations in $C$.  
\end{definition}

When the edit operations are clear in the context, we will simply write $\omega$ to refer to a covering window. 
 Figure~\ref{fig:window} shows a covering window for the change of adding the operator {\sf Filter\textsubscript{h}} to $v_2$. The covering window includes the sub-DAG $\omega(v_1)$ of $v_1$ and contains the {\sf Aggregate} operator. It also includes the sub-DAG $\omega(v_2)$ of $v_2$ and contains the {\sf Filter\textsubscript{h}} and {\sf Aggregate} operators. 

\begin{figure}[htbp]
        \vspace{-0.15in}
        \includegraphics[width=\linewidth]{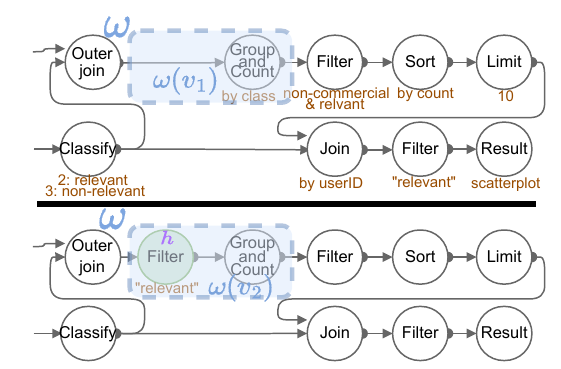}
                \vspace{-0.2in}
         \caption{A covering window $\omega$ for adding {\sf Filter\textsubscript{h}}.}
         \label{fig:window}
        \vspace{-0.1in}
\end{figure}

\begin{definition}[Equivalence of the two sub-DAGs in a window] 
\label{def:sub}
We say the two sub-DAGs $\omega(P)$ and $\omega(Q)$ of a window $\omega$ are {\em equivalent}, denoted as ``$\omega(P) \equiv \omega(Q)$,'' if they are equivalent as two stand-alone DAG's, i.e., without considering the constraints from their upstream operators.  That is, for every instance of source operators in the sub-DAGs (i.e., those operators without ancestors in the sub-DAGs), each sink $s$ of $\omega(P)$ and the corresponding $\mathcal M(s)$ in $\omega(Q)$ produces the same results.  In this case, for simplicity, we say this window is equivalent.
\end{definition} 

Figure~\ref{fig:equivalent} shows an example of a covering window $\omega'$, where its sub-DAGs $\omega'(v_1)$ and $\omega'(v_2)$ are equivalent.

\begin{figure}[htbp]
        \vspace{-0.05in}
\includegraphics[width=\linewidth]{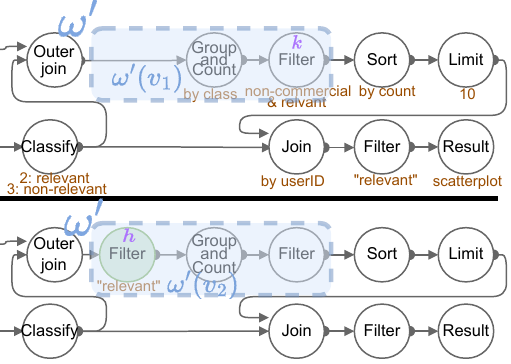}
         \caption{An example covering window $\omega'$ showing its pair of sub-DAGs are equivalent.}
         \label{fig:equivalent}
\end{figure}

Notice that for each sub-DAG in the window $\omega$, the results of its upstream operators are the input to the sub-DAG. The equivalence definition considers all instances of the sources of the sub-DAG, without considering the constraints on its input data as the results of upstream operators.  
For instance, consider the two dataflow versions in Figure~\ref{fig:constraints}. 
The two sub-DAGs of the shown window $\omega$ are clearly not equivalent as two general dataflows, as the top sub-DAG has a filter operator, while the bottom one does not. 
However, if we consider the constraints of the input data from the upstream operators, the sub-DAGs in $\omega$ are indeed equivalent, because each of them has an upstream filter operator with a predict $age < 50$, making the predicate $age < 55$ redundant. We use this definition of sub-DAG equivalence despite the above observation, because we treat the sub-DAGs in a window as a pair of stand-alone dataflow DAGs to pass  to the EV for verification (see Section~\ref{subsec:using-window}).

\begin{figure}[htbp]
          \vspace{-0.15in}
        \includegraphics[width=0.7\linewidth]{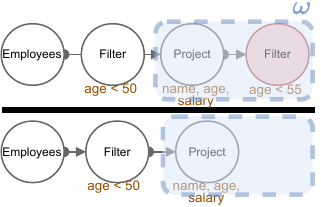}
         \caption{Two sub-DAGs in the window $\omega$ are not equivalent, as sub-DAG equivalence in  Definition~\ref{def:sub} does not consider constraints from the upstream operators. But the two complete dataflow versions are indeed equivalent.}
         \label{fig:constraints}
\end{figure}

\begin{definition}[Window containment] 
We say a window $\omega$ is {\em contained} in a window $\omega'$, donated as $\omega \subseteq_w \omega'$, if $\omega(P)$ (respectively $\omega(Q)$) of $\omega$ is a sub-DAG of the corresponding one in $\omega'$.
In this case, we call $\omega$ a {\em sub-window} of $\omega'$, and $\omega'$ a {\em super-window} of $\omega$.
\end{definition}

For instance, the window $\omega$ in Figure~\ref{fig:window} is contained in the window $\omega'$ in Figure~\ref{fig:equivalent}.

%% file: sec4-single-change.tex
\section{Two Versions with a Single Edit} 
\label{sec:single-edit}

In this section, we study how to verify the equivalence of two dataflow versions $P$ and $Q$ with a single change $c$  of the corresponding mapping $\mathcal M$ from $P$ to $Q$. We leverage a given EV $\gamma$ to verify the equivalence of two queries. 
We discuss how to use the EV to verify the equivalence of the version pair in a window (Section~\ref{subsec:using-window}), and discuss the EV's restrictions (Section~\ref{subsec:ev}). We present a concept called ``maximal covering window'', which helps in improving the performance of verifying the equivalence  (Section~\ref{subsec:maximal}), and develop a method to find maximal covering windows to verify the equivalence of the two versions (Section~\ref{subsec:maximize}). 

\subsection{Verification Using a Covering Window}
\label{subsec:using-window}

We show how to use a covering window to verify the equivalence of a version pair. 

\begin{lemma}
\label{lem:equal}
Consider a version pair $(P,Q)$ with a single edit $c$ operation between them.  If there is a covering window $\omega = (\omega(P), \omega(Q))$ of the edit operation such that the sub-DAGs of the window are equivalent, then the version pair is equivalent. 
\end{lemma}

\begin{proof}
Suppose $\omega(P) \equiv \omega(Q)$. From the definition of a covering window, every operator in one sub-DAG of the window $\omega$ has its corresponding mapped operator in the other sub-DAG of the window, and the change $c$ is included in the window. This means that the sub-DAGs of $P$ and $Q$ that precede the window $\omega$ are isomorphic (structurally identical) and the sub-DAGs of $P$ and $Q$ that follow the window are isomorphic as shown in Figure~\ref{fig:lemma}. Following the assumption that multiple runs of a dataflow produce the same result, this infers that given an instance of input sources $\mathbb D$, the sub-DAGs before the window would produce equivalent results according to definition~\ref{def:equiv}. This result becomes the input source for the window $\omega$ and given that the sub-DAGs in $\omega$ are equivalent, this means that each sink of $\omega(P)$ is equivalent to the corresponding sink (according to the mapping) of $\omega(Q)$. Hence, the output of the window, which is the input to the pair of sub-DAGs following the window, is identical, and since the operators are isomorphic, the result of the sub-DAGs following the window is equivalent. Thus, $P \equiv Q$.
\end{proof}

\begin{figure}[hbt]
           \begin{subfigure}[t]{0.49\linewidth}
\includegraphics[width=\linewidth]{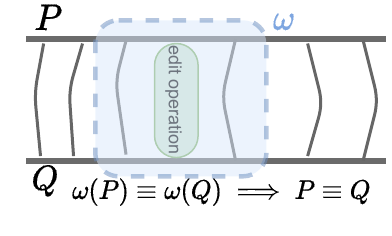}
         \caption{Using a covering window to check the equivalence of two versions.}
         \label{fig:lemma}
         \end{subfigure}
           \begin{subfigure}[t]{0.49\linewidth}
\includegraphics[width=\linewidth]{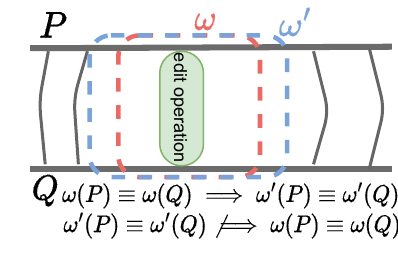}
         \caption{Subsumption of windows and their relation to version equivalence.}
         \label{fig:corollary}
         \end{subfigure}
         \caption{Conceptual examples to explain the relation between a ``covering window'' and version pair equivalence.}
         \label{fig:lemmas}
     \end{figure}

Based on this lemma, we can verify the equivalence of a pair of versions as follows: We consider a covering window and check the equivalence of its sub-DAGs by passing each pair of sinks and the sink's ancestor operators in the window (to form a query pair) to an EV. To pass a pair of sub-DAGs to the EV, we need to complete it by attaching virtual sources and sinks, which are extracted from the original DAGs. This step is vital for determining schema information. The sub-DAGs are then transformed into a representation understandable by the EV, e.g., logical DAG. If the EV shows that all the sink pairs of the sub-DAGs are equivalent, then the two versions are equivalent.


A key question is how to find such a covering window.  Notice that the two sub-DAGs in Figure~\ref{fig:window} are not equivalent. 
However, if we include the downstream {\sf Filter\textsubscript{k}} in the covering window to form a new window $\omega'$ (shown in Figure~\ref{fig:equivalent}) with a pair of sub-DAGs $\omega'(P)$ and $\omega'(Q)$, then the two sub-DAGs in $\omega'$ are equivalent. This example suggests that we may need to consider multiple windows in order to find one that is equivalent.

\subsection{EV Restrictions and Valid Windows} \label{subsec:ev}

We cannot give an arbitrary window to the EV, since each EV has certain restrictions on the sub-DAGs to verify their equivalence. 

\sadeem{
\begin{definition}[EV's restrictions]
{\em Restrictions} of an EV are a set of conditions such that for each query pair if this pair satisfies these conditions, then the EV is able to determine the equivalence of the pair without giving an {\tt Unknown} answer.
\end{definition}\label{def:restriction}

We will relax this definition in Section~\ref{sec:extension}, discuss the consequences of relaxing the definition, and propose solutions. }
There are two types of restrictions.

\begin{itemize}
    \item Restrictions due to the EV's explicit assumptions:  For example, {\sf UDP}~\cite{journals/pvldb/ChuMRCS18} and {\sf Equitas}~\cite{journals/pvldb/ZhouANHX19} support reasoning of certain operators, e.g., {\sf Aggregate} and {\sf SPJ}, but not other operators such as {\sf Sort}.
    \item Restrictions that are derived due to the modules used by the EV: For example, {\sf Equitas}~\cite{journals/pvldb/ZhouANHX19}, {\sf Spes}~\cite{journals/corr/abs-2004-00481}, and {\sf Spark Verifier}~\cite{conf/cav/GrossmanCIRS17} use an SMT solver~\cite{conf/tacas/MouraB08} to determine if a FOL formula is satisfiable or not.  
     SMT solver is not complete for determining the satisfiability of formulas when their predicates have non-linear conditions~\cite{journals/tocl/BorrallerasLROR19}. Thus, these EVs require the predicate conditions in their expressions to to be linear to make sure to receive an answer from the solver. 
\end{itemize}

As an example, the following is an example of the explicit and {\em derived} restrictions of the {\sf Equitas}~\cite{journals/pvldb/ZhouANHX19} to test the equivalence of two queries~\footnote{Applications that wish to use \sysName need to extend it to include their EV of choice if it is not Equitas~\cite{journals/pvldb/ZhouANHX19} or Spes~\cite{journals/corr/abs-2004-00481}, and incorporate the restrictions specific to those EVs.}.

\begin{itemize}
    \item[R1.] The table semantics has to be set semantics.\footnote{In this work, the application determines the desired table semantics, and \sysName decides to use an EV that supports the specified table semantics requested by the application by checking the restriction.}
    \item[R2.] All operators have to be any of the following types: SPJ, Outer join, and/or Aggregate.
    \item[R3.] The predicate conditions of SPJ operators have to be linear.
    \item[R4.] Both queries should have the same number of Outer join operators, if present.
    \item[R5.] Both queries should have the same number of Aggregate operators, if present.
    \item[R6.] If they use an Aggregate operator with an aggregation function that depends on the cardinality of the input tuples, e.g., {\sf COUNT}, then each upstream operator of the Aggregate operator has to be an SPJ operator, and the input tables are not scanned more than once.
\end{itemize}

\begin{definition} [Valid window w.r.t an EV] We say a window is {\em valid} with respect to an EV if it satisfies the EV's restrictions.
\end{definition}\label{def:valid}
In order to test if a window is valid, we pass it to a ``validator'', in which checks if the window satisfies the EV restrictions or not.
\sadeem{Testing window validity is challenging due to the complexity of exhaustively listing all potential sets of an EV's restrictions. These restrictions are designed to facilitate the evaluation of a broader range of covering windows using the provided EV (without encountering `Unknown' outcomes). To maximize completeness and identify valid windows, we can adopt a set of sufficient syntactic conditions, which may not be necessary conditions. 
If these conditions are met, we can test the equivalence of valid windows. However, this conservative approach may hinder the ability of \sysName to identify an equivalence due to the potential of more covering windows violating the rigid restrictions.
It is acknowledged that there may be additional conditions not covered that can lead to the consideration of more possible valid windows.}

 \subsection{Maximal Covering Window (MCW)} 
 \label{subsec:maximal}

A main question is how to find a valid covering window with respect to the given EV using which we can verify the equivalence of the two dataflow versions.  A naive solution is to consider all the covering windows of the edit change $c$.  For each of them, we check its validity, e.g., whether they satisfy the constraints of the EV.  If so, we pass the window to the EV to check the equivalence.  This approach is computationally costly, since there can be many covering windows. Thus our focus is to reduce the number of covering windows that need to be considered without missing a chance to detect the equivalence of the two dataflow versions. The following lemma helps us reduce the search space.
 
\begin{lemma}
\label{lemma:sub-window}
Consider a version pair ($P$, $Q$) with a single edit $c$ between them.
Suppose a covering window $\omega$ of $c$  is contained in another covering window $\omega'$. If the sub-DAGs in window $\omega$ are equivalent, then the sub-DAGs of $\omega'$ are also equivalent.
\end{lemma}

\begin{proof}
Suppose $\omega(P) \equiv \omega(Q)$. Suppose a window $\omega'$ consists of the sub-DAGs of the entire version pair, i.e. $\omega'(P) = P$ and $\omega'(Q) = Q$. This means that $\omega \subseteq \omega'$ as $\omega(P) \subseteq \omega'(P)$ and $\omega(Q) \subseteq \omega'(Q)$. Given that the sub-DAGs in $\omega$ are equivalent, from Lemma~\ref{lem:equal}, we can infer the version pair is equivalent, which means the sub-DAGs in the window $\omega'$ are equivalent.  
\end{proof}

Based on Lemma~\ref{lemma:sub-window}, we can just focus on covering windows that have as many operators as possible without violating the constraints of the EV.  If the EV shows that such a window is not equivalent, then none of its sub-windows can be equivalent.  
Based on this observation, we introduce the following concept.

\begin{definition} [Maximal Covering Window (MCW)] Given a dataflow version pair ($P, Q$) with a single edit operation $c$, a valid covering window $\omega$ is called {\em maximal} if it is not properly contained by another valid covering window.
\end{definition}

The change $c$ may have more than one MCW, For example, suppose the EV is {\sf Equitas}~\cite{journals/pvldb/ZhouANHX19}. Figure~\ref{fig:windows} shows two MCWs to cover the change of adding the {\sf Filter\textsubscript{h}} operator. One maximal window $\omega_1$ includes the change {\sf Filter\textsubscript{h}} and {\sf Left Outerjoin} on the left of the change. The window cannot include the {\sf Classifier} operator from the left side because {\sf Equitas} cannot reason its semantics~\cite{journals/pvldb/ZhouANHX19}. Similarly, the {\sf Aggregate} operator on the right cannot be included in $\omega_1$ because one of {\sf Equitas}~\cite{journals/pvldb/ZhouANHX19} restrictions is that the input of an {\sf Aggregate} operator must be an SPJ operator and the window already contains {\sf Left OuterJoin}. To include the {\sf Aggregate} operator, a new window $\omega_2$ is formed to  exclude {\sf Left OuterJoin} and include {\sf Filter\textsubscript{k}} on the right but cannot include {\sf Sort} because this operator cannot be reasoned by {\sf Equitas}~\cite{journals/pvldb/ZhouANHX19}. 

\begin{figure}[hbt]
  \centering 
\includegraphics[width=\linewidth]{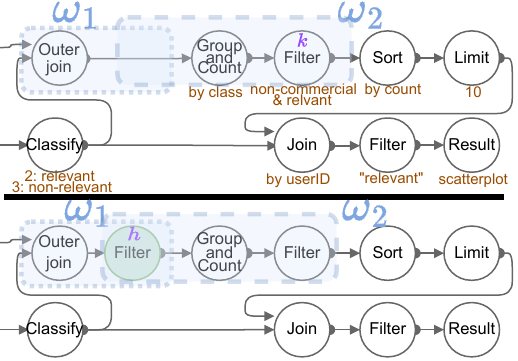}
         \caption{Two MCW $\omega_1$ and $\omega_2$ satisfying the restrictions of {\sf Equitas}~\cite{journals/pvldb/ZhouANHX19} to cover the change of adding {\sf Filter\textsubscript{h}} to $v_2$.}
         \label{fig:windows}
\end{figure}

The MCW $\omega_2$ is verified by {\sf Equitas}~\cite{journals/pvldb/ZhouANHX19} to be equivalent, whereas $\omega_1$ is not.  Notice that one equivalent covering window is enough to show the equivalence of the two dataflow versions. 

\begin{figure*}[hbt]
\includegraphics[width=\linewidth]{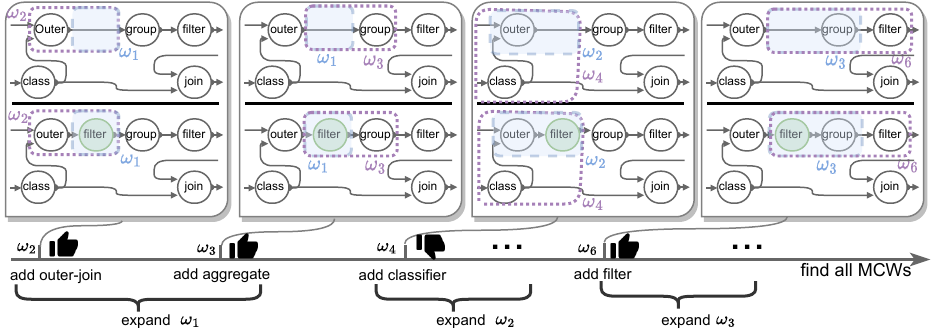}
         \caption{Example to illustrate the process of finding MCWs for the change of adding {\sf Filter\textsubscript{h}} to $v_2$.}
         \label{fig:maximize}
\end{figure*}

\subsection{Finding MCWs to Verify Equivalence}
\label{subsec:maximize}

Next we study how to efficiently find an MCW to verify the equivalence of two dataflow pairs. We present a method shown in Algorithm~\ref{alg:maximize}.  Given a version pair $P$ and $Q$ and a single edit operation $c$ based on the mapping $\mathcal M$, the method finds an MCW that is verified by the given EV $\gamma$ to be equivalent. 

\begin{algorithm}[htb]
\caption{Verifying equivalence of two dataflow versions with a single edit \label{alg:maximize}}
    \DontPrintSemicolon
    \KwIn{A version pair ($P, Q$); A single edit $c$; A mapping $\mathcal M$; An EV $\gamma$}
    \KwOut{A flag to indicate if they are equivalent} \tcp*[h]{a {\tt True} value to indicate the pair is equivalent, a {\tt False} value to indicate the pair is not equivalent, or {\tt Unknown} when the pair cannot be verified}

$\omega \leftarrow$ create an initial window to include the source and the corresponding target (operator/link) of the edit $c$\label{algl:initial}\;

$\Omega = \{ \omega \}$ \tcp*[h]{initialize a set for exploring widows} \label{algl:left}\;

\tcp*[h]{using memoization, a window is explored only once}\\
\While{$\Omega$ is not empty}{

$\omega_i \leftarrow $ remove one window from $\Omega$

\For{every neighbor of $\omega_i$ \label{algl:upstream}}{

\uIf{adding neighbor to $\omega_i$ meets EV's restrictions \label{algl:valid}}
{
add $\omega_i'$ (including the neighbor) to $\Omega$ \label{algl:add-list}\\
}
}
\uIf{none of the neighbors were added to $\omega_i$\label{algl:not-expandable}}{
\tcp*[h]{the window is maximal}\\

    \uIf{$\omega_i$ is verified equivalent by the EV \label{algl:equiv}}
    {
\Return {\tt True} \label{algl:true}
    }
        
       \uIf{$\omega_i$ is verified not equivalent by the EV and the window is the entire version pair\label{algl:notequiv}}
    {
\Return {\tt False} \label{algl:false}
    }
    }
}
      \Return {\tt Unknown} \label{algl:return}

\end{algorithm}
We use the example in Figure~\ref{fig:maximize} to explain the details of Algorithm~\ref{alg:maximize}.
The first step is to initialize the window to cover the source and target operator of the change only (line~\ref{algl:initial}). In this example, for the window $\omega_1$, its sub-DAG $\omega_1(v_2)$ contains only {\sf Filter\textsubscript{h}} and its corresponding operator using the mapping $\mathcal{M}$ in $\omega_1(v_1)$. 
Then we expand all the windows created so far, i.e.,  $\omega_1$ in this case (line~\ref{algl:left}).
To expand the window, we enumerate all possible combinations of including the neighboring operators on both $\omega_1(v_1)$ and $\omega_1(v_2)$ using the mapping.
For each neighbor, we form a new window and check if it has not been explored yet.  If not, then we check if the newly formed window is valid (lines~\ref{algl:upstream}-\ref{algl:valid}). 

In this example, we create the two windows $\omega_2$ and $\omega_3$ to include the operators {\sf Outer-join} and {\sf Aggregate} in each window,  respectively. We add those windows marked as valid in the traversal list to be further expanded in the following iterations (line~\ref{algl:add-list}).
We repeat the process on every window. After all the neighbors are explored to be added and we cannot expand the window anymore, we mark it as maximal (line~\ref{algl:not-expandable}). \sadeem{A subtle case arises when adding a single neighbor yields an invalid window, but adding a combination of neighbors yields a valid window. This is discussed in Section~\ref{subsec:complete}.} 
Then we test the equivalence of this maximal window by calling the EV. If the EV says it is equivalent, the algorithm returns {\tt TRUE} to indicate the version pair is equivalent (line~\ref{algl:equiv}). If the EV says that it is not equivalent and the window's sub-DAGs are the complete version pair, then the algorithm returns {\tt False} (line~\ref{algl:false}). Otherwise, we iterate over other windows until there are no other windows to expand. In that case, the algorithm returns {\tt Unknown} to indicate that the version equivalence cannot be verified as in line~\ref{algl:return}. 

Some EVs~\cite{journals/pvldb/ChuMRCS18, journals/pvldb/ZhouANHX19, journals/corr/abs-2004-00481} return {\tt False} to indicate that the equivalence of the version pair cannot be verified, but it does not necessarily mean that the pair is inequivalent. We take note of these EVs, and in the algorithm mentioned above, we only report {\tt False} if the EV is capable of proving the inequivalence of the pair, such as COSETTE~\cite{conf/cidr/ChuWWC17}.


%% file: sec5-multiple-changes.tex
\section{Two Versions with Multiple Edits} 
\label{sec:multi-changes}

In the previous section, we assumed there is a single edit operation to transform a dataflow version to another version. In this section, we extend the setting to discuss the case where multiple edit operations $\delta = \{c_1 \ldots c_n\}$ transform a version $P$ to a version $Q$.
A main challenge is finding covering windows for multiple edits (Section~\ref{subsec:challenge}).  We address the challenge by decomposing the version pair into a set of {\em disjoint} windows. We formally define the concepts of ``decomposition'' and ``maximal decomposition'' (Section~\ref{subsec:decomposition}). We explain how to find maximal decompositions to verify the equivalence of the version pair and prove the correctness of our solution (Section~\ref{subsec:find-maximal-decomp}). We analyze the completeness of the proposed algorithm (Section~\ref{subsec:complete}).

\subsection{Can we use overlapping windows?} 
\label{subsec:challenge}
When the two versions have more than one edit, they can have multiple covering windows. A natural question is whether we can use covering windows that overlap with each other to test the equivalence of the two versions.

\begin{definition}[Overlapping windows] We say that two windows, $\omega_1$ and $\omega_2$, \textit{overlap} if at least one operator is included in any of the sub-DAGs of both windows.
\end{definition} 
\label{def:window-overlap}

We will use an example to show that we cannot do that.  The example, shown  
Figure~\ref{fig:overlap}, is inspired from the NY Taxi dataset~\cite{nyc-taxi-data:website} to calculate the trip time based on the duration and starting time. Suppose the {\sf Select\textsubscript{x}} and {\sf Select\textsubscript{z}} operators are deleted from a version $v_1$ and {\sf Select\textsubscript{y}} operator is added to transform the dataflow to version $v_2$.  The example shows two overlapping windows $\omega$ and $\omega'$, each window is equivalent. 

\begin{figure}[hbt]
  \centering 
\includegraphics[width=\linewidth]{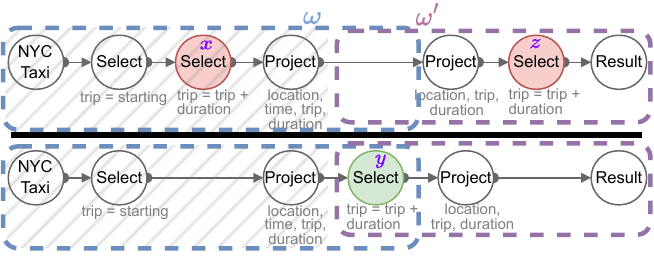}
         \caption{In this example, the blue window $\omega$ is equivalent and the purple window $\omega'$ is also equivalent. But the version pair is not equivalent. The shaded gray area is the input to window $\omega'$.}
         \label{fig:overlap}
\end{figure}

We cannot say the version pair in the example above is equivalent. The reason is that for the pair of sub-DAGs in $\omega'$ to be equivalent, the input sources have to be the same (the shaded area in grey in the example). However, we cannot infer the equivalence of the outcome of that portion of the sub-DAG. In fact, the pair of sub-DAGs in the shaded area in this example produce different results. This problem does not exist in the case of a single edit, because the input sources to any {\em covering} window (in a single edit case) will always be a one-to-one mapping of the two sub-DAGs and there is no other change outside the covering window. The solution in Section~\ref{sec:single-edit} finds {\em any} window such that its sub-DAGs are equivalent and cannot be directly used to solve the case of verifying the equivalence of the version pair when there are multiple edits.

To overcome this challenge and enable using windows to check the equivalence of the version pair, we require the covering windows to be \textit{disjoint}. In other words, each operator should be included in one and only one window. 
A naive solution is to do a simple exhaustive approach of decomposing the version pair into all possible combinations of disjoint windows. Next, we formally define a version pair decomposition and how it is used to check the equivalence of a version pair. 

\subsection{Version Pair Decomposition} \label{subsec:decomposition}
\begin{definition}[Decomposition]
    For a version pair $P$ and $Q$ with a set of edit operations $\delta = \{c_1 \ldots c_n\}$ from $P$ to $Q$, a {\em decomposition}, $\theta$ is a set of windows $\{\omega_1, \ldots, \omega_m\}$ such that:
    \begin{itemize}
        \item  Each edit is in one and only one window in the set;
        \item All the windows are disjoint;
        \item The union of the windows is the version pair.
    \end{itemize}
\end{definition}

Figure~\ref{fig:multi} shows a decomposition for the three changes in the running example. The example shows two covering windows $\omega_1$ and $\omega_2$, each covers one or more edits~\footnote{For simplicity, we only show covering windows of a decomposition in the figures throughout this section.}.
Next, we show how to use a decomposition to verify the equivalence of the version pair by generalizing Lemma~\ref{lem:equal} as follows.

\begin{figure}[hbt]
  \centering 
\includegraphics[width=0.8\linewidth]{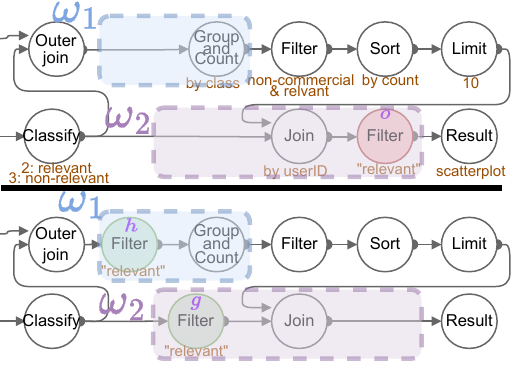}
         \caption{A decomposition $\theta$ with two covering windows $\omega_1$ and $\omega_2$ that cover the three edits.}
         \label{fig:multi}
\end{figure}

\begin{lemma} \label{lem:multi}
(Corresponding to Lemma~\ref{lem:equal}) For a version pair $P$ and $Q$ with a set of edit operations $\delta = \{c_1 \ldots c_n\}$ to transform $P$ to $Q$, if there is a decomposition $\theta$ such that every covering window in $\theta$ is equivalent, then the version pair is equivalent. 
\end{lemma}

\begin{proof}
Suppose every covering window $\omega_i$ in a decomposition $\theta$ is equivalent. Every other window that is not covering, its sub-DAGs are structurally identical, according to Definition~\ref{def:window}. Given an instance of input sources $\mathbb{D}$, we can have the following two cases. ({\tt CASE1:}) the input is processed by a pair of structurally identical sub-DAGs that are in a non-covering window. In this case, the pair of sub-DAGs produce an equivalent result since every operator is deterministic according to Assumption~\ref{def:assume}. ({\tt CASE2:}) the input is processed by a pair of sub-DAGs in a covering window. In this case, the pair of sub-DAGs produce equivalent result because we assumed each covering window is equivalent. In both cases, the output acts as the input to the following portion of the sub-DAGs (either non-covering or a covering window). This propagation continues along the pair of DAGs until the end, thus the version pair produces equivalent results as shown in Figure~\ref{fig:multi-change-lemma}.
\end{proof}

\begin{figure}[hbt]
  \centering 
\includegraphics[width=0.7\linewidth]{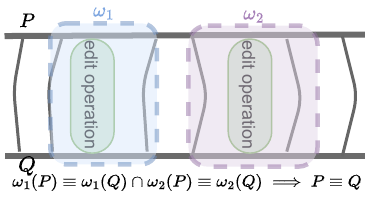}
         \caption{Using multiple covering windows on multiple edits to check the equivalence of two versions.}
         \label{fig:multi-change-lemma}
\end{figure}

A natural question is how to find a decomposition where each of its windows is equivalent.   We could exhaustively consider all the possible decompositions, but the number can grow exponentially as the size of the dataflow and the number of changes increase.  The following ``decomposition containment'' concept, defined shortly, helps us reduce the number of decompositions that need to be considered.

\begin{definition}[Decomposition containment] We say a decomposition $\theta$ is {\em contained} in another decomposition $\theta'$, denoted as $\theta \subseteq_d \theta'$, if  every window in $\theta$, there exists a window in $\theta'$ that  contains it. 
\end{definition} 
\label{def:decomp-contain}

Figure~\ref{fig:equiv-decomp} shows an example of a decomposition $\theta'$ that contains the decomposition $\theta$ in Figure~\ref{fig:multi}. We can see that in general, if a decomposition $\theta$ is contained in another decomposition $\theta'$, then each window in $\theta'$ is a concatenation of one or multiple windows in $\theta$. 

\begin{figure}[hbt]
  \centering 
\includegraphics[width=0.8\linewidth]{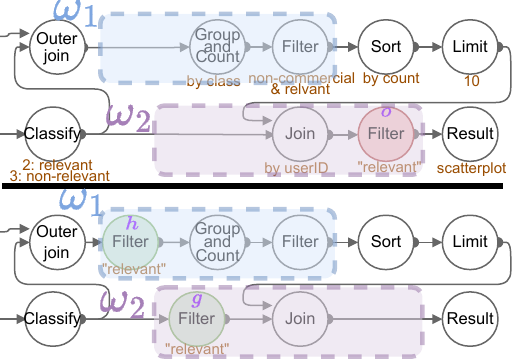}
         \caption{Example to show equivalent pair of sub-DAGs of every covering window in a decomposition $\theta'$.}
         \label{fig:equiv-decomp}
\end{figure}

The following lemma, which is a generalization of Lemma~\ref{lemma:sub-window}, can help us prune the search space by ignoring decompositions that are properly contained by other decompositions.

\begin{lemma} \label{lem:multi-contain}
(Corresponding to Lemma~\ref{lemma:sub-window}) Consider a version pair $P$ and $Q$ with a set of edit operations $\delta = \{c_1 \ldots c_n\}$ from $P$ to $Q$. Suppose a decomposition
$\theta$ is contained in another decomposition $\theta'$. If each window in $\theta$ is equivalent, then each window in $\theta'$ is also equivalent. 
\end{lemma}

\begin{proof}
Suppose each window in a decomposition $\theta$ is equivalent and the decomposition is contained in another decomposition $\theta'$.
Based on the Definition of decomposition containment~\ref{def:decomp-contain}, we know that each window in $\theta$ is contained in a window in $\theta'$. According to Lemma~\ref{lemma:sub-window}, if a window is equivalent then a window that contains it is also equivalent. We can deduce that every window in $\theta'$ is equivalent, therefore the version pair is equivalent as per Lemma~\ref{lem:multi}.
\end{proof}

\subsection{Maximal Decompositions w.r.t.~an EV}

Lemma~\ref{lem:multi-contain} shows that we can safely find decomposition that contain other ones to verify the equivalence of the version pair.  At the same time, we cannot increase each window arbitrarily, since the equivalence of each window needs to be verified by the EV, and the window needs to satisfy the restrictions of the EV.  Thus we want decompositions that are as containing as possible while each window is still valid.  We formally define the following concepts.

 \begin{definition}[Valid Decomposition] We say a decomposition $\theta$ is {\em valid} with respect to an EV if each of its covering windows is valid with respect to the EV. 
 \end{definition} 
\label{def:decomp-valid}
 
\begin{definition}[Maximal Decomposition (MD)] We say a valid decomposition $\theta$ is {\em maximal} if no other valid decomposition $\theta'$ exists such that $\theta'$ properly contains $\theta$. 
\end{definition} 
\label{def:maximal}

The decompositions w.r.t~an EV form a unique graph structure, where each decomposition is a node. It has a single root corresponding to the decomposition that includes every operator as a separate window. A downward edge indicates a ``contained-in'' relationship. A decomposition can be contained in more than one decomposition. Each leaf node at the bottom of the hierarchy is an MD as there are no other decompositions that contain it and the hierarchy may not be balanced.
If the entire version pair satisfies the EV's restrictions, then the hierarchy becomes a lattice structure with a single leaf MD being the entire version pair. Each branching factor depends on the number of changes, the number of operators, and the EV's restrictions.
Figure~\ref{fig:decomposition} shows the hierarchical relationships of the valid decompositions of the running example when the EV is {\sf Equitas}~\cite{journals/pvldb/ZhouANHX19}. The example shows two MD $\theta_{12}$ and $\theta_{16}$.

\begin{figure}[hbt]
\includegraphics[width=\linewidth, height=1.6in]{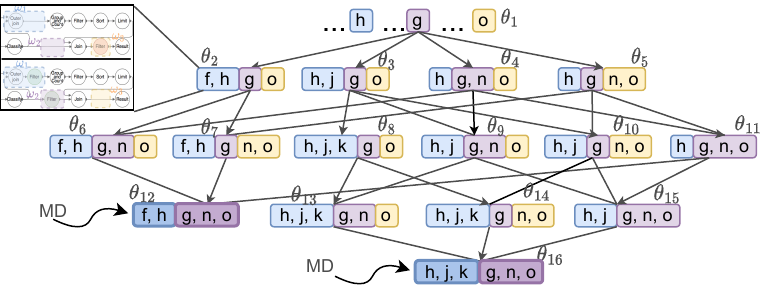}
         \caption{Hierarchy of valid decompositions w.r.t an EV. \textmd{Each letter corresponds to a pair of operators from the running example. We show the containment of covering windows and we omit details of containment of non-covering windows.}}
         \label{fig:decomposition}
\end{figure}



\subsection{Finding a Maximal Decomposition to Verify Equivalence (A Baseline Approach)} 
\label{subsec:find-maximal-decomp}
Now we present an algorithm for finding maximal decompositions shown in Algorithm~\ref{alg:decomp}.  We will explain it using the example in Figure~\ref{fig:decomposition}. We return {\tt True} to indicate the pair is equivalent if there are no changes and the two versions exactly match (Line~\ref{algl:emptydelta}-\ref{algl:return-emptydelta}).
Otherwise, we add the initial decomposition, which includes each operator as a window, to a set of decompositions to be expanded (line~\ref{algl:init-decomp}).  In each of the following iterations, we remove a decomposition from the set, and iteratively expand its windows. 
To expand a window, we follow the procedure as in Algorithm~\ref{alg:maximize} to expand its neighbors. The only difference is that the neighbors in this case are windows, and we merge windows if their union is valid (line~\ref{algl:merge}). If a window cannot be further expanded, then we mark the window as maximal to avoid checking it again (line~\ref{algl:mark}). \sadeem{A subtle case arises when the merge of two windows yields an invalid window but the merge of a combination of more windows produces a valid window. We discuss this in Section~\ref{subsec:complete}.}
If all of the windows in the decomposition are maximal, we mark the decomposition as maximal, and verify whether each covering window is equivalent by passing it to the given EV (line~\ref{algl:equiv-all}). If all of the windows are verified to be equivalent,  we return {\tt True} to indicate that the version pair is equivalent (line~\ref{algl:true-all}). 
If in the decomposition there is only a single window, which includes the entire version pair, and the EV decides that the window is not equivalent, then the algorithm returns {\tt False} (line~\ref{algl:false-all}).
Otherwise, we continue exploring other decompositions until there are no more decompositions to explore. In that case, we return {\tt Unknown} to indicate that the equivalence of the version pair cannot be determined (line~\ref{algl:phi-all}).  This algorithm generalizes Algorithm~\ref{alg:maximize} to handle cases of two versions with multiple edits.
For an efficient exploration, we only expand and maximize covering windows and verify them.

\begin{algorithm}[htb]
\caption{Verifying the equivalence of a dataflow version pair with one or multiple edits (Baseline) \label{alg:decomp}}
    \DontPrintSemicolon
        \KwIn{A version pair ($P$, $Q$); A set of edit operations $\delta$ and a mapping $\mathcal M$ from $P$ to $Q$; An EV $\gamma$}
    \KwOut{A version pair equivalence flag $EQ$} \tcp*[h]{A {\tt True} value indicates the pair is equivalent, a {\tt False} value indicates the pair is not equivalent, and an {\tt Unknown} value indicates the pair cannot be verified}

    \uIf{$\delta$ is empty \label{algl:emptydelta}}{
\Return {\tt True} \label{algl:return-emptydelta}
}

$\theta \leftarrow $ decomposition with each operator as a window \label{algl:init-decomp}

$\Theta = \{ \theta \}$ \tcp*[h]{initial set of decompositions} \label{algl:init}

\While{$\Theta$ is not empty} {

Remove a decomposition $\theta_i$ from $\Theta$ \label{algl:remove}
    
\For{every covering window $\omega_j$ (in $\theta_i$) not marked \label{algl:cover}}{
\For{each neighbor $\omega_k$ of $\omega_j$}{
\uIf{$\omega_k \cup \omega_j$ is valid and not explored before \label{algl:isvalid}}{
$\theta_i'$ $\leftarrow$ $\theta - \omega_k - \omega_j + \omega_k \cup \omega_j$ \label{algl:merge}

add $\theta_i'$ to $\Theta$
}
}
\uIf{none of the neighbor windows can be merged \label{algl:max}}
{mark $\omega_j$ \label{algl:mark}
}
}\label{algl:endloop}
\uIf{every covering $\omega \in \theta_i$ is marked \label{algl:dmaximal}}
{
\uIf{$\gamma$ verifies each covering window in $\theta_i$ to be equivalent \label{algl:equiv-all}} {
\Return {\tt True} \label{algl:true-all}
}
       \uIf{$\theta_i$ has only one window and $\gamma$ verifies it not to be equivalent \label{algl:notequiv-all}}
    {
\Return {\tt False} \label{algl:false-all}
}
}

}

\Return {\tt Unknown} \label{algl:phi-all}

\end{algorithm}



\begin{theorem}
(Correctness). Given a dataflow version pair ($P$, $Q$), an edit mapping, and a sound EV, 1) if \sysName returns {\tt True}, then $P \equiv Q$, and 2) if \sysName returns {\tt False}, then $P \not\equiv Q$.
\end{theorem}

\begin{proof}
1) Suppose  $P \not\equiv Q$. According to definition~\ref{def:equiv}, this means that for a given input sources $\mathbb D$, there is a tuple $t$ that exists in the sink of $P$ but does not exist in the sink of $Q$.
Following Assumption~\ref{def:assume} that multiple runs of a dataflow produce the same result, we can infer that there must be a set of edit operations $\delta = \{c_1, \ldots, c_n \}$ to transform $P$ to $Q$ that caused the sink of $P$ to contain the tuple $t$ but does not exist in the sink of $Q$. 
\sysName must find a valid maximal decomposition $\theta$ following Algorithm~\ref{alg:decomp}.
There are four cases the procedure terminates and returns the result:

{\tt CASE1}: The set of edits is empty and this is not the case as inferred above that there exists a change.

{\tt CASE2}: The set of maximal decompositions is empty, because none of the decompositions satisfies the EV's restrictions or none of the decompositions were verified equivalent. In this case, \sysName returns {\tt Unknown}.

{\tt CASE3}: There is a decomposition that is verified to be equivalent by a correct EV, which according to Lemma~\ref{lem:multi}, implies that the version pair is equivalent given the assumptions in our setting. However, this is not the case because we assumed $P \not\equiv Q$.

{\tt CASE4}: There is a single window in the decomposition and it is verified by the EV to be not equivalent, when the EV can verify the inequivalence of the pair, in this case \sysName returns {\tt False}. 

In all cases, \sysName did not return {\tt True}, by contraposition, this proves that $P \equiv Q$. 

2) We follow the same approach as above to prove the second case.
\end{proof}


\subsection{Improving the Completeness of Algorithm~\ref{alg:decomp}}
\label{subsec:complete}
In general, the equivalence problem for two dataflow versions is undecidable~\cite{conf/cav/GrossmanCIRS17,10.5555/551350} (reduced from First-order logic). So there is no verifier that is complete~\cite{conf/cidr/ChuWWC17}. However, there are classes of queries that are decidable such as SPJ~\cite{journals/corr/abs-2004-00481}. In this section, we show factors that affect the completeness of Algorithm~\ref{alg:decomp} and propose ways to improve its completeness.

\boldstart{1) Window validity.}
In line~\ref{algl:max} of Algorithm~\ref{alg:decomp}, if none of the neighbor windows of $\omega_j$ can be merged with $\omega_j$ to become a valid window, we mark $\omega_j$ and stop expanding it, hoping it might be a maximal window.  The following example shows that this approach could miss some opportunity to find the equivalence of two versions.

\begin{example}
Consider $\mathcal M_1$ of the two versions $P$ and $Q$ from Example~\ref{fig:completeness-mapping}. Suppose the EV is {\sf Equitas}~\cite{journals/pvldb/ZhouANHX19} and a covering window $\omega$ contains the $Project$ from $P$ and its mapped operator $Aggregate$ from $Q$. Consider the window expansion procedure in Algorithm~\ref{alg:decomp}. If we add filter operator of both versions to the window, then the merged window is not valid.  The reason is that it violates Equitas's restriction $R5$ in Section~\ref{subsec:ev}, i.e., both DAGs should have the same number of Aggregate operators.
The algorithm thus stops expanding the window. However, if we continue expanding the window till the end, the final window with three operators is still valid.
\end{example} \label{ex:monotonic}

Using this final window, we can see that the two versions are equivalent, but the algorithm missed this opportunity. This example shows that even though the algorithm is correct in terms of claiming the equivalence of two versions, it may miss opportunities to verify their equivalence.  A main reason is that the  {\sf Equitas}~\cite{journals/pvldb/ZhouANHX19} EV does not have the following property.

\begin{definition}[EV's Restriction Monotonicity]
We say an EV is {\em restriction monotonic} if for each version pair $P$ and $Q$, for each invalid window $\omega$, each window containing $\omega$ is also invalid.
\end{definition}

Intuitively, for an EV with this property, if a window is not valid (e.g., it violates the EV's restrictions), we cannot make it valid by expanding the window.  For an EV that has this property such as {\sf Spes}~\cite{journals/corr/abs-2004-00481}, when the algorithm marks the window $\omega_j$ (line~\ref{algl:mark}), this window must be maximal.  Thus further expanding the window will not generate another valid window, and the algorithm will not miss this type of opportunity to verify the equivalence. 

If the EV does not have this property such as {\sf Equitas}~\cite{journals/pvldb/ZhouANHX19}, we can improve the completeness of the algorithm as follows. We modify line~\ref{algl:isvalid} by not checking if the merged window $\omega_j \cup \omega_k$ is valid or not. We also modify line~\ref{algl:max} to test if the window $\omega_j$ is maximal with respect to the EV. This step is necessary in order to be able to terminate the expansion of a window.  We assume there is a procedure for each EV that can test if a window is maximal by reasoning the EV's restrictions.

\boldstart{2) Different edit mappings.}

Consider two different edit mappings, $\mathcal{M}_1$ and $\mathcal{M}_2$, as shown in Figure~\ref{fig:completeness-mapping}. Let us assume the given EV is {\sf Equitas}~\cite{journals/pvldb/ZhouANHX19}. If we follow the baseline Algorithm~\ref{alg:decomp}, mapping $\mathcal{M}_1$ results in a decomposition that violates {\sf Equitas}'s \textbf{R2} restriction. On the other hand, mapping $\mathcal{M}_2$ satisfies the restrictions and allows the EV to test their equivalence. This example shows that different edit mappings can lead to different decompositions. 
Notably, the edit distance of the first mapping is $2$, while the edit distance of the second one is $4$. This result shows that a minimum-distance edit mapping does not always produce the best decomposition to show the equivalence.

One way to address this issue is to enumerate all possible edit mappings~\cite{conf/gbrpr/RiesenEB13} and perform the decomposition search by calling Algorithm~\ref{alg:decomp} for each edit mapping. 
If the changes between the versions are tracked, then the corresponding mapping of the changes can be treated as the first considered edit mapping before enumerating all other edit mappings.

%% file: sec6-completeness.tex
\section{Completeness of \sysName} 
\label{subsec:completeness}
In this section, we first discuss the completeness of \sysName and the dependence on the completeness of its internal components (\S\ref{subsec:dependence}). Then we show examples of using different EVs to illustrate the restrictions that a workflow version pair needs to satisfy for \sysName to be complete for verifying the pair's equivalence and formally prove its completeness (\S\ref{subsec:example}). 

\subsection{\sysName's Completeness Dependency on Internal Components} \label{subsec:dependence}

For any pair of workflow versions, if the pair is equivalent and there is a valid decomposition w.r.t.~to a given EV where each of its covering windows is verified as equivalent by the given EV, \sysName returns {\tt True}. 
Recall that \sysName considers all possible edit mappings and explores all possible valid decompositions for each mapping.  If there is a valid decomposition under a mapping, \sysName guarantees to find it. 
For any pair of workflow versions, if the pair is inequivalent and there exists a valid decomposition that includes a single window consisting of the entire version pair, and this window is verified as inequivalent by the EV, then \sysName returns {\tt False}. 
For simplicity, throughout this section, in both cases we say there is a valid decomposition whose equivalence is determined by the given EV.  In both cases, \sysName does not return {\tt Unknown}.
Note that the completeness of \sysName relies on the completeness of the given EV. If the EV is incomplete and returns {\tt Unknown} to all possible valid decompositions generated by \sysName, accordingly \sysName  returns {\tt Unknown}. 

\begingroup
\setlength{\columnsep}{8pt}
\begin{wrapfigure}{r}{0.25\linewidth}
  \centering 
\includegraphics[width=\linewidth]{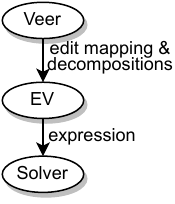}
           \vspace{-0.25in}
         \caption{Components related to the equivalence verification process.}
         \label{fig:complete}
\end{wrapfigure}
The completeness of modern EVs~\cite{journals/pvldb/ZhouANHX19,journals/corr/abs-2004-00481,journals/pvldb/ChuMRCS18,conf/cav/GrossmanCIRS17,conf/sigmod/WangZYDHDT0022} depends on the internal components used. For instance, most EVs~\cite{journals/pvldb/ZhouANHX19,journals/corr/abs-2004-00481,journals/pvldb/ChuMRCS18,conf/cav/GrossmanCIRS17,conf/sigmod/WangZYDHDT0022} model queries as expressions, such as FOL formulas, and utilize a solver, e.g., SMT, to determine the satisfiability of formulas. SMT solvers~\cite{conf/tacas/MouraB08} are complete for testing the satisfiability of linear formulas~\cite{journals/tocl/BorrallerasLROR19}. Therefore, EVs that use SMT solvers in their internal verification procedure are incomplete for verifying the equivalence of two queries (workflows or SQL) when the two queries include non-linear conditions in their predicates. Likewise, \sysName is complete for verifying the equivalence of two workflow versions that satisfy the EV's restrictions. Figure~\ref{fig:complete} illustrates the internal components \sysName uses and how these components contribute to \sysName's overall completeness. 

\endgroup

\subsection{Restrictions of Some EVs and \sysName's Completeness} \label{subsec:example}
We use the following examples on three EVs (summarized in Table~\ref{tbl:complete}) to explain \sysName's completeness process. Suppose a given EV is {\sf Spes}~\cite{journals/corr/abs-2004-00481}. {\sf Spes} determines the equivalence of two queries under the ``Bag'' table semantics. {\sf Spes} is complete for determining the equivalence of two queries that satisfy the following restrictions~\cite{journals/corr/abs-2004-00481}: 1) the two queries should contain only SPJ operators; 2) the selection predicates in every query should not include non-linear conditions. 

\begin{table*}[htbp]
\caption{Example EVs and their restrictions along with how \sysName is complete for verifying a version pair that satisfy the EV's restrictions. \label{tbl:complete}}
\begin{adjustbox}{width=\linewidth,center}
\begin{tabular}{|l|l|l|}
\hline
\textbf{EV}                                                                                 &       \textbf{\begin{tabular}[c]{@{}l@{}}EV's restrictions\end{tabular}}   &  \textbf{\begin{tabular}[c]{@{}l@{}}Explanation of why \sysName is complete\end{tabular}}                                                                                                                     \\ \hline
\textbf{{\sf Spes}~\cite{journals/corr/abs-2004-00481}}          & \begin{tabular}[c]{@{}l@{}}1. the pair should not include other than Select-Project-Join (SPJ.)\\ 2. queries should have only linear predicate conditions.\end{tabular} & \begin{tabular}[c]{@{}l@{}}\sysName finds all possible windows that satisfy the EV's \\ restrictions, and in this example, the given pair \\ satisfies the restrictions, so \sysName can find it.\end{tabular}  \\ \hline
  
\textbf{\begin{tabular}[c]{@{}l@{}}{\sf UDP}~\cite{journals/pvldb/ChuMRCS18} \end{tabular}} & \begin{tabular}[c]{@{}l@{}}1. the pair should not include other than Union-SPJ (USPJ.)\\ 2. the two versions must have a bijective mapping.\end{tabular}                 & \begin{tabular}[c]{@{}l@{}}\sysName finds all possible windows that satisfy the EV's \\ restrictions under all possible mappings, and in this \\ example, \sysName finds a bijective mapping if there exists one.\end{tabular} \\ \hline
   
\begin{tabular}[c]{@{}l@{}}\textbf{Spark} \\ \textbf{Verifier~\cite{conf/cav/GrossmanCIRS17}} \end{tabular}  & \begin{tabular}[c]{@{}l@{}}1. the pair should not include other than SPJ-Aggregate (SPJA.)\\ 2. there should not be more than one aggregate operator in each \\ version such that the aggregate is without grouping and \\ outputs a primitive, e.g., MAX and MIN.\end{tabular} & \begin{tabular}[c]{@{}l@{}}\sysName finds all possible windows that satisfy the EV's \\ restrictions, and in this example, the given pair \\ satisfies the restrictions, so \sysName finds it.\end{tabular}                                                   \\ \hline

\end{tabular}
\end{adjustbox}
\end{table*}

\begin{theorem}
\sadeem{(Completeness.) \sysName is complete for determining the equivalence of two workflow versions ($P, Q$) if the pair satisfies the restrictions of a given EV.} 
\end{theorem} \label{thrm:complete}

\begin{proof}
\sadeem{Suppose the two versions satisfy the restrictions of the given EV. Since \sysName considers all possible mappings and all possible decompositions that satisfy the EV's restrictions, it will find a decomposition with a single window that includes the entire pair because the given pair satisfies the EV's restriction. According to Definition~\ref{def:restriction}, the EV is able to determine the equivalence of a pair if the pair satisfies the EV's restriction and the EV is sound. 
\sysName returns the equivalence result from the EV.
}
\end{proof}

%% file: sec7-optimizations.tex
\section{\sysOpt: Improving Verification Performance}
\label{sec:optimize}

In this section, we develop four techniques to improve the performance of the baseline algorithm for verifying the equivalence of two dataflow versions.
We show how to reduce the search space of the decompositions by dividing the version pair into segments (Section~\ref{subsec:segment}). We present a way to detect and prune decompositions that are not equivalent (Section~\ref{subsec:pruning}).
We also discuss how to rank the decompositions to efficiently explore their search space (Section~\ref{subsec:ranking}). Lastly, we propose a way to efficiently identify the inequivalence of two dataflow versions (Section~\ref{subsec:inequiv}).

\subsection{Reducing Search Space Using Segmentations} \label{subsec:segment}

The size of the decomposition structure in Figure~\ref{fig:decomposition} depends on a few factors, such as the number of operators in the dataflow, the number of changes between the two versions, and the EV's restrictions. When the number of operators increases, the size of the possible decompositions increases.  Thus we want to reduce the search space to improve the performance of the algorithm.

The purpose of enumerating the decompositions is to find all possible cuts of the version pair to verify their equivalence. In some cases a covering window of one edit operation will never overlap with a covering of another edit operation, as shown in  Figure~\ref{fig:abstract-segment}.  In this case, we can consider the covering windows of those never overlapping separately.  Based on this observation, we introduce the following concepts.

 \begin{figure}[htbp]
\includegraphics[width=0.7\linewidth]{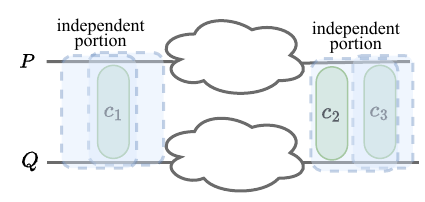}
         \caption{An example where any covering window of an edit operation $c_1$ never overlaps with a covering window of another edit operation $c_2$ or $c_3$.}
         \label{fig:abstract-segment}
         \end{figure}
\begin{definition}[Segment and segmentation]
Consider two dataflow versions $P$ and
$Q$ with a set of edits $\delta = \{c_1, \ldots, c_n\}$ from $P$ to $Q$ and a corresponding
mapping $\mathcal M$ from $P$ to $Q$. 
A {\em segment $\mathcal S$} is a window of $P$ and $Q$ under the mapping $\mathcal M$. A {\em segmentation $\psi$} is a set of disjoint segments, such that they contain all the edits in $\delta$, and there is no valid covering window that includes operators from two different segments.
 
\end{definition}
\label{def:segmentation}

A version pair may have more than one segmentation. For example, consider a version pair with a single edit. One segmentation has a single segment, which includes the entire version pair. Another segmentation includes a segment that was constructed by finding the union of MCWs of the edit.

\boldstart{Computing a segmentation.}  We present two ways to compute a segmentation. {\em 1) Using unions of MCWs:}
For each edit $c_i \in \delta$, we compute all its MCWs, and take their union, denoted as window $U_i$. We iteratively check for each window $U_i$ if it overlaps with any other window $U_j$, and if so, we merge them. We repeat this step until no window overlaps with other windows. Each remaining window becomes a segment and this forms a segmentation. Notice that a segment may not satisfy the restrictions of the given EV.
{\em 2) Using  operators not supported by the EV:} 
We identify the  operators not supported by the given EV. For example, a {\sf Sort} operator cannot be supported by {\sf Equitas}~\cite{journals/pvldb/ZhouANHX19}. Then we mark these operators as the boundaries of segments. The window between two such operators forms a segment. Compared to the second approach, the first one produces fine-grained segments, but is computationally more expensive.

\boldstart{Using a segmentation to verify the equivalence of the version pair.} As there is no valid covering window spanning over two segments, we can divide the problem of checking the equivalence of $P$ and $Q$ into sub-problems, where each sub-problem is to check the equivalence of the two sub-DAGs in a segment. Then to prove the equivalence of a version pair, each segment in a segmentation needs to be equivalent. A segment is equivalent, if there is any decomposition such that every covering window in the decomposition is equivalent. We can organize the components of the version pair verification problem as an AND/OR tree as shown in the Figure~\ref{fig:and-or-tree}.

 \begin{figure}[htbp]
\includegraphics[width=0.7\linewidth]{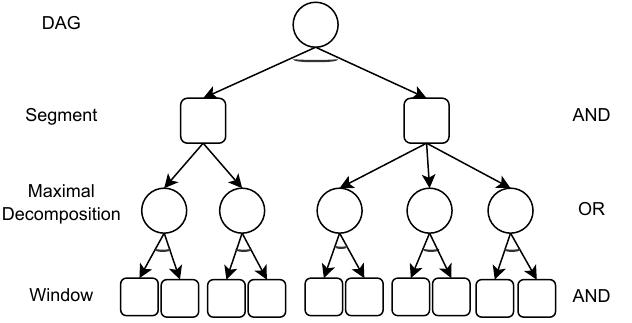}
         \caption{A sample abstract AND/OR tree to organize the components of the version pair verification problem.}
         \label{fig:and-or-tree}
         \end{figure}
         
\begin{lemma} \label{lem:segment}
For a version pair $P$ and $Q$ with a set of edit operations $\delta = \{c_1 \ldots c_n\}$ from $P$ to $Q$, if every segment $\mathcal S$ in a segmentation $\psi$ is equivalent, then the version pair is equivalent. 
\end{lemma}

\begin{proof}
Suppose every segment $S_i$ in a segmentation $\psi$ is equivalent. Since according to Definition~\ref{def:segmentation} a segment is a window and each change is covered in all of the segments in a segmentation, then we can infer that any part of the version pair that is not in a segment is structurally identical. Following the same procedure of the proof for Lemma~\ref{lem:multi}, we can say that the result is either from a structurally identical pair of sub-DAGs or from a segment, which is said to be equivalent. We can deduce that the version pair is equivalent.
\end{proof}

Algorithm~\ref{alg:segments} shows how to use a segmentation to check the equivalence of two versions.
We first construct a segmentation. 
For each segment we find if its pair is equivalent by calling   Algorithm~\ref{alg:decomp}. If any segment is not equivalent, we can terminate the procedure early. We repeat this step until all of the segments are verified equivalent and we return {\tt True}. Otherwise we return {\tt Unknown}. For the case where there is a single segment consisting of the entire version pair and Algorithm~\ref{alg:decomp} returns {\tt False}, the algorithm returns {\tt False}.

\begin{algorithm}[htb]
\caption{Using segments to verify the equivalence\label{alg:segments}}
    \DontPrintSemicolon
        \KwIn{A version pair ($P$, $Q$); A set of edit operations $\delta$ and a mapping $\mathcal M$ from $P$ to $Q$; An EV $\gamma$}
    \KwOut{A version pair equivalence flag $EQ$} \tcp*[h]{A {\tt True} value indicates the pair is equivalent, a {\tt False} value to indicate the version pair is not equivalent, and an {\tt Unknown} value indicates the pair cannot be verified}
    
$\psi \leftarrow$
constructSegmentation $(P, Q, \mathcal M)$ \label{algl:segment}

\For{every segment $\mathcal S_i \in \psi$} {
$result_i \leftarrow$ Algorithm~\ref{alg:decomp} ($\mathcal S_i$, $\delta$, $\mathcal M$, $\gamma$)

\uIf{$result_i$ is not {\tt True} \label{algl:segment-false}}{
\BREAK
}
}
\uIf{every $result_i$ is {\tt True} \label{algl:segment-true}}{
\Return {\tt True}
}
\uIf{$result_i$ is {\tt False} and there is only one segment including the entire version pair \label{algl:segment-allfalse.}}{
\Return {\tt False}
}
\Return {\tt Unknown} \label{algl:phi-segment}

\end{algorithm}

Figure~\ref{fig:segment-divide} shows the segments of the running example when using {\sf Equitas}~\cite{journals/pvldb/ZhouANHX19} as the EV. Using the second approach for computing a segmentation, we know {\sf Equitas}~\cite{journals/pvldb/ZhouANHX19} does not support the {\sf Sort} operator, so we divide the version pair into two segments.
         \begin{wrapfigure}{r}{0.5\linewidth}
\includegraphics[width=\linewidth]{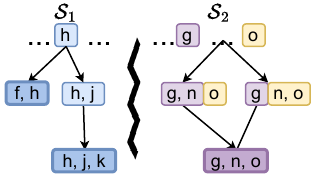}
         \caption{Two segments to reduce the decomposition-space of the running example.}
         \label{fig:segment-divide}
         \end{wrapfigure} 
The first one $\mathcal S_1$ includes those operators before {\sf Sort}, and the second one $\mathcal S_2$ includes those operators after the {\sf Sort}.
The example shows the benefit of using segments to reduce the decomposition-space to a total of $8$ (the sum of number of decompositions in every segment) compared to $16$ (the number of all possible combinations of decompositions across segments) when we do not use segments.


\subsection{Pruning Stale Decompositions} 
\label{subsec:pruning}

Another way to improve the performance is to prune {\em stale} decompositions, i.e., those that would not be verified equivalent even if they are further expanded. 
For instance, Figure~\ref{fig:pruning} shows part of the decomposition hierarchy of the running example. Consider the decomposition $\theta_2$. Notice that the first window, $\omega_1 (f, h)$, cannot be further expanded and is marked ``maximal'' but the decomposition can still be further expanded by the other two windows, thus the decomposition is not maximized. After expanding the other windows and reaching a maximal decomposition, we realize that the decomposition is not equivalent because one of its windows, e.g., $\omega_1$, is not equivalent. 

Based on this observation, if one of the windows in a decomposition becomes maximal, we can immediately test its equivalence.  If it is not equivalent, we can terminate the traversal of the decompositions after this one.
To do this optimization, we modify Algorithm~\ref{alg:decomp} to test the equivalence of a maximal window after Line~\ref{algl:mark}\footnote{We can test the equivalence of the other windows for early termination.}. 
If the window is equivalent, we continue the search as before.

         \begin{figure}[htbp]
\includegraphics[width=\linewidth]{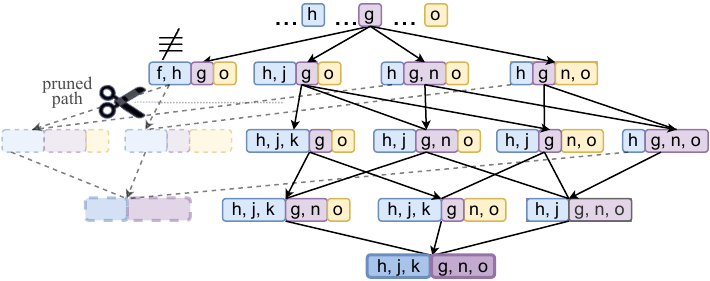}
         \caption{Example to show the pruned paths after verifying the maximal window highlighted in blue to be not equivalent.}
         \label{fig:pruning}
     \end{figure}

\subsection{Ranking-Based Search} 
\label{subsec:ranking}


\boldstart{Ranking segments within a segmentation.} Algorithm~\ref{alg:segments} needs an order to verify those segments in a segmentation one by one. If any segment is not equivalent, then there is no need for verifying the other segments. We want to rank the segments such that we first evaluate the smallest one to get a quick answer for a possibility of early termination. We consider different signals of a segment $S$ to compute its score. Various signals and ranking functions can be used. An example scoring function is $\mathcal{F}(S)= m_S + n_S$, where $m_S$ is its number of operators and $n_S$ is its number of changes. A segment should be ranked higher if it has fewer changes. The reason is that fewer changes lead to a smaller number of decompositions, and consequently, testing the segment's equivalence takes less time. Similarly, if a segment's number of operators is smaller, then the number of decompositions is also smaller and would produce the result faster.

For instance, the numbers of operators in $\mathcal S_1$ and $\mathcal S_2$ in Figure~\ref{fig:segment-divide} are $4$ and $3$, respectively.  Their numbers of changes are $1$ and $2$, respectively. The ranking score for both segments is the total of both metrics, which is $5$. Then any of the two segments can be explored first, and indeed the example shows that the number of decompositions in both segments is the same.



\boldstart{Ranking decompositions within a segment.} For each segment, we use Algorithm~\ref{alg:decomp} to explore its decompositions. The algorithm needs an order (line~\ref{algl:remove}) to explore the decompositions. The order, if not chosen properly, can lead to exploring many decompositions before finding an equivalent one. We can optimize the performance by ranking the decompositions and performing a best-first search exploration. Again, various signals and ranking functions can be used to rank a decomposition. An example ranking function for a decomposition $d$ is $\mathcal{G}(d)= o_d - w_d$, where $o_d$ is the average number of operators in its covering windows, and $w_d$ is the number of its unmerged windows (those windows that include a single operator and are not merged with a covering window). A decomposition is ranked higher if it is closer to reaching an MD for a chance of finding an equivalent one. Intuitively, if the number of operators in every covering window is large, then it may be closer to reaching an MD. Similarly, if there are only a few remaining unmerged windows, then the decomposition may be close to reaching its maximality.

For instance, decomposition $\theta_3$ in Figure~\ref{fig:decomposition} has $11$ unmerged windows, and the average number of operators in its covering windows is $1$. While $\theta_6$ has $10$ unmerged windows, and the average number of operators in its covering windows is $2$. Using the example ranking function, the score of $\theta_3$ is $1-11 = -10$ and the score of $\theta_6$ is $2-10 = -8$. Thus, $\theta_6$ is ranked higher, and it is indeed closer to reaching an MD.

\subsection{Identifying Inequivalent Pairs Efficiently}\label{subsec:inequiv}

In this section, we use an example to show how to quickly detect the inequivalence between two dataflow versions using a symbolic representation to represent partial information of the result of the sink of each version. Consider the case where two dataflow versions $P$ and $Q$ are inequivalent, as shown in Figure~\ref{fig:inequiv}. The approach discussed so far attempts to find a decomposition in which all its windows are verified equivalent. However, in cases where the version pair is inequivalent, as in this example, such a decomposition does not exist, and the search framework would continue to look for one \sadeem{unsuccessfully}. Moreover, detecting the inequivalence of the pair happens if there exists a decomposition that includes the entire version pair and satisfies the given EV's restrictions. Although the cost of maximizing the window and testing if it is valid could be low, testing the equivalence of maximal decompositions by pushing it to the EV incurs an overhead due to the EV's reasoning about the semantics of the window. Thus, we want to avoid sending a window to the EV if we can quickly determine beforehand that the version pair is not equivalent.

\begin{figure}[hbt]
\centering
           \begin{subfigure}[t]{0.49\linewidth}
           \centering
\includegraphics[width=0.75\linewidth]{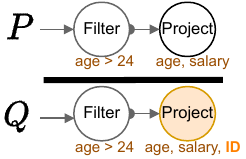}
         \caption{A sample of two inequivalent dataflow versions.}
         \label{fig:inequiv}
         \end{subfigure}
         \hfill
           \begin{subfigure}[t]{0.49\linewidth}
           \centering
\includegraphics[width=0.85\linewidth]{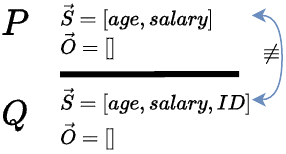}
         \caption{A partial symbolic representation of two versions showing the projected columns are different.}
         \label{fig:symbolic}
         \end{subfigure}
         \caption{Example of two inequivalent dataflow versions and their partial symbolic representation.}
         \label{fig:inequiv-ex}
     \end{figure}
     
To quickly identify the inequivalence between two dataflow versions, our approach is to create a lightweight representation that allows us to partially reason about the semantics of the version pair. This approach relies on a symbolic representation similar to other existing methods~\cite{conf/sigmod/LeFevreSHTPC14,journals/pvldb/ZhouANHX19}, denoted as $(\vec{S}, \vec{O})$. In this representation, $\vec{S}$ and $\vec{O}$ are lists that represent the projected columns in the table and the columns on which the result table is sorted, respectively. To construct the representation, we follow the same techniques in existing literature~\cite{journals/pvldb/ChuMRCS18,journals/pvldb/ZhouANHX19} by using predefined transformations for each operator. Operators inherit the representation from their upstream/parent operator and update the fields based on their internal logic. In this way, if the list of projected columns (based on $\vec{S}$) of version $P$ is different from those in $Q$, as in Figure~\ref{fig:symbolic}, we can know the two versions do not produce the same results. We can apply the same check to the sorted columns.
 


%% file: sec8-extensions.tex
\section{Extensions} \label{sec:extension}
In this section, we discuss relaxing the definition of EV's restrictions then discuss the consequences of relaxing the definition and propose extending Algorithm~\ref{alg:decomp} to handle incomplete EVs and handle multiple given EVs.

\boldstart{Relaxing EV's restrictions.}
Recall an EV's restrictions are conditions that a query pair must satisfy to guarantee the EV's completeness for determining the equivalence of the query pair. This definition of restrictions limits the opportunity to cover more query pairs. 
%
         %
Thus, we relax the definition of EV's restrictions as follows.

\begin{definition}[Relaxed EV's restrictions]
{\em Restrictions} of an EV are a set of conditions such that for each query pair if this pair satisfies these conditions, then the EV can attempt to determine the equivalence of the pair.
\end{definition}\label{def:relaxed-restriction}

However, relaxing the definition of an EV's restriction may not guarantee the completeness of the EV and may introduce the following implication.

\boldstart{Handling greedy window verification when an EV is incomplete.}
In Line~\ref{algl:dmaximal} of Algorithm~\ref{alg:decomp} we push testing the equivalence of a decomposition to the given EV only when the decomposition is marked maximal. The following example shows that this approach could miss some opportunity to find the equivalence of two versions, because the EV is not able to verify the equivalence of the two sub-DAGs in the maximal window.

\begin{example}
Consider two dataflow versions $P$ and $Q$ with a single edit $c$ based on a given mapping $\mathcal M$. Let $\gamma$ be a given EV. Suppose a covering window $\omega_c$ satisfies the restrictions of $\gamma$, and the EV is able to verify the equivalence of the two sub-DAGs in $\omega_c$. According to Algorithm~\ref{alg:decomp}, we do not check the equivalence of this window if it is not marked maximal. Let $\omega'$ be the only MCW that contains $\omega$. Following Line~\ref{algl:dmaximal} in Algorithm~\ref{alg:decomp}, if a window is maximal ($\omega'$ in this example) we push testing its equivalence to the EV. However, suppose in this case, the EV returns {\tt Unknown}, because EVs are mostly incomplete~\cite{conf/cidr/ChuWWC17} for verifying two general relational expressions. Since there is no other MCW to test, \sysName accordingly returns {\tt Unknown} for verifying the equivalence of the sub-DAGs in $\omega'$. However, if we pushed testing the equivalence of the smaller window $\omega_c$, then \sysName would have been able to verify the equivalence of the pair.
\end{example}

This example highlights the significance of verifying the equivalence of sub-DAGs within smaller windows before expanding to larger windows. The challenge arises when an EV can verify the equivalence of a small window but fails to do so for a larger one. To address this, we modify Line~\ref{algl:dmaximal} in Algorithm~\ref{alg:decomp} to check the equivalence of smaller windows by backtracking when the maximal window is not verified. This modification ensures that we do have more opportunities to verify the equivalence of the version pair. Note that this approach may introduce a computational overhead due to the repeated checking of each window and not just the maximal ones.

\boldstart{Using multiple EVs.}
As mentioned earlier, the problem of verifying the equivalence of two relational expressions is undecidable~\cite{10.5555/551350}. Thus, any given EV would have limitations and is incomplete for solving the problem of deciding the equivalence of two queries. To harness the capabilities of different EVs, we extend \sysName to take in a set of EVs and their associated restrictions. We do not modify Algorithm~\ref{alg:decomp}. However, we extend the `isValid' function in Line~\ref{algl:isvalid} to encode a window is valid w.r.t which EV so that when we verify the equivalence of the sub-DAGs in the window in Line~\ref{algl:equiv-all}, we call the corresponding EV the window satisfies.

%% file: sec9-experiments.tex
\section{Experiments} 
\label{sec:experiment}

In this section, we present an experimental evaluation of the proposed solutions. We want to answer the following questions:
Can our solution verify the equivalence of versions in a real-world pipelines workload? How does our solution perform compared to other verifiers? How the optimization techniques help in the performance? What are the parameters that affect the performance?

\subsection{Experimental Setup}
\boldstart{Synthetic workload.} We constructed four dataflows $W1-W4$ on TPC-DS~\cite{misc/tpcds} dataset as shown in Table~\ref{tbl:workload}. 
For example, dataflow $W1$'s first version was constructed based on TPC-DS Q40, which contains $17$ operators including an outer join and an aggregate operators.
dataflow $W2$'s first version was constructed based on TPC-DS Q18, which contains $20$ operators.
We omit details of other operators included in the dataflows such as {\sf Unnest}, {\sf UDF}, and {\sf Sort} as these do not affect the performance of the experimental result as we explain in each experiment.

\boldstart{Real workload.}
We analyzed a total of $179$ real-world pipelines from Texera~\cite{texera}. Among the dataflows, $81\%$ had deterministic sources and operators, and we focused our analysis on these dataflows. Among the analyzed dataflows, $8\%$ consisted primarily of $8$ operators, and another $8$ had $12$ operators. Additionally, $33\%$ of the dataflows contained $3$ different versions, while $19\%$ had $35$ versions. $58\%$ of the versions had a single edit, while $22\%$ had two edits. We also observed that the {\sf UDF} operator was changed in $17\%$ of the cases, followed by the {\sf Projection} operator ($6\%$ of the time) and the {\sf Filter} operator ($6\%$ of the time). From this set of dataflows, we selected four as a representative subset, which is presented as $W5 \ldots W8$ in Table~\ref{tbl:workload} and we used IMDB~\cite{imdbdataset:website} and Twitter~\cite{twitter-api} datasets.

\begin{table}[htbp]
\caption{Workloads used in the experiments. \label{tbl:workload}}
      \vspace{-0.15in}
\begin{adjustbox}{width=\linewidth,center}
\begin{tabular}{|l|l|l|r|r|r|} \hline
\begin{tabular}[c]{@{}l@{}}\textbf{Work} \\ \textbf{flow\#}\end{tabular} & \textbf{Description} & \textbf{Type of operators}& \begin{tabular}[c]{@{}l@{}}\textbf{\# of } \\ \textbf{operators} \end{tabular} &   \begin{tabular}[c]{@{}l@{}}\textbf{\# of } \\ \textbf{links} \end{tabular}  & \begin{tabular}[c]{@{}l@{}}\textbf{\# of} \\\textbf{versions} \end{tabular}  \\
\hline \hline
$W1$ & TPC-DS $Q40$ & \begin{tabular}[c]{@{}l@{}} 4 joins and 1 aggregate \\ operators \end{tabular} & $17$         & $16$ & $5$  \\ \hline

$W2$ & TPC-DS $Q18$ & \begin{tabular}[c]{@{}l@{}} 5 joins and 1 aggregate \\ operators \end{tabular} &    $20$      & $20$ & $9$  \\ \hline
$W3$ & TPC-DS $Q71$ & \begin{tabular}[c]{@{}l@{}} 1 replicate, 1 union, \\ 5 joins, and 1 aggregate \\ operators \end{tabular}
&      $23$        &  $23$ & $4$   \\ \hline
$W4$ & TPC-DS $Q33$  & \begin{tabular}[c]{@{}l@{}} 3 replicates, 1 union, \\ 9 joins, and 4 aggregates \\ operators \end{tabular}
&      $28$        & $34$ & $3$ \\ \hline

$W5$ & \begin{tabular}[c]{@{}l@{}}IMDB ratio of\\non-original to\\ original movie titles\end{tabular}  &   \begin{tabular}[c]{@{}l@{}} 1 replicate, 2 joins, \\ 2 aggregate operators \end{tabular} &   12     & 12 & 3 \\ \hline

$W6$ & \begin{tabular}[c]{@{}l@{}}IMDB all movies\\of directors with\\ certain criteria\end{tabular} & \begin{tabular}[c]{@{}l@{}} 2 replicates, 4 joins, \\ 2 unnest operators \end{tabular} &  18   & 20 & 3 \\ \hline
$W7$ & \begin{tabular}[c]{@{}l@{}}Tobacco Twitter \\analysis\end{tabular}  & \begin{tabular}[c]{@{}l@{}} 1 outer join, 1 aggregate, \\ classifier \end{tabular} &  14   & 13 &  3\\ \hline

$W8$ & \begin{tabular}[c]{@{}l@{}}Wildfire Twitter\\ analysis\end{tabular} & \begin{tabular}[c]{@{}l@{}}  1 join, 1 UDF\end{tabular} & 13  & 12 & 3 \\ \hline
\end{tabular}
\end{adjustbox}
\end{table}

\boldstart{Edit operations.}
For each real-world dataflow, we used the edits performed by the users. For each synthetic dataflow, we constructed versions by performing edit operations. We used two types of edit operations.

(1) Calcite transformation rules~\cite{calcitebenchmark:website} for equivalent pairs: These edits are common for rewriting and optimizing dataflows, so these edits would produce a version that is {\em equivalent} to the first version. For example, `testEmptyProject' is a single edit of adding an empty projection operator. In addition, `testPushProjectPastFilter' and `testPushFilterPastAgg' are two example edits that produce more than a single change, in particular, one for deleting an operator and another is for pushing it past other operator. We used a variation of different numbers of edits, different placements of the edits, etc., for each experiment. Thus, we have numbers of pairs as shown in Table~\ref{tbl:workload}. For each pair of versions, one of the versions is always the original one.

(2) TPC-DS V2.1~\cite{misc/tpcds} iterative edits for inequivalent pairs: These edits are common for exploratory and iterative analytics, so they may produce a version that is {\em not equivalent} to the first version. Example edits are adding a new filtering condition or changing the aggregate function as in TPC-DS queries. 
We constructed one version for each dataflow using two edit operations from this type of transformations to test our solution when the version pair is not equivalent.

We randomized the edits and their placements in the dataflow DAG, such that it is a valid edit.
Unless otherwise stated, we used two edit operations from Calcite in all of the experiments. Appendix~\ref{sec:appendix-b} shows a sample of these dataflows.

\boldstart{Implementation.}
We implemented the baseline (\sysName) and an optimized version (\sysOpt) in Java8 and Scala. We implemented {\sf Equitas}~\cite{journals/pvldb/ZhouANHX19} as the EV in Scala.
We implemented \sysOpt by including the optimization techniques presented in Section~\ref{sec:optimize}.
We evaluated the solution by comparing \sysName and \sysOpt against a state of the art verifier (Spes~\cite{journals/corr/abs-2004-00481}). We ran the experiments on a MacBook Pro running the MacOS Monterey operating system with a
2.2GHz Intel Core i7 CPU, 16GB DDR3 RAM, and 256GB SSD.
Every experiment was run three times.

\subsection{Comparisons with Other EVs}
To our best knowledge, \sysName is the first technique to verify the equivalence of complex dataflows. To evaluate its performance, we compared \sysName and \sysOpt against {\sf Spes}~\cite{journals/corr/abs-2004-00481}, known for its proficiency in verifying query equivalence compared to other solutions. We chose one equivalent pair and one inequivalent pair of versions with two edits from each dataflow. Among the $8$ dataflows examined, {\sf Spes}~\cite{journals/corr/abs-2004-00481} failed to verify the equivalence and inequivalence of any of the pairs, because all of the dataflow versions included operators not supported by {\sf Spes}. In contrast, \sysName and \sysOpt successfully verified the equivalence of $62\%$ ($W1, W2, W3, W5$) and $78\%$ ($W1 \ldots W6$), respectively, of the equivalent pairs. Both \sysName and \sysName did not verify the equivalence of the versions in $W7$ because none of the constructed decompositions were verified as equivalent by the EV. Moreover, \sysName and \sysName did not verify the equivalent pairs of $W8$ because the change to its versions was made on a UDF operator, resulting in the absence of a valid window that satisfies the EV's restrictions used in our experiments. \sysOpt was able to detect the inequivalence (using the heuristic discussed in Section~\ref{subsec:inequiv}) of about 50\% of the inequivalent pairs ($W5 \ldots W8$). We note that \sysName and \sysOpt can be made more powerful if we employ an EV that can reason the semantics of a UDF operator.
Table~\ref{tbl:comparison} summarizes the evaluation of the compared techniques.

\begin{table}[htbp]
\caption{Comparison evaluation of \sysName and \sysOpt against {\sf Spes}~\cite{journals/corr/abs-2004-00481}. \label{tbl:comparison}}
      \vspace{-0.15in}
\begin{adjustbox}{width=\linewidth,center}
\begin{tabular}{|l|r|r|r|r|} \hline
\begin{tabular}[c]{@{}l@{}}\textbf{Verifier}\end{tabular} & \begin{tabular}[c]{@{}l@{}}\textbf{\% of proved}\\ \textbf{equivalent} \\ \textbf{pairs}\end{tabular} & \textbf{Avg. time (s)}& \begin{tabular}[c]{@{}l@{}}\textbf{\% of proved} \\ \textbf{inequivalent} \\ \textbf{pairs} \end{tabular} &   \begin{tabular}[c]{@{}l@{}}\textbf{Avg. time (s)}\end{tabular} \\
\hline \hline
{\sf Spes} & $0.0$ & NA & $0.0$ & NA  \\ \hline

\sysName & $62.5$ & $32.1$ &    $0.0$      & $44.5$  \\ \hline

\sysOpt & $75.0$ & $0.1$
&      $50.0$        &  $4.1$    \\ \hline
\end{tabular}
\end{adjustbox}
\end{table}

\subsection{Evaluating \sysOpt Optimizations} 
\label{subsec:exp-opt}
We used dataflow $W3$ for evaluating the first three optimization techniques discussed in Section~\ref{sec:optimize}. We used three edit operations: one edit was after the Union operator (which is not supported by {\sf Equitas}~\cite{journals/pvldb/ZhouANHX19}) and two edits (pushFilterPastJoin) were before the Union.
We used the baseline to verify the equivalence of the pairs, and we tried different combinations of enabling the optimization techniques.
We want to know the effect of these optimization techniques on the performance of verifying the equivalence.

Table~\ref{tbl:drove-plus} shows the result of the experiments.
The worst performance was the baseline itself when all of the optimization techniques were disabled, resulting in a total of $19,656$ decompositions explored in $27$ minutes.
When only ``pruning'' was enabled, it was slower than all of the other combinations of enabling the techniques because it tested $108$ MCWs for possibility of pruning them. Its performance was better than the baseline thanks to the early termination, where it resulted in $3,614$ explored decompositions in $111$ seconds.
When ``segmentation'' was enabled, there were only two segments, and the total number of explored decompositions  was lower. In particular, when we combined ``segmentation'' and  ``ranking'', one of the segments had $8$ explored decompositions  while the other had $13$. If ``segmentation'' was enabled without ``ranking'', then the total number of explored decompositions was $430$, which was only $2\%$ of the number of explored decompositions when ``segmentation'' was not enabled. The time it took to construct the segmentation was negligible.
When ``ranking'' was enabled, the number of decompositions explored was around $21$. It took an average of $0.04$ seconds for exploring the decompositions and $0.40$ for testing the equivalence by calling the EV.
 Since the performance of enabling all of the optimization techniques was the best, in the remaining experiments we enabled all of them for \sysOpt.

\begin{table}[htbp]
\caption{Result of enabling optimizations (W3 with three edits). \textmd{ ``S'' indicates segmentation, ``P'' indicates pruning, and ``R'' indicates ranking. A \checkmark
means the optimization was enabled, a $\times$ means the optimization was disabled.The results are sorted based on the worst performance.} \label{tbl:drove-plus}}
\begin{adjustbox}{width=\linewidth,center}
\begin{tabular}{|c|c|c|r|r|r|r|} \hline
\textbf{S} & \textbf{P} &    \textbf{R} & \begin{tabular}[c]{@{}l@{}} \textbf{\# of decompositions} \\ \textbf{explored} 
\end{tabular} & \textbf{Exploration (s)} & \textbf{Calling EV (s)} & \textbf{Total time (s)} 
\\ 
\hline \hline
$\times$
&      $\times$        & $\times$   & $19,656$  & $1,629$ & $0.22$ & $1,629$                                                                              \\ \hline
$\times$
&      \checkmark        & $\times$   & $3,614$  & $111$ & $0.15$ & $111$                                                                              \\ \hline
\checkmark
&      \checkmark        & $\times$   & $430$  & $0.82$ & $0.20$ & $1.02$                                                                              \\ \hline
\checkmark
&      $\times$        & $\times$   & $430$  & $0.51$ & $0.18$ & $0.69$                                                                              \\ \hline
$\times$
&      \checkmark       & \checkmark  & $20$  & $0.39$ & $0.12$ & $0.52$                                                                              \\ \hline
\checkmark 
&      \checkmark        &  \checkmark   & $21$ & $0.20$ & $0.31$ & $0.51$                                                                               \\ \hline
$\times$ &    $\times$      & \checkmark & $20$ &  $0.07$ & $0.23$ & $0.30$                                                                                \\ \hline
\checkmark & $\times$         & \checkmark      & $21$ & $0.03$ & $0.21$ & $0.24$                                                               \\ \hline
\end{tabular}
\end{adjustbox}
\end{table}

\subsection{Verifying Two Versions with Multiple Edits}
We compared the performance of the baseline and \sysOpt. We want to know how much time each approach took to test the equivalence of the pair and how many decompositions each approach explored. We used dataflows $W1-W8$ with two edits. We used one equivalent pair and one inequivalent pair from each dataflow to evaluate the performance in these two cases.
Most dataflows in the experiment had one segment, except dataflows $W3$, $W5$, and $W6$, each of which has two segments. The overhead for each of the following steps, `is maximal' (line~\ref{algl:max}), `is valid' (line~\ref{algl:isvalid}), and `merge' (line~\ref{algl:merge}) in Algorithm~\ref{alg:decomp} was negligible, thus we only report the overhead of calling the EV.

\boldstart{Performance for verifying equivalent pairs.}
Figure~\ref{fig:decomposition-equiv} shows the number of decompositions explored by each approach. In general, the baseline explored more decompositions, with an average of $3,354$ compared to \sysOpt's average of $16$, which is less than 1\% of the baseline. The baseline was not able to finish testing the equivalence of $W3$ in less than an hour. The reason is because of the large number of neighboring windows that were caused by a large number of links in the dataflow.
\sysOpt was able to find a segmentation for $W3$ and $W6$. It was unable to discover a valid segmentation for $W5$ because all of its operators are supported by the EV, but we used the second approach of finding a segmentation as we discussed in Section~\ref{subsec:segment}. We note that the overhead of constructing a segmentation using the second approach was negligible. For dataflow $W7$, the size of the windows in a decomposition were small because the windows violated the restrictions of the used EV. Therefore, the ``expanding decompositions'' step stopped early and thus the search space (accordingly the running time) was small for both approaches. For dataflow $W8$, both \sysName and \sysOpt detected that the change was done on a non-supported operator (UDF) by the chosen EV (Equitas~\cite{journals/pvldb/ZhouANHX19}), thus the decomposition was not expanded to explore other ones and the algorithm terminated without verifying its equivalence. 

Figure~\ref{fig:performance-overhead} shows the running time for each approach to verify the equivalence. The baseline took $2$ seconds to verify the equivalence of $W1$, and $2$ minutes for verifying $W3$.
\sysOpt, on the other hand, had a running time of a sub-second in verifying the equivalence of all of the dataflows. \sysOpt tested 9 MCWs for a chance of pruning inequivalent decompositions when verifying $W6$. This caused the running time for verifying $W6$ to increase due to the overhead of calling the EV. In general, the overhead of calling the EV was about the same for both approaches. In particular, it took an average of $0.04$ and $0.10$ seconds for both the baseline and \sysOpt, respectively, to call the EV.

\begin{figure}[htb]
           \begin{subfigure}[t]{0.49\linewidth}
\includegraphics[width=\linewidth]{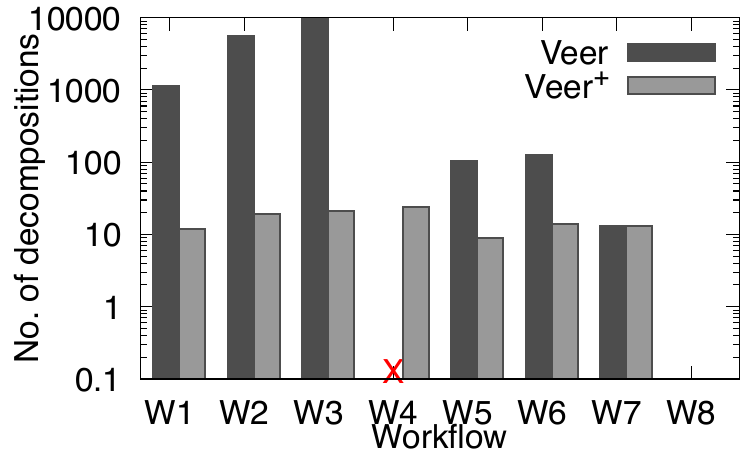}
         \caption{\# of explored decompositions.}
         \label{fig:decomposition-equiv}
         \end{subfigure}
         \hfill
         \begin{subfigure}[t]{0.49\linewidth}
\includegraphics[width=\linewidth]{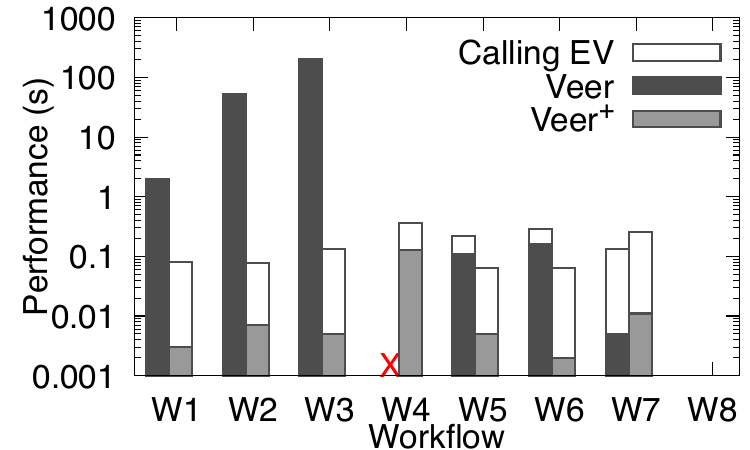}
         \caption{Running time.}
         \label{fig:performance-overhead}
\end{subfigure}
\caption{Comparison between the two algorithms for verifying equivalent pairs with two edits. \textmd{An ``$\times$'' means the algorithm was not able to finish running within one hour.}}
         \end{figure}

\boldstart{Performance of verifying inequivalent pairs.}
Figure~\ref{fig:decomposition-not} shows the number of decompositions explored by each approach. Since the pairs are not equivalent, \sysName almost exhaustively explored all of the possible decompositions, trying to find an equivalent one. \sysOpt explored fewer decompositions compared to the baseline when testing $W3$, thanks to the segmentation optimization. Both approaches were not able to finish testing $W4$ within one hour because of the large number of possible neighboring windows. \sysOpt was able to quickly detect the inequivalence of the pairs of dataflows $W5 \ldots W8$ thanks to the partial symbolic representation discussed in Scetion~\ref{subsec:inequiv}, resulting in \sysOpt not exploring any decompositions for these dataflows.

The result of the running time of each approach is shown in Figure~\ref{fig:performance-not}. \sysName's performance when verifying inequivalent pairs was the same as when verifying equivalent pairs because, in both cases, it explored the same number of decompositions. On the other hand, \sysOpt's running time was longer than when the pairs were equivalent for dataflows $W1 \ldots W4$. We observe that for $W1$, \sysOpt's running time was even longer than the baseline due to the overhead of calling the EV up to $130$, compared to only $4$ times for the baseline. \sysOpt called the EV more as it tried to continuously test MCWs when exploring a decomposition for a chance of pruning inequivalent decompositions. \sysOpt's performance on $W3$ was better than the baseline. The reason is that there were two segments, and each segment had a single change. We note that \sysOpt tested the equivalence of both segments, even though there could have been a chance of early termination if the inequivalent segment was tested first. The time it took \sysOpt to verify the inequivalence of the pairs in dataflows $W5 \ldots W8$ was negligible. The heuristic approach was not effective in detecting the inequivalence of the TPC-DS dataflows $W1 \ldots W4$. This limitation arises from the technique's reliance on identifying differences in the \textit{final} projected columns, which remained the same across all versions of these dataflows (due to the aggregation operator), with most changes occurring in the filtering conditions.

\begin{figure}[htb]
           \begin{subfigure}[t]{0.49\linewidth}
\includegraphics[width=\linewidth]{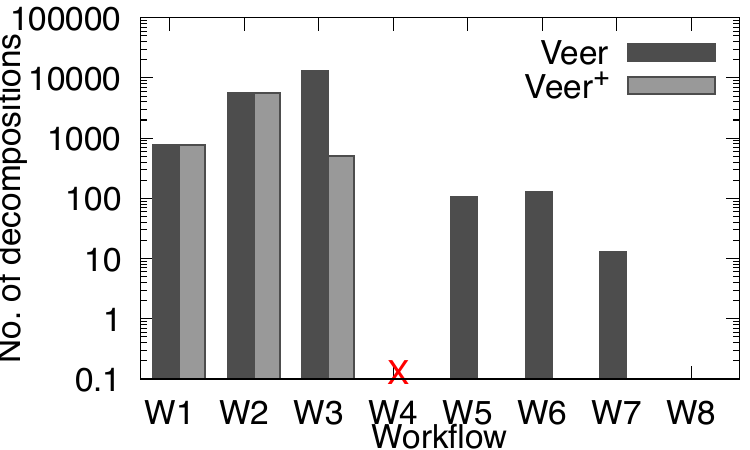}
         \caption{\# of explored decompositions.}
         \label{fig:decomposition-not}
         \end{subfigure}
         \hfill
         \begin{subfigure}[t]{0.49\linewidth}
\includegraphics[width=\linewidth]{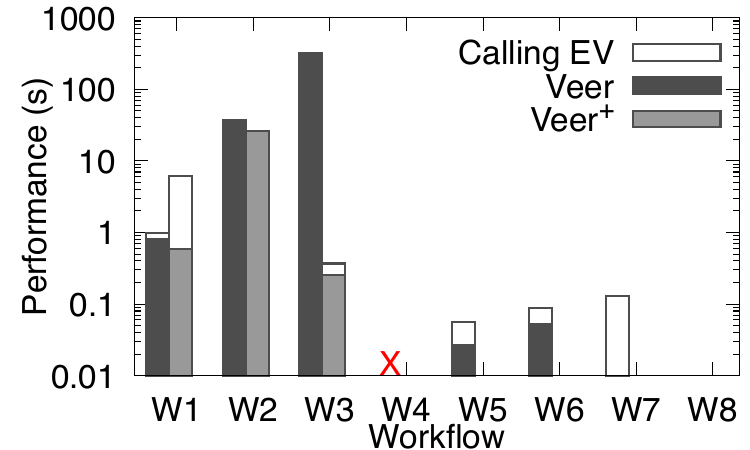}
         \caption{Running time.}
         \label{fig:performance-not}
\end{subfigure}
\caption{Comparison between the two algorithms for verifying inequivalent pairs with two edits. \textmd{An ``$\times$'' sign means the algorithm was not able to finish within an hour.}}
         \end{figure}

\subsection{Effect of the Distance Between Edits}
We evaluated the effect of the placement of changes on the performance of both approaches. We are particularly interested in how many decompositions would be explored and how long each approach would take if the changes were far apart or close together in the version DAG. We used $W2$ for the experiment with two edits. We use the `number of hops' to indicate how far apart the changes were from each other. A $0$ indicates that they were next to each other, and a $3$ indicates that they were separated by three operators between them. For a fair comparison, the operators that were separating the changes were one-to-one operators, i.e., operators with one input and one output links.

Figure~\ref{fig:hops-decomp} shows the number of decompositions explored by each approach. The baseline's number of decompositions increased from $2,770$ to $11,375$ as the number of hops increased. This is because it took longer for the two covering windows, one for each edit, to merge into a single one. Before the two covering windows merge, each one produces more decompositions to explore due to merging with its own neighbors. \sysOpt's number of explored decompositions remained the same at $21$ thanks to the ranking optimization, as once one covering window includes a neighboring window, its size is larger than the other covering window and would be explored first until both covering windows merge.

Figure~\ref{fig:hops-time} shows the time each approach took to verify the equivalence of a pair. The performance of each approach was proportional to the number of explored decompositions. The baseline took between $9.7$ seconds and $3$ minutes, while \sysOpt's performance remained in the sub-second range ($0.095$ seconds).

\boldstart{Effect of type of changed operators.} We note that when any of the changes were on an unsupported operator by the EV, then both \sysName and \sysOpt were not able to verify their equivalence. We also note that the running time to prove the pair's equivalence, was negligible because the exploration stops after detecting an `invalid' covering window.

\begin{figure}[htb]
         \begin{subfigure}[t]{0.49\linewidth}
\includegraphics[width=\linewidth]{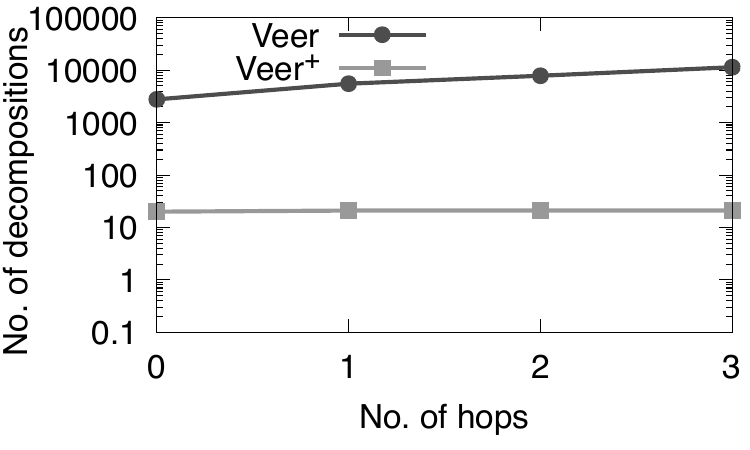}
         \caption{\# of explored decompositions.}
         \label{fig:hops-decomp}
\end{subfigure}
\hfill
           \begin{subfigure}[t]{0.49\linewidth}
\includegraphics[width=\linewidth]{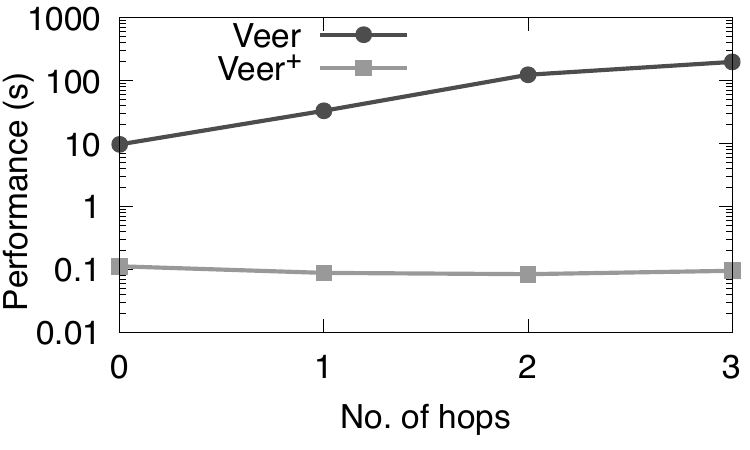}
         \caption{Running time.}
         \label{fig:hops-time}
         \end{subfigure}
         \caption{Effect of the distance between changes (W2 with two edits)}
         \end{figure}

\subsection{Effect of the Number of Changes}
In iterative data analytics, when the task is exploratory, there can be many changes between two consecutive versions. Once the analytical task is formulated, there are typically only minimal changes to refine some parameters~\cite{conf/sigmod/XuKAR22}.
We want to evaluate the effect of the number of changes on the number of decompositions and the time  each approach takes to verify a version pair. The number of changes, intuitively, increases the number of initial covering windows, and consequently, the possible different combinations of merging with neighboring windows increases. We used $W1$ in the experiment.

Figure~\ref{fig:changes-decomp} shows the number of decompositions explored by each approach and the total number of ``valid'' decompositions.The latter increased from $356$ to $11,448$ as we increased the number of changes from $1$ to $4$.
The baseline explored almost all those decompositions, with an average of $67\%$ of the total decompositions, in order to reach a maximal one. \sysOpt's number of explored decompositions, on the other hand, was not affected by the increase in the number of changes and remained the same at around $14$. The ranking optimization caused a larger window to be explored first, which sped up the merging of the separate covering windows, those that include the changes.

Figure~\ref{fig:changes} shows the time taken by each approach  to verify the equivalence of a pair. Both approaches' time was proportional to the number of explored decompositions. The baseline showed a performance of around $0.42$ seconds when there was a single change, up to slightly more than a minute at $75$ seconds when there were four changes. \sysOpt, on the other hand, maintained a sub-second performance with an average of $0.1$ seconds.

\begin{figure}[htb]
         \begin{subfigure}[t]{0.49\linewidth}
\includegraphics[width=\linewidth]{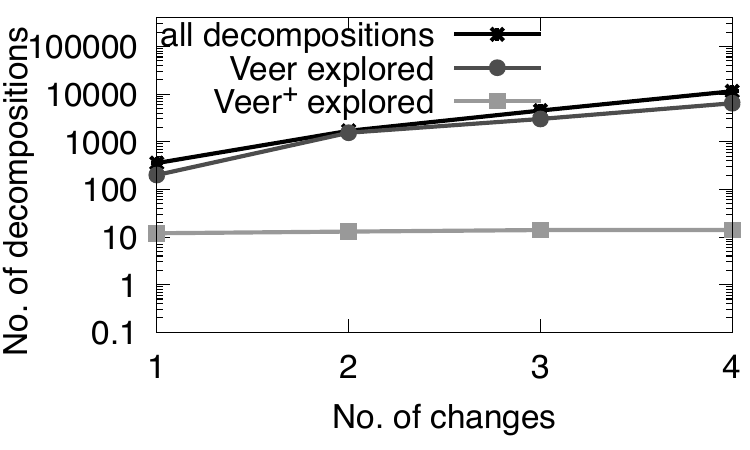}
         \caption{\# of explored decompositions.}
         \label{fig:changes-decomp}
\end{subfigure}
\hfill
           \begin{subfigure}[t]{0.49\linewidth}
\includegraphics[width=\linewidth]{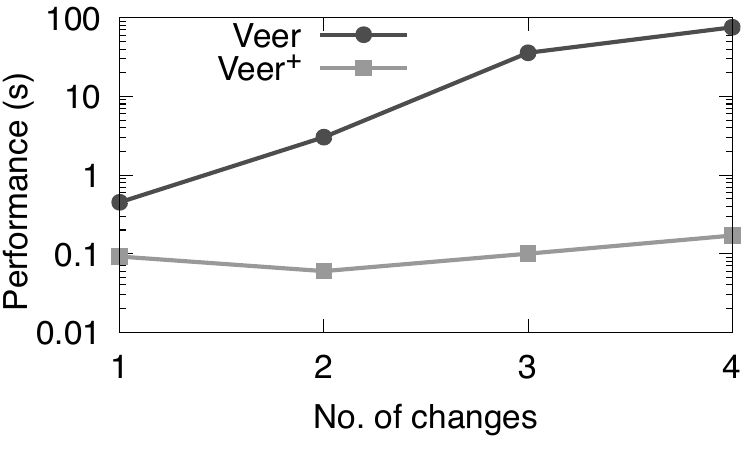}
         \caption{Running time.}
         \label{fig:changes}
         \end{subfigure}
         \caption{Effect of the number of changes (W1).}
         \end{figure}
         
 \subsection{Effect of the Number of Operators}
We evaluated the effect of the number of operators. We used $W2$ with two edits and varied the number of operators from $22$ to $25$. We varied the number of operators in two different ways.
One was varying the number of operators by including only those supported by the EV. These operators may be included in the covering windows, thus their neighbors would be considered during the decomposition exploration. The other type was varying the number of non-supported operators, as their inclusion in the dataflow DAG would not affect the performance of the algorithms. 

\boldstart{Varying the number of supported operators.}
Figure~\ref{fig:size-decomp} shows the number of explored decompositions. The baseline explored $6,650$ decompositions when there were $22$ operators, and $7,700$ decompositions when there were $25$ operators. \sysOpt had a linear increase in the number of explored decompositions from $21$ to $24$ when we increased the number of operators from $22$ to $25$.
Figure~\ref{fig:size-time} shows the results. We observed that the performance of \sysName was negatively affected due to the addition of possible decompositions from these operators' neighbors while the performance of \sysOpt remained the same. In particular, \sysName verified the pair from a minute up to $1.4$ minutes, while \sysOpt verified the pair in a sub-second.

\begin{figure}[htb]
           \begin{subfigure}[t]{0.49\linewidth}
\includegraphics[width=\linewidth]{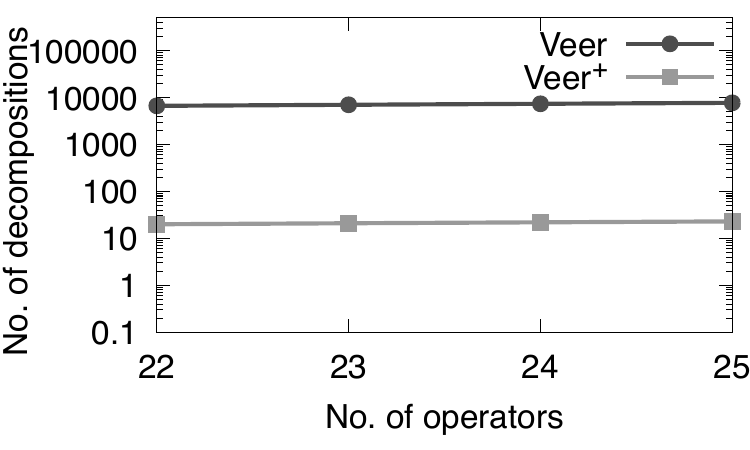}
         \caption{\# of explored decompositions.}
         \label{fig:size-decomp}
         \end{subfigure}
         \hfill
         \begin{subfigure}[t]{0.49\linewidth}
\includegraphics[width=\linewidth]{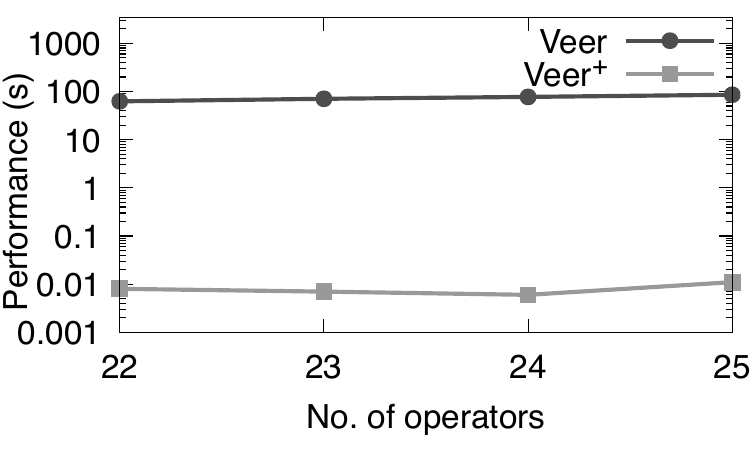}
         \caption{Running time.}
         \label{fig:size-time}
\end{subfigure}
\caption{Effect of the number of operators (W2 with two edits).}
         \end{figure}
                  
 \boldstart{Varying the number of unsupported operators.}
Both the baseline and \sysOpt were not affected by the increase in the number of unsupported operators as these operators were not included in the covering windows.

\sadeem{\subsection{Limitations of \sysName}
\sysName could verify complex dataflow version pairs that other verifiers were unable to verify. However, given the undecidability of the problem of determining the equivalence of two dataflow versions, there are cases where \sysName fails to verify due to the following reasons:

\textit{(1) Determinism and context:}
\sysName focuses on a small portion of the dataflow pair (windows) and ignores the context. It treats the input to the small windows as any instance of input sources. In some cases, these windows may be inequivalent but the entire pair is. Moreover, \sysName assumes that the input sources are not changing and that the operator functions are deterministic across different versions.

\textit{(2) Dependence on EV:}
\sysName is a general-purpose framework that internally depends on existing EVs. When the changes between the versions are performed in a UDF operator, then \sysName would need to rely on an EV that can reason about the semantics of a UDF.
We will address these limitations in a followup work~\cite{conf/hilda/AlsudaisK023}.
}

%% file: sec1-related-works.tex
\section{Related Works}
\hspace{\parindent}

\emph{Equivalence verification.}
There are many studies to solve the problem of verifying the equivalence of two SQL queries under certain assumptions. These solutions were applicable to a small class of SQL queries, such as conjunctive queries~\cite{conf/stoc/ChandraM77,conf/edbt/AfratiLM04,journals/jacm/SagivY80,conf/pods/JayramKV06}. With the recent advancement of developing proof assists and solvers~\cite{conf/tacas/MouraB08,conf/cade/MouraKADR15}, there have been new solutions~\cite{journals/pvldb/ChuMRCS18,journals/pvldb/ZhouANHX19,journals/corr/abs-2004-00481}. UDP~\cite{journals/pvldb/ChuMRCS18} and WeTune's verifier~\cite{conf/sigmod/WangZYDHDT0022} use semirings to model the semantics of a pair of queries and use a proof assist, such as Lean~\cite{conf/cade/MouraKADR15} to prove if the expressions are equivalent. These two works support reasoning semantics of two queries with integrity constraints. Equitas~\cite{journals/pvldb/ZhouANHX19} and Spes~\cite{journals/corr/abs-2004-00481} model the semantics of the pair into a First-Order Logic (FOL) formula and push the formula to be solved by a solver such as SMT~\cite{conf/tacas/MouraB08}. These two works support queries with three-valued variables. Other works also use an SMT solver to verify the equivalence of a pair of Spark jobs~\cite{conf/cav/GrossmanCIRS17}. Our solution  uses them as black boxes to verify the equivalence of a version pair. 
The work in~\cite{journals/debu/Chandra022} finds a weighted edit distance based on the semantic equivalence of two queries to grade students queries.

\emph{Tracking dataflow executions.} There has been an increasing interest in enabling the reproducibility of data analytics pipelines. These tools track the evolution and versioning of datasets, models, and results.  At a high level they can be classified as two categories. The first includes those that track experiment results of different versions of  ML models and the corresponding hyper-parameters~\cite{conf/icse/ZhangXFWYX16, conf/sigmod/ChenCDD0HKMMNOP20, conf/sigmod/VartakSLVHMZ16, klaus_greff-proc-scipy-2017, conf/sigmod/GharibiWARL19, conf/icde/MiaoLDD17}.
The second includes solutions to track results of different versions of data processing dataflows~\cite{journals/debu/0001D18, conf/sc/WoodmanHWM11, databricks-tracking,conf/vldb/Alsudais22,journals/pvldb/PimentelMBF17,conf/icde/CallahanFFSSSV06}. 
These solutions are motivations for our work.

\emph{Materialization reuse and MQO.} 
There is a large body of work on answering data processing dataflows using views~\cite{journals/pvldb/RoyJGOGRMJ21,conf/cloud/RamjitIWN19,conf/sigmod/DursunBCK17,conf/sigmod/DerakhshanMKRM22,journals/pvldb/JindalKRP18}.
Some solutions~\cite{journals/pvldb/ElghandourA12} focus on deciding which results to store to maximize future reuse. Other solutions~\cite{conf/icde/NagelBV13, conf/sigmod/ZhouLFL07} focus on identifying materialization reuse opportunities by relying on finding an exact match of the dataflow's DAG. 
On the other hand, semantic query optimization works~\cite{journals/vldb/Halevy01, journals/tcs/FaginKMP05, conf/icdt/Schmidt0L10, journals/vldb/KossmannPN22,conf/edbt/ChaudharyZMK23} reason the semantics of the query to identify reuse opportunities that are not limited to structural matching. However, these solutions are applicable to a specific class of functions, such as user defined function (UDF)~\cite{conf/cloud/RamjitIWN19, conf/sigmod/XuKAR22,sparkudf:website}, and do not generalize to finding reuse opportunities by finding equivalence of any pair of dataflows.

%% file: sec10-conclusion.tex
\section{Conclusion}
In this paper, we studied the problem of verifying the equivalence of two dataflow versions. 
We presented a solution called ``\sysName,'' which leverages the fact that two workflow versions can be very similar except for a few changes. 
We analyzed the restrictions of existing EVs and presented a concept called a ``window'' to leverage the existing solutions for verifying the equivalence.
We proposed a solution using the windows to verify the equivalence of a version pair with a single edit.
We discussed the challenges of verifying a version pair with multiple edits and proposed a baseline algorithm.  We proposed optimization techniques to speed up the performance of the baseline.
We conducted a thorough experimental study and showed the high efficiency and effectiveness of the solution.

\balance

%% file: appendix.tex
\appendix

\section{Sample of the Workflows used in the Experiments} \label{sec:appendix-b}
The following Figure~\ref{fig:workflow-stats} shows the details of a real workload collected from one deployment of Texera~\cite{texera}.

 \begin{figure*}[htbp]
         \centering
\includegraphics[width=\linewidth]{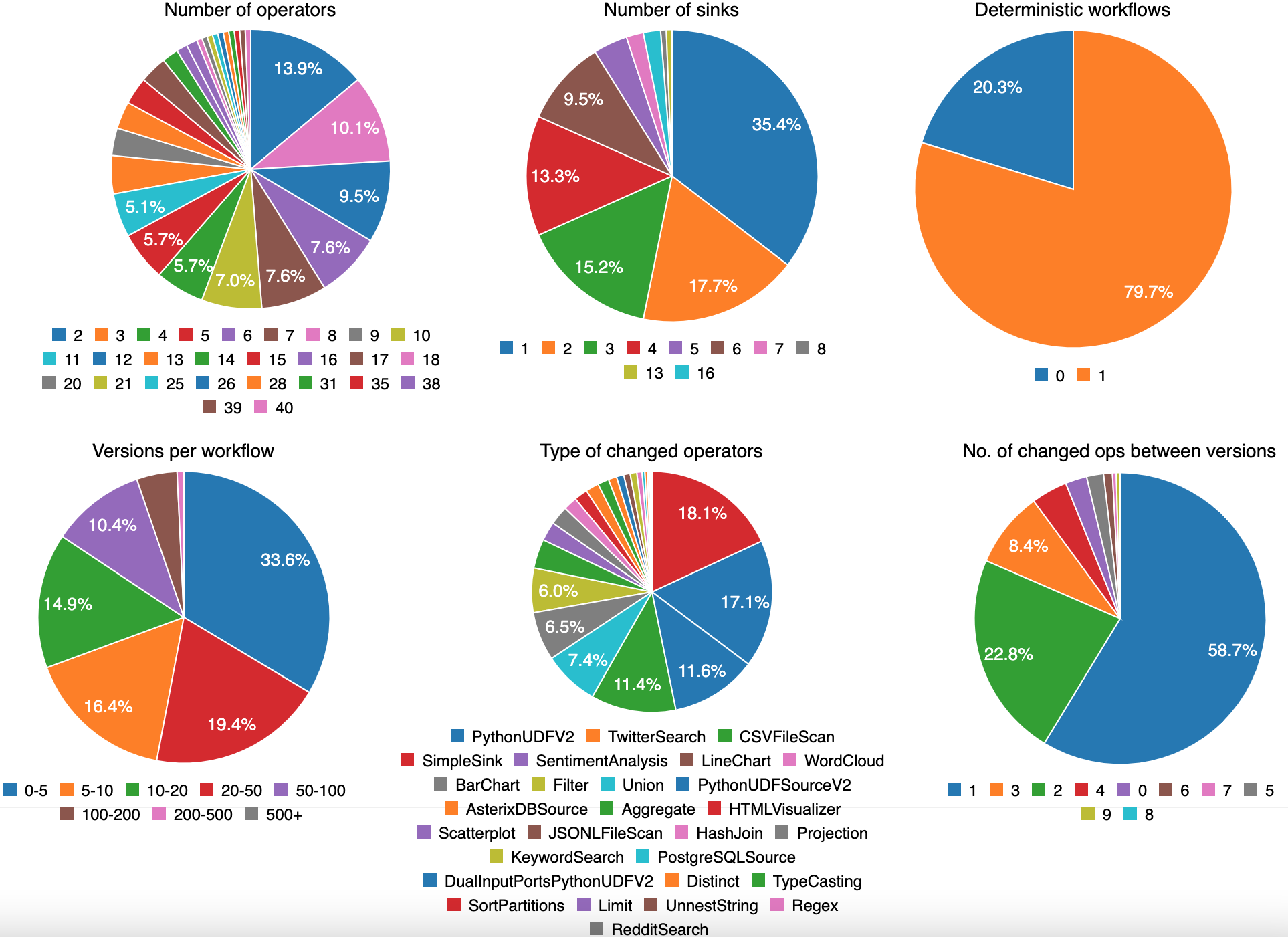}
  \caption{Details of workflow from one real workload}
           \label{fig:workflow-stats}
\end{figure*}

Figures~\ref{fig:wkflow}-~\ref{fig:wkflow6} show samples of the initial version of the workflows used in the experiments. The UDF operator included in the transformed TPC-DS queries to Texera workflows, is to include the logic of Order By, because Sort operator was not part of the Texera system at the time the workflow was constructed. The workflows include in some cases a sequence of filter operators, because at the time the workflows were constructed, Texera did not support the {\sf AND} operation to join multiple predicate conditions. 

   \begin{figure*}[htbp]
         \centering
\includegraphics[width=\linewidth]{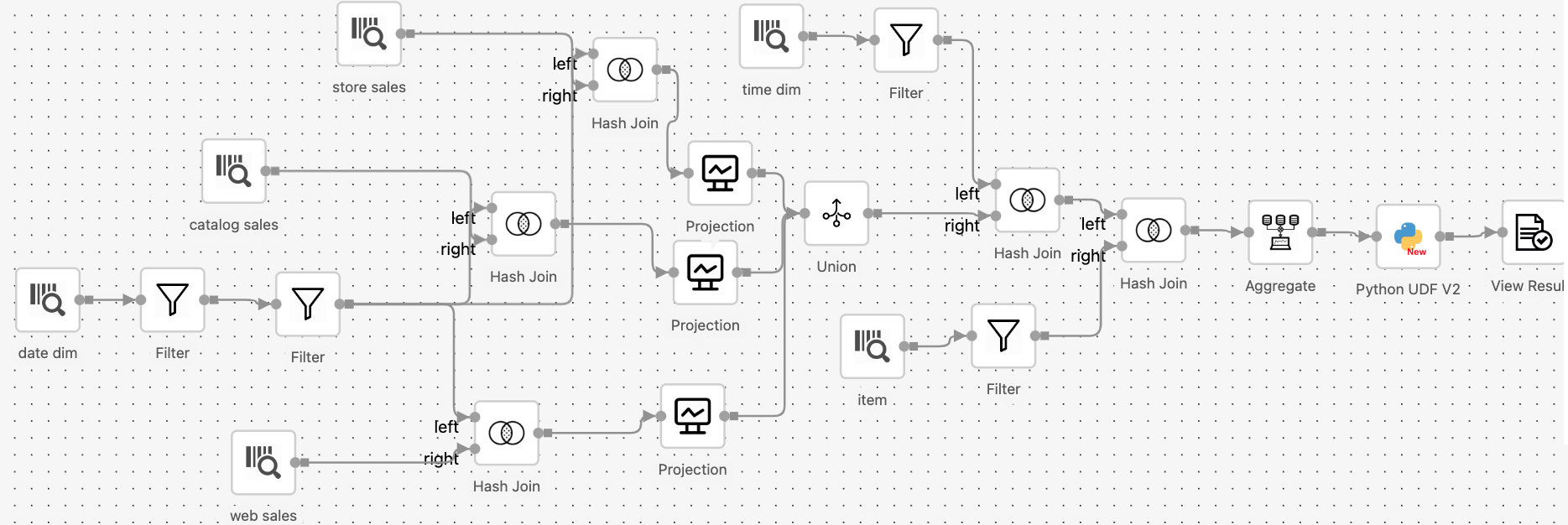}
  \caption{Sample workflow of TPC-DS Q71 in Texera.}
           \label{fig:wkflow}
\end{figure*}

   \begin{figure*}[htbp]
         \centering
\includegraphics[width=\linewidth]{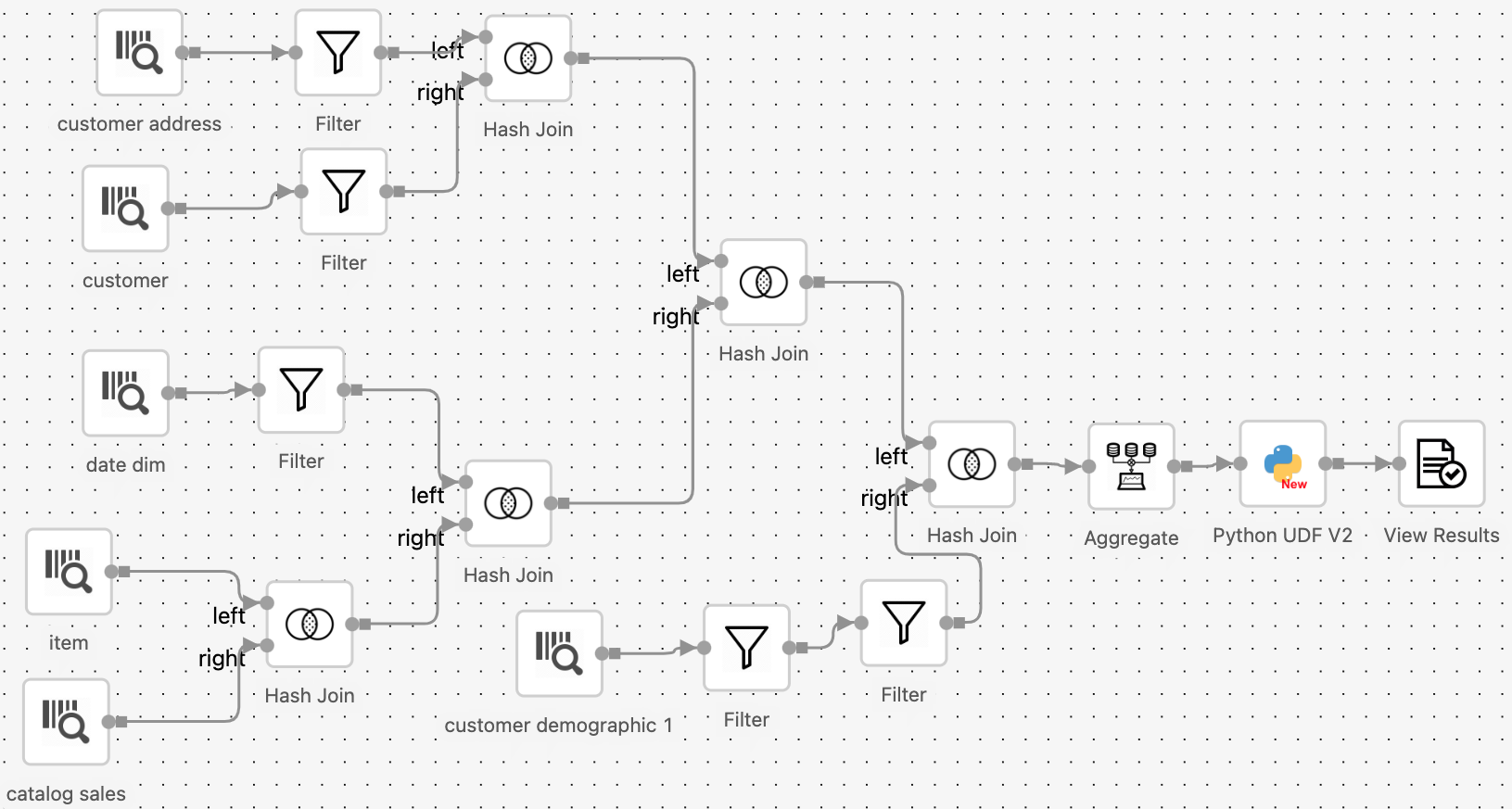}
  \caption{Sample workflow of TPC-DS Q18 in Texera.}
           \label{fig:wkflow1}
\end{figure*}

   \begin{figure*}[htbp]
         \centering
\includegraphics[width=\linewidth]{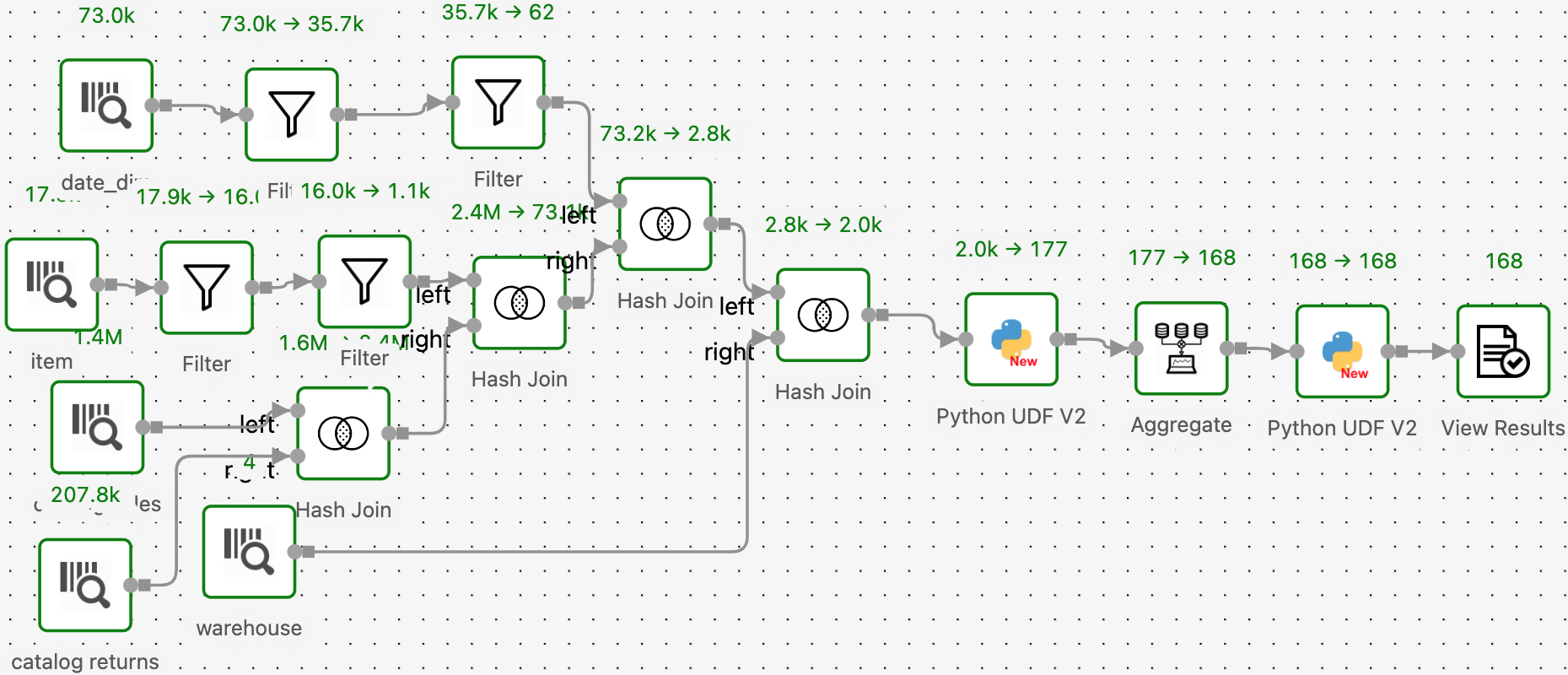}
  \caption{Sample workflow of TPC-DS Q40 in Texera.}
           \label{fig:wkflow2}
\end{figure*}

   \begin{figure*}[htbp]
         \centering
\includegraphics[width=\linewidth]{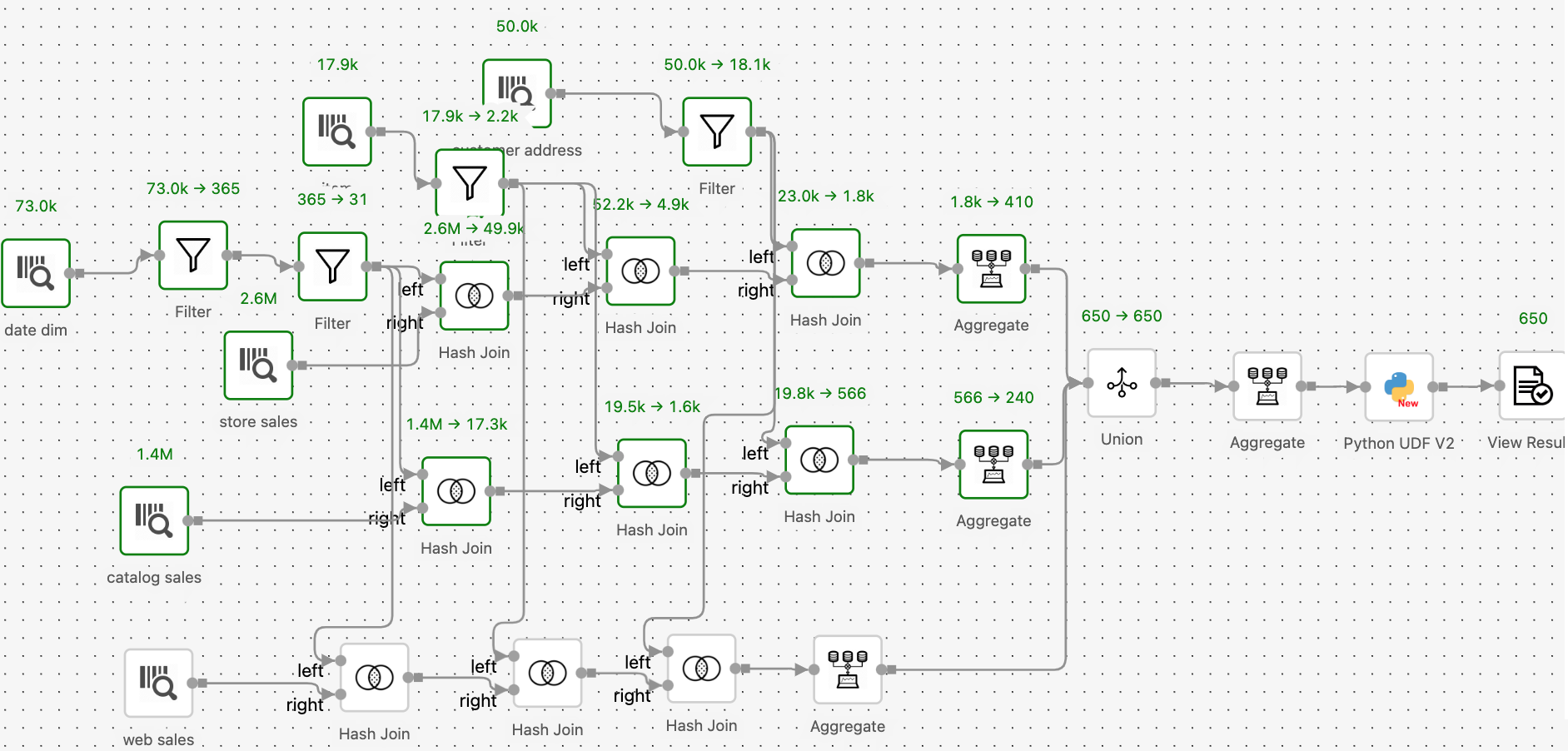}
  \caption{Sample workflow of TPC-DS Q33 in Texera.}
           \label{fig:wkflow3}
\end{figure*}

   \begin{figure*}[htbp]
         \centering
\includegraphics[width=\linewidth]{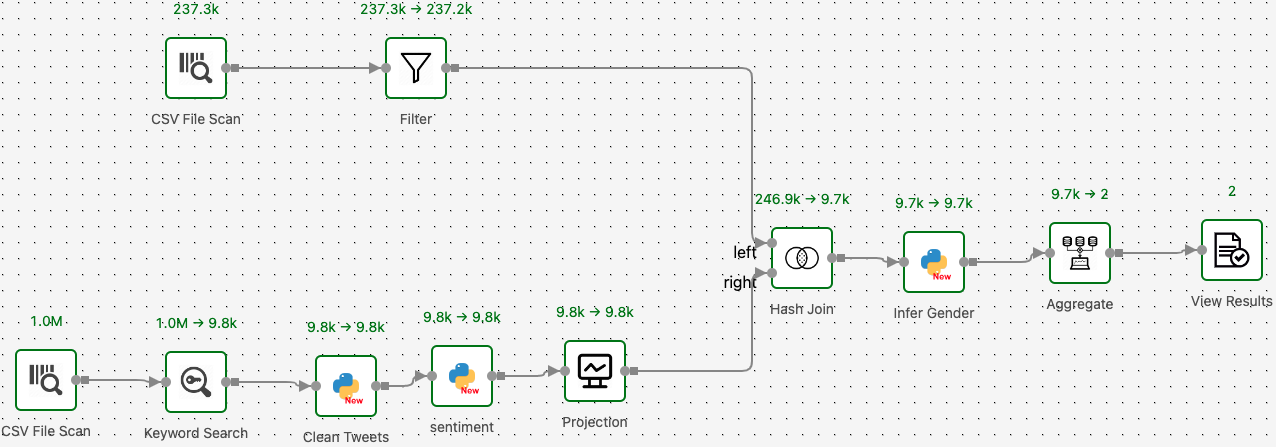}
  \caption{Sample workflow of Tweets to infer the gender of the Tweeter.}
           \label{fig:wkflow4}
\end{figure*}

   \begin{figure*}[htbp]
         \centering
\includegraphics[width=\linewidth]{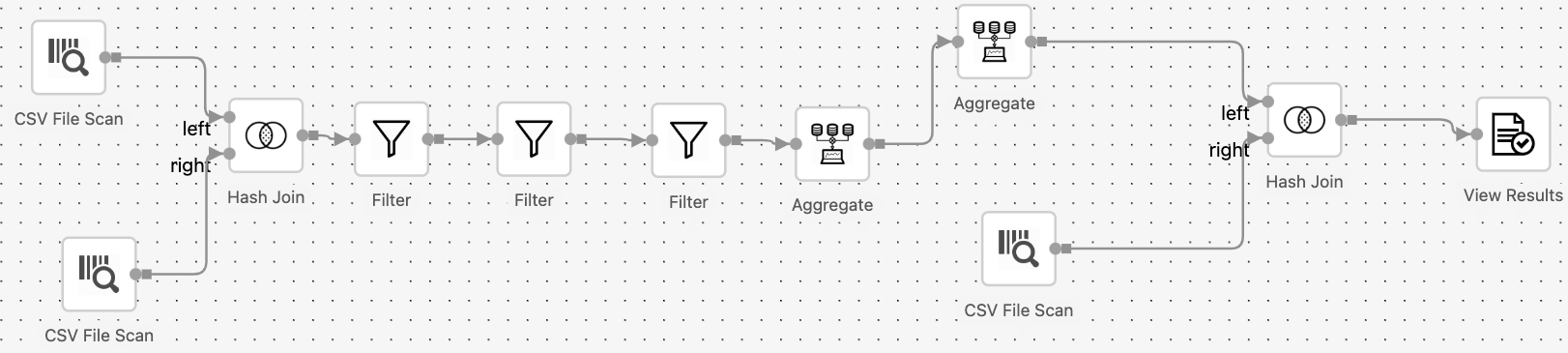}
  \caption{Sample workflow of IMDB movies to calculate the ratio of original to non-original movies.}
           \label{fig:wkflow5}
\end{figure*}

   \begin{figure*}[htbp]
         \centering
\includegraphics[width=\linewidth]{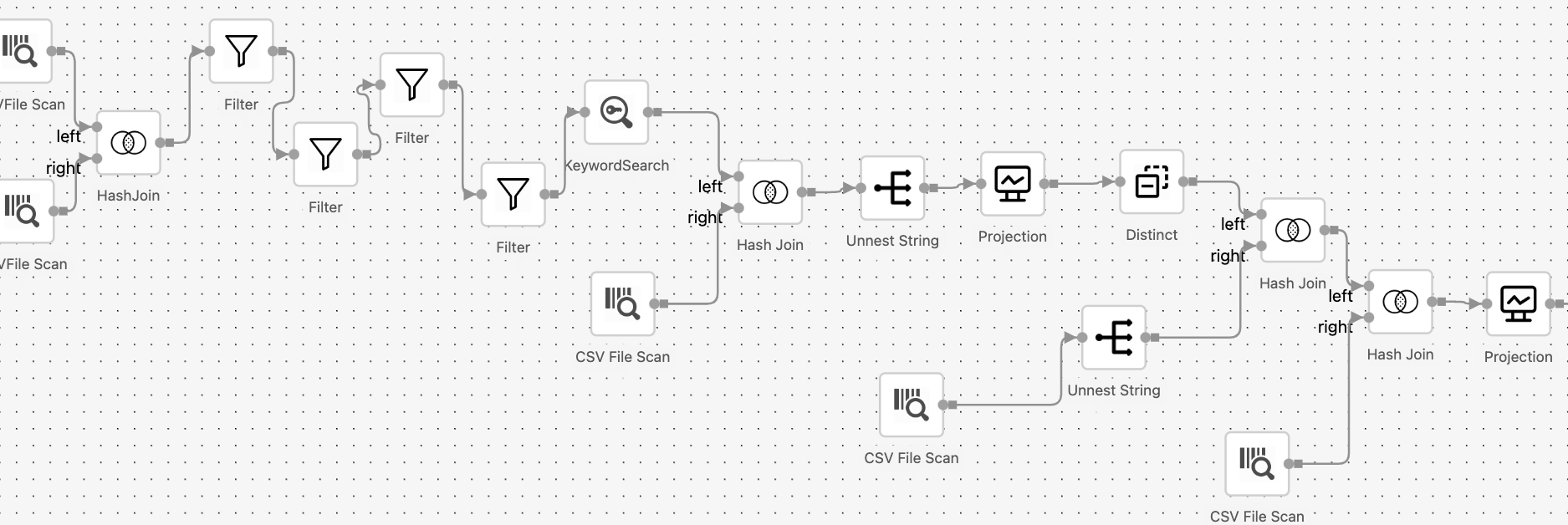}
  \caption{A portion of a workflow to get all movies of directors that have certain criteria.}
           \label{fig:wkflow6}
\end{figure*}

\section{Details of the Workflows Used to Analyze their Complexity} \label{sec:appendix-c}

Table~\ref{tbl:workflows-alteryx} details the Alteryx~\cite{alteryx} workflows that were collected to evaluate their complexity. Alteryx has many sample workflows when installing the desktop version. We chose some of the workflows that are publicly available to give examples to the readers.

\begin{table*}[tbp]
\caption{Sample Alteryx workflows. \label{tbl:workflows-alteryx}}
\begin{adjustbox}{width=\linewidth,center}
\begin{tabular}{|p{6.7cm}|p{9.7cm}|} \hline
 \textbf{Use case name} & \textbf{Reference} \\
\hline \hline
 Establishing a Reusable ETL Workflow for Non-Profits & \url{https://community.alteryx.com/t5/Maveryx-Success-Stories/Establishing-a-Reusable-ETL-Workflow-for-Non-Profits/ta-p/182871} \\ \hline
  Basic In-DB Configuration for Impala and Spark & \url{https://community.alteryx.com/t5/Engine-Works/Tackle-your-Big-Data-Use-Cases-with-Alteryx-in-Database/ba-p/4527} \\  \hline
 Automated Billing Reconciliation Using the Image to Text Tool & \url{https://community.alteryx.com/t5/Data-Science/Automated-Billing-Reconciliation-Using-the-Image-to-Text-Tool/ba-p/1051772} \\ \hline
 UK Tool Centre Modernizes Multiple Business Departments & \url{https://community.alteryx.com/t5/Maveryx-Success-Stories/UK-Tool-Centre-Modernizes-Multiple-Business-Departments/ta-p/433268} \\ \hline
Introducing the Workflow Summary Tool, Powered by Generative AI & \url{https://community.alteryx.com/t5/Data-Science/Introducing-the-Workflow-Summary-Tool-Powered-by-Generative-AI/ba-p/1122275} \\ \hline
Sales Pipeline Designer Workflow  & \url{https://www.alteryx.com/resources/use-case/sales-pipeline-consolidation-and-reporting} \\  \hline
Applicant Reporting Live & \url{https://hkrtrainings.com/alteryx-workflow} \\ \hline
 Office of Finance & \url{https://www.alteryx.com/resources/use-case/reconciliation-of-financial-systems} \\  \hline
\end{tabular}
\end{adjustbox}
\end{table*}

\section{Workflows Used in Evaluation of Existing EVs} \label{sec:appendix-a}
The workload of Texera~\cite{texera} belongs to real customers and cannot be made publicly available. Table~\ref{tbl:workflows} details the Knime~\cite{knimeworkflows:website} workflows that were collected to evaluate the support of existing EVs in verifying the equivalence of these workflows. Knime had 18,663 workflows at the time of collection. We examined the top 37 workflows.

\begin{table*}[b]
\caption{Knime workflows used in the evaluation. \label{tbl:workflows}}
\begin{adjustbox}{width=\linewidth,center}
\begin{tabular}{|p{6.7cm}|p{9.7cm}|} \hline
 \textbf{Workflow name} & \textbf{Reference} \\
\hline \hline
 The Machine Learning Canvas in KNIME & \url{https://hub.knime.com/ali_alkan/spaces/Data%20and%20Analytics%20Strategy%20for%20Business/Machine%20Learning%20Canvas%20in%20KNIME~e_J8HACxUdPLXaEY/current-state} \\ \hline
  REST API for Sentiment Analysis & \url{https://hub.knime.com/jtyler/spaces/Public/Sentiment_Predictor~gXyySX5ZAa9QzDnv/current-state} \\  \hline
 Automated Reporting of Receivables & \url{https://hub.knime.com/knime/spaces/Finance,%20Accounting,%20and%20Audit/Automated%20Reporting%20of%20Receivables~Ptlrd84tYrTMNtud/current-state} \\ \hline
 TweetKollidR & \url{https://hub.knime.com/angusveitch/spaces/Public/TweetKollidR~vy4F4J-RKsG_bjDR/current-state} \\ \hline
TextKleaner & \url{https://hub.knime.com/angusveitch/spaces/Public/TextKleaner~vw0_xvNhfjw3hPM-/current-state} \\ \hline
Convert XLS to SDF  & \url{https://hub.knime.com/swebb/spaces/Public/Demos/Convert%20XLS%20to%20SDF~ak1sOw7TwRPWWBB0/current-state} \\  \hline
Reproducible (Minimal) Workflow Example Template & \url{https://hub.knime.com/ipazin/spaces/Public/Reproducible%20(Minimal)%20Workflow%20Example%20Template~9a4RPB0FMAVuN7OQ/current-state} \\ \hline
 AutoML Component via Interactive Views & \url{https://hub.knime.com/knime/spaces/Examples/09_Enterprise/04_Integrated_Deployment/03_AutoML_Deployment/01_AutoML_Component_via_Interactive_Views~ZBbj2fObfRwrp8ht/most-recent} \\  \hline
COVID-19 Live Visualization using Guided Analytics & \url{https://hub.knime.com/paolotamag/spaces/Public/COVID-19_Live_Visualization~SjvvBM2fXG3APLP-/most-recent} \\  \hline
 Read Data from Google Sheets and Google Drive & \url{https://hub.knime.com/knime/spaces/Beginners%20Space/01_Read/03_Read_Data_from_Google_Drive~JKP_DPzsa9BQ0k-9/current-state} \\  \hline
 Sentiment Analysis with BERT & \url{https://hub.knime.com/knime/spaces/Examples/04_Analytics/14_Deep_Learning/04_TensorFlow2/01_BERT_Sentiment_Analysis~CKMYTMCoD5OJ0LH_/most-recent} \\ \hline
Transform Data using GroupBy and Joiner nodes & \url{https://hub.knime.com/knime/spaces/Beginners%20Space/03_Transform/01_Transform_Using_Group_By_and_Joiner_nodes~4LoOhUEVLc4PqLfI/current-state} \\  \hline
 Read Data from Amazon S3 & \url{https://hub.knime.com/knime/spaces/Beginners%20Space/01_Read/04_Read_Data_from_Amazon_S3~3QR26_e1-7sXIliB/current-state} \\ \hline 
Read Data from Microsoft Azure Cloud & \url{https://hub.knime.com/knime/spaces/Beginners%20Space/01_Read/05_Read_Data_from_Microsoft_Azure_Cloud~7a381TgHiTEVMe52/current-state}
 \\  \hline
 TeachOpenCADD - a teaching platform for computer-aided drug design using KNIME & \url{https://hub.knime.com/volkamerlab/spaces/Public/TeachOpenCADD/TeachOpenCADD~xYhrR1mfFcGNxz7I/current-state} \\ \hline
 Transform Data using Pivoting node & \url{https://hub.knime.com/knime/spaces/Beginners%20Space/03_Transform/02_Transform_Using_the_Pivoting_Node~XZ430qpw8_PFIPuS/current-state} \\  \hline
Transform Data using Rule Engine and String Manipulation nodes & \url{https://hub.knime.com/knime/spaces/Beginners%20Space/03_Transform/03_Transform_Using_Rule_Engine_and_String_Manipulation_Node~LKvg8ff1sJAU-p14/current-state} \\ \hline
 Transform Data using Math Formula and Cell Splitter nodes & \url{https://hub.knime.com/knime/spaces/Beginners%20Space/03_Transform/04_Transform_Using_Math_formula_and_Cell_Splitter_node~G3KROtVim9UnBVW2/current-state} \\  \hline
ChatGPT as a KNIME chat dashboard & \url{https://hub.knime.com/roberto_cadili/spaces/Public/ChatGPT%20as%20a%20KNIME%20chat%20dashboard~L1Wc6IcyxchsP1SB/current-state} \\  \hline
Various Examples of Regex in KNIME & \url{https://hub.knime.com/victor_palacios/spaces/Public/Various%20Examples%20of%20Regex%20in%20KNIME~J1anMyaS7MMiRV29/current-state} \\ \hline
 Budget Monitoring Report & \url{https://hub.knime.com/knime/spaces/Finance,%20Accounting,%20and%20Audit/Budget%20Monitoring%20Report~_jd6zP-aXbDUQZ8q/current-state} \\ \hline
\end{tabular}
\end{adjustbox}
\end{table*}

\begin{table*}[b]
\begin{adjustbox}{width=\linewidth,center}
\begin{tabular}{|p{6.7cm}|p{9.7cm}|} \hline
Analyze Data by Training a k-Means Clustering on Location Data & \url{https://hub.knime.com/knime/spaces/Beginners%20Space/04_Analyze/03_Analyze_Clustering_location_data~gZluj12KJvYshWKz/current-state} \\  \hline
Looping over all columns and manipulation of each & \url{https://hub.knime.com/knime/spaces/Examples/06_Control_Structures/04_Loops/03_Looping_over_all_columns_and_manipulation_of_each~Ck4Q3ZW7PEUAyl4L/most-recent} \\  \hline
 Multivariate Time Series Analysis with an RNN - Training & \url{https://hub.knime.com/kathrin/spaces/Multivariate%20Times%20Series%20with%20RNN/Multivariate_Time_Series_RNN_Keras_Training~B45XEOAuWeQBzO9b/current-state} \\  \hline
Sentiment Analysis (Classification) of Documents & \url{https://hub.knime.com/knime/spaces/Examples/08_Other_Analytics_Types/01_Text_Processing/03_Sentiment_Classification~ZHAExldZ5M7q6hdG/most-recent} \\  \hline
 Explore Sunburst and Stacked Area Chart & \url{https://hub.knime.com/knime/spaces/Beginners%20Space/02_Explore/04_Explore_Sunburst_and_Stacked_Area_Chart~Ywx0qSLVIymG5M9d/current-state} \\  \hline
 Explore Scatter Plot and Box Plot & \url{https://hub.knime.com/knime/spaces/Beginners%20Space/02_Explore/01_Explore_Scatter_Plot_and_Box_Plot~7W5Q1ZAUwH9nyI2o/current-state} \\ \hline
 Explore Donut Chart and Data Explorer & \url{https://hub.knime.com/knime/spaces/Beginners%20Space/02_Explore/02_Explore_Donut_Chart_and%20_Data_Explorer~V04Ul-Nb7HAaBEEO/current-state} \\ \hline
Explore Bar Chart and Line Plot & \url{https://hub.knime.com/knime/spaces/Beginners%20Space/02_Explore/03_Explore_Bar_Chart_and_Line_Plot~bfRWPtdBOMx5SHea/current-state} \\ \hline
 ML Prototyping for Bioactivity Data & \url{https://hub.knime.com/knime/spaces/Life%20Sciences/Events/2021_03_03_Cheminformatics_with_KNIME_Webinar/ML%20Prototyping%20for%20Bioactivity%20Data~ztcsrat0YL_Pgz_q/current-state} \\ \hline
Performing a k-Means clustering & \url{https://hub.knime.com/knime/spaces/Examples/04_Analytics/03_Clustering/01_Performing_a_k-Means_Clustering~QqKlIQmDnCAg0lEQ/most-recent} \\  \hline
 Recursive replacement in strings based on a replacement dictionary & \url{https://hub.knime.com/knime/spaces/Examples/06_Control_Structures/04_Loops/08_Example_for_Recursive_Replacement_of_Strings~MjiD7VtF32L8KI8X/current-state} \\  \hline
Challenge 1 - Keep Marge Simpson & \url{https://hub.knime.com/rs1/spaces/Public/weekly_challenges/Challenge1-KeepSimpsonMarge~at1e_xXCV22tmTW3/current-state} \\  \hline
 Guided Labeling for Document Classification & \url{https://hub.knime.com/knime/spaces/Examples/50_Applications/57_Guided_Labeling/01_Guided_Labeling_for_Document_Classification~y5nhpbd1PP5F4WKH/current-state} \\ \hline
 Guided Automation & \url{https://hub.knime.com/knime/spaces/Examples/50_Applications/36_Guided_Analytics_for_ML_Automation/01_Guided_Analytics_for_ML_Automation/01_Guided_Analytics_for_ML_Automation~eAGfGtEAIr-1iYR-/current-state} \\ \hline
\end{tabular}
\end{adjustbox}
\end{table*}


%% file: main.bbl

\begin{thebibliography}{59}


\ifx \showCODEN    \undefined \def \showCODEN     #1{\unskip}     \fi
\ifx \showDOI      \undefined \def \showDOI       #1{#1}\fi
\ifx \showISBNx    \undefined \def \showISBNx     #1{\unskip}     \fi
\ifx \showISBNxiii \undefined \def \showISBNxiii  #1{\unskip}     \fi
\ifx \showISSN     \undefined \def \showISSN      #1{\unskip}     \fi
\ifx \showLCCN     \undefined \def \showLCCN      #1{\unskip}     \fi
\ifx \shownote     \undefined \def \shownote      #1{#1}          \fi
\ifx \showarticletitle \undefined \def \showarticletitle #1{#1}   \fi
\ifx \showURL      \undefined \def \showURL       {\relax}        \fi
\providecommand\bibfield[2]{#2}
\providecommand\bibinfo[2]{#2}
\providecommand\natexlab[1]{#1}
\providecommand\showeprint[2][]{arXiv:#2}

\bibitem[\protect\citeauthoryear{Abiteboul, Hull, and Vianu}{Abiteboul et~al\mbox{.}}{1995}]%
        {10.5555/551350}
\bibfield{author}{\bibinfo{person}{Serge Abiteboul}, \bibinfo{person}{Richard Hull}, {and} \bibinfo{person}{Victor Vianu}.} \bibinfo{year}{1995}\natexlab{}.
\newblock \bibinfo{booktitle}{\emph{Foundations of Databases: The Logical Level} (\bibinfo{edition}{1st} ed.)}.
\newblock \bibinfo{publisher}{Addison-Wesley Longman Publishing Co., Inc.}, \bibinfo{address}{USA}.
\newblock
\showISBNx{0201537710}


\bibitem[\protect\citeauthoryear{Afrati, Li, and Mitra}{Afrati et~al\mbox{.}}{2004}]%
        {conf/edbt/AfratiLM04}
\bibfield{author}{\bibinfo{person}{Foto~N. Afrati}, \bibinfo{person}{Chen Li}, {and} \bibinfo{person}{Prasenjit Mitra}.} \bibinfo{year}{2004}\natexlab{}.
\newblock \showarticletitle{On Containment of Conjunctive Queries with Arithmetic Comparisons}. In \bibinfo{booktitle}{\emph{EDBT}}. \bibinfo{pages}{459--476}.
\newblock


\bibitem[\protect\citeauthoryear{Ahmad, Kennedy, Koch, and Nikolic}{Ahmad et~al\mbox{.}}{2012}]%
        {journals/pvldb/AhmadKKN12}
\bibfield{author}{\bibinfo{person}{Yanif Ahmad}, \bibinfo{person}{Oliver Kennedy}, \bibinfo{person}{Christoph Koch}, {and} \bibinfo{person}{Milos Nikolic}.} \bibinfo{year}{2012}\natexlab{}.
\newblock \showarticletitle{DBToaster: Higher-order Delta Processing for Dynamic, Frequently Fresh Views}.
\newblock \bibinfo{journal}{\emph{Proc. {VLDB} Endow.}} \bibinfo{volume}{5}, \bibinfo{number}{10} (\bibinfo{year}{2012}), \bibinfo{pages}{968--979}.
\newblock
\urldef\tempurl%
\url{https://doi.org/10.14778/2336664.2336670}
\showDOI{\tempurl}


\bibitem[\protect\citeauthoryear{Alsudais}{Alsudais}{2022}]%
        {conf/vldb/Alsudais22}
\bibfield{author}{\bibinfo{person}{Sadeem Alsudais}.} \bibinfo{year}{2022}\natexlab{}.
\newblock \showarticletitle{Drove: Tracking Execution Results of Workflows on Large Data}. In \bibinfo{booktitle}{\emph{Proceedings of the {VLDB} 2022 PhD Workshop co-located with the 48th International Conference on Very Large Databases {(VLDB} 2022), Sydney, Australia, September 5, 2022}} \emph{(\bibinfo{series}{{CEUR} Workshop Proceedings})}, \bibfield{editor}{\bibinfo{person}{Zhifeng Bao} {and} \bibinfo{person}{Timos~K. Sellis}} (Eds.), Vol.~\bibinfo{volume}{3186}. \bibinfo{publisher}{CEUR-WS.org}.
\newblock
\urldef\tempurl%
\url{http://ceur-ws.org/Vol-3186/paper\_10.pdf}
\showURL{%
\tempurl}


\bibitem[\protect\citeauthoryear{Alsudais, Kumar, and Li}{Alsudais et~al\mbox{.}}{2023}]%
        {conf/hilda/AlsudaisK023}
\bibfield{author}{\bibinfo{person}{Sadeem Alsudais}, \bibinfo{person}{Avinash Kumar}, {and} \bibinfo{person}{Chen Li}.} \bibinfo{year}{2023}\natexlab{}.
\newblock \showarticletitle{Raven: Accelerating Execution of Iterative Data Analytics by Reusing Results of Previous Equivalent Versions}. In \bibinfo{booktitle}{\emph{Proceedings of the Workshop on Human-In-the-Loop Data Analytics, {HILDA} 2023, Seattle, WA, USA, 18 June 2023}}. \bibinfo{publisher}{{ACM}}, \bibinfo{pages}{3:1--3:7}.
\newblock
\urldef\tempurl%
\url{https://doi.org/10.1145/3597465.3605219}
\showDOI{\tempurl}


\bibitem[\protect\citeauthoryear{Alteryx}{Alteryx}{[n.d.]}]%
        {alteryx}

\newblock
\newblock
\shownote{Alteryx Website, \url{https://www.alteryx.com/}.}


\bibitem[\protect\citeauthoryear{APACHE Flink}{APACHE Flink}{[n.d.]}]%
        {misc/flink}

\newblock
\newblock
\shownote{Apache Flink http://flink.apache.org.}


\bibitem[\protect\citeauthoryear{Borralleras, Larraz, Rodr{\'{\i}}guez{-}Carbonell, Oliveras, and Rubio}{Borralleras et~al\mbox{.}}{2019}]%
        {journals/tocl/BorrallerasLROR19}
\bibfield{author}{\bibinfo{person}{Cristina Borralleras}, \bibinfo{person}{Daniel Larraz}, \bibinfo{person}{Enric Rodr{\'{\i}}guez{-}Carbonell}, \bibinfo{person}{Albert Oliveras}, {and} \bibinfo{person}{Albert Rubio}.} \bibinfo{year}{2019}\natexlab{}.
\newblock \showarticletitle{Incomplete {SMT} Techniques for Solving Non-Linear Formulas over the Integers}.
\newblock \bibinfo{journal}{\emph{{ACM} Trans. Comput. Log.}} \bibinfo{volume}{20}, \bibinfo{number}{4} (\bibinfo{year}{2019}), \bibinfo{pages}{25:1--25:36}.
\newblock
\urldef\tempurl%
\url{https://doi.org/10.1145/3340923}
\showDOI{\tempurl}


\bibitem[\protect\citeauthoryear{Calcite}{Calcite}{[n.d.]}]%
        {calcitebenchmark:website}

\newblock
\newblock
\shownote{Calcite benchmark, \url{https://github.com/uwdb/Cosette/tree/master/examples/calcite}.}


\bibitem[\protect\citeauthoryear{Callahan, Freire, Santos, Scheidegger, Silva, and Vo}{Callahan et~al\mbox{.}}{2006}]%
        {conf/icde/CallahanFFSSSV06}
\bibfield{author}{\bibinfo{person}{Steven~P. Callahan}, \bibinfo{person}{Juliana Freire}, \bibinfo{person}{Emanuele Santos}, \bibinfo{person}{Carlos~Eduardo Scheidegger}, \bibinfo{person}{Cl{\'{a}}udio~T. Silva}, {and} \bibinfo{person}{Huy~T. Vo}.} \bibinfo{year}{2006}\natexlab{}.
\newblock \showarticletitle{Managing the Evolution of Dataflows with VisTrails}. In \bibinfo{booktitle}{\emph{Proceedings of the 22nd International Conference on Data Engineering Workshops, {ICDE} 2006, 3-7 April 2006, Atlanta, GA, {USA}}}. \bibinfo{publisher}{{IEEE} Computer Society}, \bibinfo{pages}{71}.
\newblock
\urldef\tempurl%
\url{https://doi.org/10.1109/ICDEW.2006.75}
\showDOI{\tempurl}


\bibitem[\protect\citeauthoryear{Chandra and Merlin}{Chandra and Merlin}{1977}]%
        {conf/stoc/ChandraM77}
\bibfield{author}{\bibinfo{person}{Ashok~K. Chandra} {and} \bibinfo{person}{Philip~M. Merlin}.} \bibinfo{year}{1977}\natexlab{}.
\newblock \showarticletitle{Optimal Implementation of Conjunctive Queries in Relational Data Bases}. In \bibinfo{booktitle}{\emph{Proceedings of the 9th Annual {ACM} Symposium on Theory of Computing, May 4-6, 1977, Boulder, Colorado, {USA}}}, \bibfield{editor}{\bibinfo{person}{John~E. Hopcroft}, \bibinfo{person}{Emily~P. Friedman}, {and} \bibinfo{person}{Michael~A. Harrison}} (Eds.). \bibinfo{publisher}{{ACM}}, \bibinfo{pages}{77--90}.
\newblock
\urldef\tempurl%
\url{https://doi.org/10.1145/800105.803397}
\showDOI{\tempurl}


\bibitem[\protect\citeauthoryear{Chandra and Sudarshan}{Chandra and Sudarshan}{2022}]%
        {journals/debu/Chandra022}
\bibfield{author}{\bibinfo{person}{Bikash Chandra} {and} \bibinfo{person}{S. Sudarshan}.} \bibinfo{year}{2022}\natexlab{}.
\newblock \showarticletitle{Automated Grading of {SQL} Queries}.
\newblock \bibinfo{journal}{\emph{{IEEE} Data Eng. Bull.}} \bibinfo{volume}{45}, \bibinfo{number}{3} (\bibinfo{year}{2022}), \bibinfo{pages}{17--28}.
\newblock
\urldef\tempurl%
\url{http://sites.computer.org/debull/A22sept/p17.pdf}
\showURL{%
\tempurl}


\bibitem[\protect\citeauthoryear{Chaudhary, Zeuch, Markl, and Karimov}{Chaudhary et~al\mbox{.}}{2023}]%
        {conf/edbt/ChaudharyZMK23}
\bibfield{author}{\bibinfo{person}{Ankit Chaudhary}, \bibinfo{person}{Steffen Zeuch}, \bibinfo{person}{Volker Markl}, {and} \bibinfo{person}{Jeyhun Karimov}.} \bibinfo{year}{2023}\natexlab{}.
\newblock \showarticletitle{Incremental Stream Query Merging}. In \bibinfo{booktitle}{\emph{Proceedings 26th International Conference on Extending Database Technology, {EDBT} 2023, Ioannina, Greece, March 28-31, 2023}}. \bibinfo{publisher}{OpenProceedings.org}, \bibinfo{pages}{604--617}.
\newblock
\urldef\tempurl%
\url{https://doi.org/10.48786/EDBT.2023.51}
\showDOI{\tempurl}


\bibitem[\protect\citeauthoryear{Chen, Chow, Davidson, DCunha, Ghodsi, Hong, Konwinski, Mewald, Murching, Nykodym, Ogilvie, Parkhe, Singh, Xie, Zaharia, Zang, Zheng, and Zumar}{Chen et~al\mbox{.}}{2020}]%
        {conf/sigmod/ChenCDD0HKMMNOP20}
\bibfield{author}{\bibinfo{person}{Andrew Chen}, \bibinfo{person}{Andy Chow}, \bibinfo{person}{Aaron Davidson}, \bibinfo{person}{Arjun DCunha}, \bibinfo{person}{Ali Ghodsi}, \bibinfo{person}{Sue~Ann Hong}, \bibinfo{person}{Andy Konwinski}, \bibinfo{person}{Clemens Mewald}, \bibinfo{person}{Siddharth Murching}, \bibinfo{person}{Tomas Nykodym}, \bibinfo{person}{Paul Ogilvie}, \bibinfo{person}{Mani Parkhe}, \bibinfo{person}{Avesh Singh}, \bibinfo{person}{Fen Xie}, \bibinfo{person}{Matei Zaharia}, \bibinfo{person}{Richard Zang}, \bibinfo{person}{Juntai Zheng}, {and} \bibinfo{person}{Corey Zumar}.} \bibinfo{year}{2020}\natexlab{}.
\newblock \showarticletitle{Developments in MLflow: {A} System to Accelerate the Machine Learning Lifecycle}. In \bibinfo{booktitle}{\emph{DEEM@SIGMOD'20}}.
\newblock


\bibitem[\protect\citeauthoryear{Chu, Murphy, Roesch, Cheung, and Suciu}{Chu et~al\mbox{.}}{2018}]%
        {journals/pvldb/ChuMRCS18}
\bibfield{author}{\bibinfo{person}{Shumo Chu}, \bibinfo{person}{Brendan Murphy}, \bibinfo{person}{Jared Roesch}, \bibinfo{person}{Alvin Cheung}, {and} \bibinfo{person}{Dan Suciu}.} \bibinfo{year}{2018}\natexlab{}.
\newblock \showarticletitle{Axiomatic Foundations and Algorithms for Deciding Semantic Equivalences of {SQL} Queries}.
\newblock \bibinfo{journal}{\emph{{VLDB}'18}} (\bibinfo{year}{2018}).
\newblock


\bibitem[\protect\citeauthoryear{Chu, Wang, Weitz, and Cheung}{Chu et~al\mbox{.}}{2017}]%
        {conf/cidr/ChuWWC17}
\bibfield{author}{\bibinfo{person}{Shumo Chu}, \bibinfo{person}{Chenglong Wang}, \bibinfo{person}{Konstantin Weitz}, {and} \bibinfo{person}{Alvin Cheung}.} \bibinfo{year}{2017}\natexlab{}.
\newblock \showarticletitle{Cosette: An Automated Prover for {SQL}}. In \bibinfo{booktitle}{\emph{8th Biennial Conference on Innovative Data Systems Research, {CIDR} 2017, Chaminade, CA, USA, January 8-11, 2017, Online Proceedings}}. \bibinfo{publisher}{www.cidrdb.org}.
\newblock
\urldef\tempurl%
\url{http://cidrdb.org/cidr2017/papers/p51-chu-cidr17.pdf}
\showURL{%
\tempurl}


\bibitem[\protect\citeauthoryear{databricks}{databricks}{[n.d.]}]%
        {databricks-tracking}

\newblock
\newblock
\shownote{Databricks Data Science Website, \url{https://www.databricks.com/product/data-science}.}


\bibitem[\protect\citeauthoryear{de~Moura and Bj{\o}rner}{de~Moura and Bj{\o}rner}{2008}]%
        {conf/tacas/MouraB08}
\bibfield{author}{\bibinfo{person}{Leonardo~Mendon{\c{c}}a de Moura} {and} \bibinfo{person}{Nikolaj~S. Bj{\o}rner}.} \bibinfo{year}{2008}\natexlab{}.
\newblock \showarticletitle{{Z3:} An Efficient {SMT} Solver}. In \bibinfo{booktitle}{\emph{{TACAS}'08}}.
\newblock


\bibitem[\protect\citeauthoryear{de~Moura, Kong, Avigad, van Doorn, and von Raumer}{de~Moura et~al\mbox{.}}{2015}]%
        {conf/cade/MouraKADR15}
\bibfield{author}{\bibinfo{person}{Leonardo~Mendon{\c{c}}a de Moura}, \bibinfo{person}{Soonho Kong}, \bibinfo{person}{Jeremy Avigad}, \bibinfo{person}{Floris van Doorn}, {and} \bibinfo{person}{Jakob von Raumer}.} \bibinfo{year}{2015}\natexlab{}.
\newblock \showarticletitle{The Lean Theorem Prover (System Description)}. In \bibinfo{booktitle}{\emph{Automated Deduction - {CADE-25} - 25th International Conference on Automated Deduction, Berlin, Germany, August 1-7, 2015, Proceedings}} \emph{(\bibinfo{series}{Lecture Notes in Computer Science})}, \bibfield{editor}{\bibinfo{person}{Amy~P. Felty} {and} \bibinfo{person}{Aart Middeldorp}} (Eds.), Vol.~\bibinfo{volume}{9195}. \bibinfo{publisher}{Springer}, \bibinfo{pages}{378--388}.
\newblock
\urldef\tempurl%
\url{https://doi.org/10.1007/978-3-319-21401-6\_26}
\showDOI{\tempurl}


\bibitem[\protect\citeauthoryear{Derakhshan, Mahdiraji, Kaoudi, Rabl, and Markl}{Derakhshan et~al\mbox{.}}{2022}]%
        {conf/sigmod/DerakhshanMKRM22}
\bibfield{author}{\bibinfo{person}{Behrouz Derakhshan}, \bibinfo{person}{Alireza~Rezaei Mahdiraji}, \bibinfo{person}{Zoi Kaoudi}, \bibinfo{person}{Tilmann Rabl}, {and} \bibinfo{person}{Volker Markl}.} \bibinfo{year}{2022}\natexlab{}.
\newblock \showarticletitle{Materialization and Reuse Optimizations for Production Data Science Pipelines}. In \bibinfo{booktitle}{\emph{{SIGMOD} '22: International Conference on Management of Data, Philadelphia, PA, USA, June 12 - 17, 2022}}. \bibinfo{publisher}{{ACM}}, \bibinfo{pages}{1962--1976}.
\newblock
\urldef\tempurl%
\url{https://doi.org/10.1145/3514221.3526186}
\showDOI{\tempurl}


\bibitem[\protect\citeauthoryear{Dursun, Binnig, {\c{C}}etintemel, and Kraska}{Dursun et~al\mbox{.}}{2017}]%
        {conf/sigmod/DursunBCK17}
\bibfield{author}{\bibinfo{person}{Kayhan Dursun}, \bibinfo{person}{Carsten Binnig}, \bibinfo{person}{Ugur {\c{C}}etintemel}, {and} \bibinfo{person}{Tim Kraska}.} \bibinfo{year}{2017}\natexlab{}.
\newblock \showarticletitle{Revisiting Reuse in Main Memory Database Systems}. In \bibinfo{booktitle}{\emph{{SIGMOD}'17}}.
\newblock


\bibitem[\protect\citeauthoryear{Elghandour and Aboulnaga}{Elghandour and Aboulnaga}{2012}]%
        {journals/pvldb/ElghandourA12}
\bibfield{author}{\bibinfo{person}{Iman Elghandour} {and} \bibinfo{person}{Ashraf Aboulnaga}.} \bibinfo{year}{2012}\natexlab{}.
\newblock \showarticletitle{ReStore: Reusing Results of MapReduce Jobs}.
\newblock \bibinfo{journal}{\emph{{VLDB}'12}} (\bibinfo{year}{2012}).
\newblock


\bibitem[\protect\citeauthoryear{Fagin, Kolaitis, Miller, and Popa}{Fagin et~al\mbox{.}}{2005}]%
        {journals/tcs/FaginKMP05}
\bibfield{author}{\bibinfo{person}{Ronald Fagin}, \bibinfo{person}{Phokion~G. Kolaitis}, \bibinfo{person}{Ren{\'{e}}e~J. Miller}, {and} \bibinfo{person}{Lucian Popa}.} \bibinfo{year}{2005}\natexlab{}.
\newblock \showarticletitle{Data exchange: semantics and query answering}.
\newblock \bibinfo{journal}{\emph{Theor. Comput. Sci.}} \bibinfo{volume}{336}, \bibinfo{number}{1} (\bibinfo{year}{2005}), \bibinfo{pages}{89--124}.
\newblock
\urldef\tempurl%
\url{https://doi.org/10.1016/j.tcs.2004.10.033}
\showDOI{\tempurl}


\bibitem[\protect\citeauthoryear{Gharibi, Walunj, Alanazi, Rella, and Lee}{Gharibi et~al\mbox{.}}{2019}]%
        {conf/sigmod/GharibiWARL19}
\bibfield{author}{\bibinfo{person}{Gharib Gharibi}, \bibinfo{person}{Vijay Walunj}, \bibinfo{person}{Rakan Alanazi}, \bibinfo{person}{Sirisha Rella}, {and} \bibinfo{person}{Yugyung Lee}.} \bibinfo{year}{2019}\natexlab{}.
\newblock \showarticletitle{Automated Management of Deep Learning Experiments}. In \bibinfo{booktitle}{\emph{DEEM@SIGMOD'19}}.
\newblock


\bibitem[\protect\citeauthoryear{Grossman, Cohen, Itzhaky, Rinetzky, and Sagiv}{Grossman et~al\mbox{.}}{2017}]%
        {conf/cav/GrossmanCIRS17}
\bibfield{author}{\bibinfo{person}{Shelly Grossman}, \bibinfo{person}{Sara Cohen}, \bibinfo{person}{Shachar Itzhaky}, \bibinfo{person}{Noam Rinetzky}, {and} \bibinfo{person}{Mooly Sagiv}.} \bibinfo{year}{2017}\natexlab{}.
\newblock \showarticletitle{Verifying Equivalence of Spark Programs}. In \bibinfo{booktitle}{\emph{{CAV}'17}}.
\newblock


\bibitem[\protect\citeauthoryear{Halevy}{Halevy}{2001}]%
        {journals/vldb/Halevy01}
\bibfield{author}{\bibinfo{person}{Alon~Y. Halevy}.} \bibinfo{year}{2001}\natexlab{}.
\newblock \showarticletitle{Answering Queries Using Views: A Survey}.
\newblock \bibinfo{journal}{\emph{The VLDB Journal}} \bibinfo{volume}{10}, \bibinfo{number}{4} (\bibinfo{date}{Dec.} \bibinfo{year}{2001}), \bibinfo{pages}{270--294}.
\newblock
\showISSN{1066-8888}
\urldef\tempurl%
\url{https://doi.org/10.1007/s007780100054}
\showDOI{\tempurl}


\bibitem[\protect\citeauthoryear{IMDB Datasets Website}{IMDB Datasets Website}{[n.d.]}]%
        {imdbdataset:website}
\bibinfo{title}{IMDB Datasets Website}.
\newblock
\newblock
\urldef\tempurl%
\url{https://www.imdb.com/interfaces/}
\showURL{%
\tempurl}


\bibitem[\protect\citeauthoryear{IMDB Workload Website}{IMDB Workload Website}{[n.d.]}]%
        {imdbload:website}
\bibinfo{title}{IMDB Workload Website}.
\newblock
\newblock
\urldef\tempurl%
\url{https://github.com/juanmanubens/SQL-Advanced-Queries/blob/master/imdb.sql}
\showURL{%
\tempurl}


\bibitem[\protect\citeauthoryear{Jayram, Kolaitis, and Vee}{Jayram et~al\mbox{.}}{2006}]%
        {conf/pods/JayramKV06}
\bibfield{author}{\bibinfo{person}{T.~S. Jayram}, \bibinfo{person}{Phokion~G. Kolaitis}, {and} \bibinfo{person}{Erik Vee}.} \bibinfo{year}{2006}\natexlab{}.
\newblock \showarticletitle{The containment problem for {REAL} conjunctive queries with inequalities}. In \bibinfo{booktitle}{\emph{Proceedings of the Twenty-Fifth {ACM} {SIGACT-SIGMOD-SIGART} Symposium on Principles of Database Systems, June 26-28, 2006, Chicago, Illinois, {USA}}}, \bibfield{editor}{\bibinfo{person}{Stijn Vansummeren}} (Ed.). \bibinfo{publisher}{{ACM}}, \bibinfo{pages}{80--89}.
\newblock
\urldef\tempurl%
\url{https://doi.org/10.1145/1142351.1142363}
\showDOI{\tempurl}


\bibitem[\protect\citeauthoryear{Jindal, Karanasos, Rao, and Patel}{Jindal et~al\mbox{.}}{2018}]%
        {journals/pvldb/JindalKRP18}
\bibfield{author}{\bibinfo{person}{Alekh Jindal}, \bibinfo{person}{Konstantinos Karanasos}, \bibinfo{person}{Sriram Rao}, {and} \bibinfo{person}{Hiren Patel}.} \bibinfo{year}{2018}\natexlab{}.
\newblock \showarticletitle{Selecting Subexpressions to Materialize at Datacenter Scale}.
\newblock \bibinfo{journal}{\emph{Proc. {VLDB} Endow.}} \bibinfo{volume}{11}, \bibinfo{number}{7} (\bibinfo{year}{2018}), \bibinfo{pages}{800--812}.
\newblock
\urldef\tempurl%
\url{https://doi.org/10.14778/3192965.3192971}
\showDOI{\tempurl}


\bibitem[\protect\citeauthoryear{{K}laus {G}reff, {A}aron {K}lein, {M}artin {C}hovanec, {F}rank {H}utter, and {J}\"urgen {S}chmidhuber}{{K}laus {G}reff et~al\mbox{.}}{2017}]%
        {klaus_greff-proc-scipy-2017}
\bibfield{author}{\bibinfo{person}{{K}laus {G}reff}, \bibinfo{person}{{A}aron {K}lein}, \bibinfo{person}{{M}artin {C}hovanec}, \bibinfo{person}{{F}rank {H}utter}, {and} \bibinfo{person}{{J}\"urgen {S}chmidhuber}.} \bibinfo{year}{2017}\natexlab{}.
\newblock \showarticletitle{{T}he {S}acred {I}nfrastructure for {C}omputational {R}esearch}. In \bibinfo{booktitle}{\emph{SciPy'17}}.
\newblock


\bibitem[\protect\citeauthoryear{Knime Workflows Website}{Knime Workflows Website}{[n.d.]}]%
        {knimeworkflows:website}
\bibinfo{title}{Knime Workflows Website}.
\newblock
\newblock
\urldef\tempurl%
\url{https://hub.knime.com/search?type=Workflow&sort=maxKudos}
\showURL{%
\tempurl}


\bibitem[\protect\citeauthoryear{Kossmann, Papenbrock, and Naumann}{Kossmann et~al\mbox{.}}{2022}]%
        {journals/vldb/KossmannPN22}
\bibfield{author}{\bibinfo{person}{Jan Kossmann}, \bibinfo{person}{Thorsten Papenbrock}, {and} \bibinfo{person}{Felix Naumann}.} \bibinfo{year}{2022}\natexlab{}.
\newblock \showarticletitle{Data dependencies for query optimization: a survey}.
\newblock \bibinfo{journal}{\emph{{VLDB} J.}} \bibinfo{volume}{31}, \bibinfo{number}{1} (\bibinfo{year}{2022}), \bibinfo{pages}{1--22}.
\newblock
\urldef\tempurl%
\url{https://doi.org/10.1007/s00778-021-00676-3}
\showDOI{\tempurl}


\bibitem[\protect\citeauthoryear{Kumar, Wang, Ni, and Li}{Kumar et~al\mbox{.}}{2020}]%
        {journals/pvldb/KumarWNL20}
\bibfield{author}{\bibinfo{person}{Avinash Kumar}, \bibinfo{person}{Zuozhi Wang}, \bibinfo{person}{Shengquan Ni}, {and} \bibinfo{person}{Chen Li}.} \bibinfo{year}{2020}\natexlab{}.
\newblock \showarticletitle{Amber: {A} Debuggable Dataflow System Based on the Actor Model}.
\newblock \bibinfo{journal}{\emph{Proc. {VLDB} Endow.}} \bibinfo{volume}{13}, \bibinfo{number}{5} (\bibinfo{year}{2020}), \bibinfo{pages}{740--753}.
\newblock
\urldef\tempurl%
\url{https://doi.org/10.14778/3377369.3377381}
\showDOI{\tempurl}


\bibitem[\protect\citeauthoryear{LeFevre, Sankaranarayanan, Hacig{\"{u}}m{\"{u}}s, Tatemura, Polyzotis, and Carey}{LeFevre et~al\mbox{.}}{2014}]%
        {conf/sigmod/LeFevreSHTPC14}
\bibfield{author}{\bibinfo{person}{Jeff LeFevre}, \bibinfo{person}{Jagan Sankaranarayanan}, \bibinfo{person}{Hakan Hacig{\"{u}}m{\"{u}}s}, \bibinfo{person}{Jun'ichi Tatemura}, \bibinfo{person}{Neoklis Polyzotis}, {and} \bibinfo{person}{Michael~J. Carey}.} \bibinfo{year}{2014}\natexlab{}.
\newblock \showarticletitle{Opportunistic physical design for big data analytics}. In \bibinfo{booktitle}{\emph{International Conference on Management of Data, {SIGMOD} 2014, Snowbird, UT, USA, June 22-27, 2014}}, \bibfield{editor}{\bibinfo{person}{Curtis~E. Dyreson}, \bibinfo{person}{Feifei Li}, {and} \bibinfo{person}{M.~Tamer {\"{O}}zsu}} (Eds.). \bibinfo{publisher}{{ACM}}, \bibinfo{pages}{851--862}.
\newblock
\urldef\tempurl%
\url{https://doi.org/10.1145/2588555.2610512}
\showDOI{\tempurl}


\bibitem[\protect\citeauthoryear{Miao and Deshpande}{Miao and Deshpande}{2018}]%
        {journals/debu/0001D18}
\bibfield{author}{\bibinfo{person}{Hui Miao} {and} \bibinfo{person}{Amol Deshpande}.} \bibinfo{year}{2018}\natexlab{}.
\newblock \showarticletitle{ProvDB: Provenance-enabled Lifecycle Management of Collaborative Data Analysis Workflows}.
\newblock \bibinfo{journal}{\emph{{IEEE} Data Eng. Bull.}} (\bibinfo{year}{2018}).
\newblock


\bibitem[\protect\citeauthoryear{Miao, Li, Davis, and Deshpande}{Miao et~al\mbox{.}}{2017}]%
        {conf/icde/MiaoLDD17}
\bibfield{author}{\bibinfo{person}{Hui Miao}, \bibinfo{person}{Ang Li}, \bibinfo{person}{Larry~S. Davis}, {and} \bibinfo{person}{Amol Deshpande}.} \bibinfo{year}{2017}\natexlab{}.
\newblock \showarticletitle{Towards Unified Data and Lifecycle Management for Deep Learning}. In \bibinfo{booktitle}{\emph{ICDE'17}}.
\newblock


\bibitem[\protect\citeauthoryear{Nagel, Boncz, and Viglas}{Nagel et~al\mbox{.}}{2013}]%
        {conf/icde/NagelBV13}
\bibfield{author}{\bibinfo{person}{Fabian Nagel}, \bibinfo{person}{Peter~A. Boncz}, {and} \bibinfo{person}{Stratis Viglas}.} \bibinfo{year}{2013}\natexlab{}.
\newblock \showarticletitle{Recycling in pipelined query evaluation}. In \bibinfo{booktitle}{\emph{{ICDE}'13}}.
\newblock


\bibitem[\protect\citeauthoryear{{{NYC Taxi Data}}}{{{NYC Taxi Data}}}{[n.d.]}]%
        {nyc-taxi-data:website}

\newblock
\newblock
\shownote{\url{https://www1.nyc.gov/site/tlc/about/tlc-trip-record-data.page}.}


\bibitem[\protect\citeauthoryear{Optimizing Apache Spark UDFs Website}{Optimizing Apache Spark UDFs Website}{[n.d.]}]%
        {sparkudf:website}
\bibinfo{title}{Optimizing Apache Spark UDFs Website}.
\newblock
\newblock
\urldef\tempurl%
\url{https://www.databricks.com/session_eu20/optimizing-apache-spark-udfs}
\showURL{%
\tempurl}


\bibitem[\protect\citeauthoryear{Orange Data Mining Workflows}{Orange Data Mining Workflows}{[n.d.]}]%
        {orangeworkflows:website}
\bibinfo{title}{Orange Data Mining Workflows}.
\newblock
\newblock
\urldef\tempurl%
\url{https://orangedatamining.com/workflows/}
\showURL{%
\tempurl}


\bibitem[\protect\citeauthoryear{Perez and Jermaine}{Perez and Jermaine}{2014}]%
        {conf/icde/PerezJ14}
\bibfield{author}{\bibinfo{person}{Luis~Leopoldo Perez} {and} \bibinfo{person}{Christopher~M. Jermaine}.} \bibinfo{year}{2014}\natexlab{}.
\newblock \showarticletitle{History-aware query optimization with materialized intermediate views}. In \bibinfo{booktitle}{\emph{{IEEE} 30th International Conference on Data Engineering, Chicago, {ICDE} 2014, IL, USA, March 31 - April 4, 2014}}, \bibfield{editor}{\bibinfo{person}{Isabel~F. Cruz}, \bibinfo{person}{Elena Ferrari}, \bibinfo{person}{Yufei Tao}, \bibinfo{person}{Elisa Bertino}, {and} \bibinfo{person}{Goce Trajcevski}} (Eds.). \bibinfo{publisher}{{IEEE} Computer Society}, \bibinfo{pages}{520--531}.
\newblock
\urldef\tempurl%
\url{https://doi.org/10.1109/ICDE.2014.6816678}
\showDOI{\tempurl}


\bibitem[\protect\citeauthoryear{Pimentel, Murta, Braganholo, and Freire}{Pimentel et~al\mbox{.}}{2017}]%
        {journals/pvldb/PimentelMBF17}
\bibfield{author}{\bibinfo{person}{Jo{\~{a}}o~Felipe Pimentel}, \bibinfo{person}{Leonardo Murta}, \bibinfo{person}{Vanessa Braganholo}, {and} \bibinfo{person}{Juliana Freire}.} \bibinfo{year}{2017}\natexlab{}.
\newblock \showarticletitle{noWorkflow: a Tool for Collecting, Analyzing, and Managing Provenance from Python Scripts}.
\newblock \bibinfo{journal}{\emph{VLDB}} (\bibinfo{year}{2017}).
\newblock


\bibitem[\protect\citeauthoryear{Ramjit, Interlandi, Wu, and Netravali}{Ramjit et~al\mbox{.}}{2019}]%
        {conf/cloud/RamjitIWN19}
\bibfield{author}{\bibinfo{person}{Lana Ramjit}, \bibinfo{person}{Matteo Interlandi}, \bibinfo{person}{Eugene Wu}, {and} \bibinfo{person}{Ravi Netravali}.} \bibinfo{year}{2019}\natexlab{}.
\newblock \showarticletitle{Acorn: Aggressive Result Caching in Distributed Data Processing Frameworks}. In \bibinfo{booktitle}{\emph{Proceedings of the {ACM} Symposium on Cloud Computing, SoCC 2019, Santa Cruz, CA, USA, November 20-23, 2019}}. \bibinfo{publisher}{{ACM}}, \bibinfo{pages}{206--219}.
\newblock
\urldef\tempurl%
\url{https://doi.org/10.1145/3357223.3362702}
\showDOI{\tempurl}


\bibitem[\protect\citeauthoryear{Riesen, Emmenegger, and Bunke}{Riesen et~al\mbox{.}}{2013}]%
        {conf/gbrpr/RiesenEB13}
\bibfield{author}{\bibinfo{person}{Kaspar Riesen}, \bibinfo{person}{Sandro Emmenegger}, {and} \bibinfo{person}{Horst Bunke}.} \bibinfo{year}{2013}\natexlab{}.
\newblock \showarticletitle{A Novel Software Toolkit for Graph Edit Distance Computation}. In \bibinfo{booktitle}{\emph{Graph-Based Representations in Pattern Recognition - 9th {IAPR-TC-15} International Workshop, GbRPR 2013, Vienna, Austria, May 15-17, 2013. Proceedings}} \emph{(\bibinfo{series}{Lecture Notes in Computer Science})}, \bibfield{editor}{\bibinfo{person}{Walter~G. Kropatsch}, \bibinfo{person}{Nicole~M. Artner}, \bibinfo{person}{Yll Haxhimusa}, {and} \bibinfo{person}{Xiaoyi Jiang}} (Eds.), Vol.~\bibinfo{volume}{7877}. \bibinfo{publisher}{Springer}, \bibinfo{pages}{142--151}.
\newblock
\urldef\tempurl%
\url{https://doi.org/10.1007/978-3-642-38221-5\_15}
\showDOI{\tempurl}


\bibitem[\protect\citeauthoryear{Roy, Jindal, Gomatam, Ouyang, Gosalia, Ravi, Mann, and Jain}{Roy et~al\mbox{.}}{2021}]%
        {journals/pvldb/RoyJGOGRMJ21}
\bibfield{author}{\bibinfo{person}{Abhishek Roy}, \bibinfo{person}{Alekh Jindal}, \bibinfo{person}{Priyanka Gomatam}, \bibinfo{person}{Xiating Ouyang}, \bibinfo{person}{Ashit Gosalia}, \bibinfo{person}{Nishkam Ravi}, \bibinfo{person}{Swinky Mann}, {and} \bibinfo{person}{Prakhar Jain}.} \bibinfo{year}{2021}\natexlab{}.
\newblock \showarticletitle{SparkCruise: Workload Optimization in Managed Spark Clusters at Microsoft}.
\newblock \bibinfo{journal}{\emph{Proc. {VLDB} Endow.}} \bibinfo{volume}{14}, \bibinfo{number}{12} (\bibinfo{year}{2021}), \bibinfo{pages}{3122--3134}.
\newblock
\urldef\tempurl%
\url{https://doi.org/10.14778/3476311.3476388}
\showDOI{\tempurl}


\bibitem[\protect\citeauthoryear{Sagiv and Yannakakis}{Sagiv and Yannakakis}{1980}]%
        {journals/jacm/SagivY80}
\bibfield{author}{\bibinfo{person}{Yehoshua Sagiv} {and} \bibinfo{person}{Mihalis Yannakakis}.} \bibinfo{year}{1980}\natexlab{}.
\newblock \showarticletitle{Equivalences Among Relational Expressions with the Union and Difference Operators}.
\newblock \bibinfo{journal}{\emph{J. {ACM}}} \bibinfo{volume}{27}, \bibinfo{number}{4} (\bibinfo{year}{1980}), \bibinfo{pages}{633--655}.
\newblock
\urldef\tempurl%
\url{https://doi.org/10.1145/322217.322221}
\showDOI{\tempurl}


\bibitem[\protect\citeauthoryear{Schmidt, Meier, and Lausen}{Schmidt et~al\mbox{.}}{2010}]%
        {conf/icdt/Schmidt0L10}
\bibfield{author}{\bibinfo{person}{Michael Schmidt}, \bibinfo{person}{Michael Meier}, {and} \bibinfo{person}{Georg Lausen}.} \bibinfo{year}{2010}\natexlab{}.
\newblock \showarticletitle{Foundations of {SPARQL} query optimization}. In \bibinfo{booktitle}{\emph{Database Theory - {ICDT} 2010, 13th International Conference, Lausanne, Switzerland, March 23-25, 2010, Proceedings}} \emph{(\bibinfo{series}{{ACM} International Conference Proceeding Series})}, \bibfield{editor}{\bibinfo{person}{Luc Segoufin}} (Ed.). \bibinfo{publisher}{{ACM}}, \bibinfo{pages}{4--33}.
\newblock
\urldef\tempurl%
\url{https://doi.org/10.1145/1804669.1804675}
\showDOI{\tempurl}


\bibitem[\protect\citeauthoryear{Texera}{Texera}{[n.d.]}]%
        {texera}

\newblock
\newblock
\shownote{Texera Website, \url{https://github.com/Texera/texera}.}


\bibitem[\protect\citeauthoryear{TPC-DS}{TPC-DS}{[n.d.]}]%
        {misc/tpcds}

\newblock
\newblock
\shownote{TPC-DS http://www.tpc.org/tpcds/.}


\bibitem[\protect\citeauthoryear{Twitter API v1.1}{Twitter API v1.1}{[n.d.]}]%
        {twitter-api}
\bibinfo{title}{Twitter API v1.1}.
\newblock
\newblock
\urldef\tempurl%
\url{https://developer.twitter.com/en/docs/twitter-api/v1/tweets/filter-realtime/overview}
\showURL{%
\tempurl}


\bibitem[\protect\citeauthoryear{Vartak, Subramanyam, Lee, Viswanathan, Husnoo, Madden, and Zaharia}{Vartak et~al\mbox{.}}{2016}]%
        {conf/sigmod/VartakSLVHMZ16}
\bibfield{author}{\bibinfo{person}{Manasi Vartak}, \bibinfo{person}{Harihar Subramanyam}, \bibinfo{person}{Wei{-}En Lee}, \bibinfo{person}{Srinidhi Viswanathan}, \bibinfo{person}{Saadiyah Husnoo}, \bibinfo{person}{Samuel Madden}, {and} \bibinfo{person}{Matei Zaharia}.} \bibinfo{year}{2016}\natexlab{}.
\newblock \showarticletitle{ModelDB: a system for machine learning model management}. In \bibinfo{booktitle}{\emph{HILDA@SIGMOD'16}}.
\newblock


\bibitem[\protect\citeauthoryear{Wang, Zhou, Yang, Ding, Hu, Ding, Tang, Chen, and Li}{Wang et~al\mbox{.}}{2022}]%
        {conf/sigmod/WangZYDHDT0022}
\bibfield{author}{\bibinfo{person}{Zhaoguo Wang}, \bibinfo{person}{Zhou Zhou}, \bibinfo{person}{Yicun Yang}, \bibinfo{person}{Haoran Ding}, \bibinfo{person}{Gansen Hu}, \bibinfo{person}{Ding Ding}, \bibinfo{person}{Chuzhe Tang}, \bibinfo{person}{Haibo Chen}, {and} \bibinfo{person}{Jinyang Li}.} \bibinfo{year}{2022}\natexlab{}.
\newblock \showarticletitle{WeTune: Automatic Discovery and Verification of Query Rewrite Rules}. In \bibinfo{booktitle}{\emph{{SIGMOD} '22: International Conference on Management of Data, Philadelphia, PA, USA, June 12 - 17, 2022}}. \bibinfo{publisher}{{ACM}}, \bibinfo{pages}{94--107}.
\newblock
\urldef\tempurl%
\url{https://doi.org/10.1145/3514221.3526125}
\showDOI{\tempurl}


\bibitem[\protect\citeauthoryear{Woodman, Hiden, Watson, and Missier}{Woodman et~al\mbox{.}}{2011}]%
        {conf/sc/WoodmanHWM11}
\bibfield{author}{\bibinfo{person}{Simon Woodman}, \bibinfo{person}{Hugo Hiden}, \bibinfo{person}{Paul Watson}, {and} \bibinfo{person}{Paolo Missier}.} \bibinfo{year}{2011}\natexlab{}.
\newblock \showarticletitle{Achieving reproducibility by combining provenance with service and workflow versioning}. In \bibinfo{booktitle}{\emph{WORKS'11}}.
\newblock


\bibitem[\protect\citeauthoryear{Xu, Kakkar, Arulraj, and Ramachandran}{Xu et~al\mbox{.}}{2022}]%
        {conf/sigmod/XuKAR22}
\bibfield{author}{\bibinfo{person}{Zhuangdi Xu}, \bibinfo{person}{Gaurav~Tarlok Kakkar}, \bibinfo{person}{Joy Arulraj}, {and} \bibinfo{person}{Umakishore Ramachandran}.} \bibinfo{year}{2022}\natexlab{}.
\newblock \showarticletitle{{EVA:} {A} Symbolic Approach to Accelerating Exploratory Video Analytics with Materialized Views}. In \bibinfo{booktitle}{\emph{{SIGMOD} '22: International Conference on Management of Data, Philadelphia, PA, USA, June 12 - 17, 2022}}, \bibfield{editor}{\bibinfo{person}{Zachary Ives}, \bibinfo{person}{Angela Bonifati}, {and} \bibinfo{person}{Amr~El Abbadi}} (Eds.). \bibinfo{publisher}{{ACM}}, \bibinfo{pages}{602--616}.
\newblock
\urldef\tempurl%
\url{https://doi.org/10.1145/3514221.3526142}
\showDOI{\tempurl}


\bibitem[\protect\citeauthoryear{Zhang, Xu, Frise, Wu, Yu, and Xu}{Zhang et~al\mbox{.}}{2016}]%
        {conf/icse/ZhangXFWYX16}
\bibfield{author}{\bibinfo{person}{Yang Zhang}, \bibinfo{person}{Fangzhou Xu}, \bibinfo{person}{Erwin Frise}, \bibinfo{person}{Siqi Wu}, \bibinfo{person}{Bin Yu}, {and} \bibinfo{person}{Wei Xu}.} \bibinfo{year}{2016}\natexlab{}.
\newblock \showarticletitle{DataLab: a version data management and analytics system}. In \bibinfo{booktitle}{\emph{BIGDSE@ICSE'16}}.
\newblock


\bibitem[\protect\citeauthoryear{Zhou, Larson, Freytag, and Lehner}{Zhou et~al\mbox{.}}{2007}]%
        {conf/sigmod/ZhouLFL07}
\bibfield{author}{\bibinfo{person}{Jingren Zhou}, \bibinfo{person}{Per{-}{\AA}ke Larson}, \bibinfo{person}{Johann~Christoph Freytag}, {and} \bibinfo{person}{Wolfgang Lehner}.} \bibinfo{year}{2007}\natexlab{}.
\newblock \showarticletitle{Efficient exploitation of similar subexpressions for query processing}. In \bibinfo{booktitle}{\emph{{SIGMOD}'07}}.
\newblock


\bibitem[\protect\citeauthoryear{Zhou, Arulraj, Navathe, Harris, and Wu}{Zhou et~al\mbox{.}}{2022}]%
        {journals/corr/abs-2004-00481}
\bibfield{author}{\bibinfo{person}{Qi Zhou}, \bibinfo{person}{Joy Arulraj}, \bibinfo{person}{Shamkant~B. Navathe}, \bibinfo{person}{William Harris}, {and} \bibinfo{person}{Jinpeng Wu}.} \bibinfo{year}{2022}\natexlab{}.
\newblock \showarticletitle{{SPES:} {A} Symbolic Approach to Proving Query Equivalence Under Bag Semantics}.
\newblock  (\bibinfo{year}{2022}), \bibinfo{pages}{2735--2748}.
\newblock
\urldef\tempurl%
\url{https://doi.org/10.1109/ICDE53745.2022.00250}
\showDOI{\tempurl}


\bibitem[\protect\citeauthoryear{Zhou, Arulraj, Navathe, Harris, and Xu}{Zhou et~al\mbox{.}}{2019}]%
        {journals/pvldb/ZhouANHX19}
\bibfield{author}{\bibinfo{person}{Qi Zhou}, \bibinfo{person}{Joy Arulraj}, \bibinfo{person}{Shamkant~B. Navathe}, \bibinfo{person}{William Harris}, {and} \bibinfo{person}{Dong Xu}.} \bibinfo{year}{2019}\natexlab{}.
\newblock \showarticletitle{Automated Verification of Query Equivalence Using Satisfiability Modulo Theories}.
\newblock \bibinfo{journal}{\emph{{VLDB}'19}} (\bibinfo{year}{2019}).
\newblock


\end{thebibliography}
